%% file: arXiv_V5.tex
\documentclass[journal]{IEEEtran}

\usepackage[square,numbers,sort&compress]{natbib}
\usepackage{hyperref} % 
\usepackage% 
{hyperref} % 
\hypersetup{
    colorlinks,
    linkcolor={blue!80!black},% 
    citecolor={green!50!black},
    urlcolor={blue!80!black}
}

\usepackage{pgfplots}
\pgfplotsset{compat=1.18}

\usepackage[normalem]{ulem}
\usepackage[makeroom]{cancel}
\usepackage{soul}

\usepackage{amsthm}
\usepackage{amsmath}
\usepackage{mathrsfs}
\usepackage{amssymb} 
\usepackage{stackrel}
\usepackage{mathtools}

\usepackage{physics}

\usepackage{balance}
\usepackage{verbatim}
\usepackage{enumerate}

\usepackage[T1]{fontenc}
\usepackage{dsfont}
\usepackage{bbding}
\usepackage{keystroke}
\usepackage{pifont}
\usepackage{relsize} % 

\usepackage[normalem]{ulem}

\usepackage{graphicx}
\usepackage{xcolor}
\usepackage{tikz}
\usetikzlibrary{calc}
\usepackage{tabularx}
\usetikzlibrary{arrows,shapes}

\usetikzlibrary{matrix,positioning}
\usetikzlibrary{arrows}
\usetikzlibrary{fit,backgrounds}
\usetikzlibrary{positioning}

\newcommand{\TrNorm}[1]{\left\lVert #1\right\rVert_1}

\newcommand{\RelEntropy}[2]{D(#1||#2)}

\renewcommand{\trace}{\mathrm{Tr}}

\newcommand{\identity}{\mathds{1}}

\theoremstyle{remark}	\newtheorem{theorem}{Theorem}
\theoremstyle{remark}	\newtheorem{lemma}[theorem]{Lemma}
\theoremstyle{remark}	
\theoremstyle{remark}	\newtheorem{proposition}[theorem]{Proposition}
\theoremstyle{remark} \newtheorem{definition}{Definition}
\theoremstyle{remark} \newtheorem{remark}{Remark}
\theoremstyle{remark}

\begin{document}

\title{Covert Entanglement Generation and Secrecy   % 
\thanks{% 
This work was supported in part by the National Science Foundation
under Grant CCF-2006679 and Grant CNS-2107265, in part
by Israel Science Foundation under Grant 939/23 and Grant 2691/23,
in part by German-Israeli Project Cooperation (DIP) within the Deutsche
Forschungsgemeinschaft (DFG) under Grant 2032991, in part by the Ollen\-dorff
Minerva Center (OMC) of the Technion under Grant 86160946,
and in part by the QERNEL Quantum Computing Theory Research Hub of
Israel Planning and Budgeting Committee of the Council for Higher Education
(VATAT) under Grant 2072651. An earlier version of this paper
was presented in part at the 2025 13th Beyond IID in Information Theory
workshop and in part at the 2025 IEEE Information Theory Workshop (ITW).
}
\thanks{% 
Ohad Kimelfeld is with the Faculty of Physics and the Helen Diller
Quantum Center, Technion---Israel Institute of Technology, Haifa 3200003,
Israel (e-mail: ohad.kim@campus.technion.ac.il).% 
}
\thanks{% 
Boulat A. Bash is with the Department of Electrical and Computer
Engineering and the Wyant College of Optical Sciences, The University of
Arizona, Tucson, AZ 85721 USA (e-mail: boulat@arizona.edu).% 
}
\thanks{% 
Uzi Pereg is with the Faculty of Electrical and Computer Engineering and
the Helen Diller Quantum Center, Technion---Israel Institute of Technology,
Haifa 3200003, Israel (e-mail: uzipereg@technion.ac.il).% 
}
}

\author{
    \IEEEauthorblockN{Ohad Kimelfeld, Boulat A. Bash and
    Uzi Pereg} % 
} 

\maketitle

\begin{abstract}
 We determine the covert capacity for entanglement generation over a noisy quantum channel. 
 While secrecy guarantees that the transmitted information remains inaccessible to an adversary,
 covert communication ensures that the transmission itself remains undetectable.
The entanglement dimension follows
a square root law (SRL) in the covert setting, i.e., % 
$O(\sqrt{n})$ Einstein-Podolsky-Rosen (EPR) pairs can be distributed covertly and reliably over $n$ channel uses. % 
We begin with covert communication of classical information under a % 
secrecy constraint. We then leverage this result to construct a coding scheme for covert entanglement generation. % 
Single-letter expressions are derived for the covert key-assisted and unassisted secrecy capacities, as well as for the covert entanglement-generation capacity.

\end{abstract}

\section{Introduction}
\label{section_introduction}

\IEEEPARstart{P}{rivacy} is a fundamental aspect of communication systems 
\cite{menezes96HAC,talb2006,boche2014secrecy,guha2014quantum,8076713,pereg2021key,9319007,10461354,10206763,seitz2024private,munar2024joint}.
Traditional security approaches --- such as data encryption, information-theoretic secrecy, and quantum key distribution --- are designed to prevent an eavesdropper from recovering the transmitted information or part of it \cite{BlochGunluYenerOggierPoorSankarSchaefer:21p}.
Covert communication prevents the detection of transmitted signals by masking them in noise. Although covertness can strengthen security, it comes at the cost of the \emph{square root law} (SRL), allowing reliable and covert transmission of only $O(\sqrt{n})$ bits in $n$ channel uses \cite{bash12sqrtlawisit,BashGT13,bash2015quantum,WangWZ15,Bloch16,TahmasbiBloch:19p,tahmasbi2020covertIEEE,BullockGGB20,tahmasbi2020steganography,zivarifard2025covert, SheikholeslamiB16,10886999}. 
A tutorial introduction to covert communication \cite{bash15covertcommmag} and a recent survey \cite{CAXXZYN:23p} are available.
While secrecy and covertness can be seen as two orthogonal approaches to secure communication, 
their combined setting 
is considered for 
classical channels in \cite[Sec.~VII-C]{Bloch16}.
While covert communication focuses on hiding the existence of the transmission from an observing warden, physical-layer security (PLS) guarantees that, even if the signal is intercepted, the private information remains protected \cite{nguyen2026covert}.

Recently, there has been a growing interest in how pre-shared entanglement resources can boost covert communication performance. The concept behind entanglement-assisted communication is to utilize inactive periods to generate shared entanglement, which can then enhance the throughput once transmission resumes \cite{BennettShorSmolinThapliyal:02p,NoetzelDiAdamo:20c,PeregDeppeBoche:23p}.
Gagatsos et al.~\cite{GagatsosBB20} showed that, in continuous-variable communication, entanglement assistance enables the transmission of information on the order of $O(\sqrt{n}\log{n})$ information bits. 
Wang et al.~\cite{wang2024resource} improved their result, showing that the benefits of entanglement can be achieved with fewer entanglement resources than previously established. 
Zlotnick et al.~\cite{ZlotnickBP23} showed that $O(\sqrt{n}\log{n})$ information bits can be sent with entanglement assistance over finite-dimensional quantum channels, more specifically, qubit depolarizing channels. This highlights the significance of covert entanglement generation. 

 \begin{figure*}[bt!]
\center
    \input{Covert_EG_No_Key_Figure}
    \caption{
    Covert entanglement generation. Alice prepares a maximally entangled state
$\ket{\Phi}_{RM}$,
locally, where $R$ is a resource that she keeps, and $M$ is the resource that she would like to distribute to Bob. Alice makes a decision on whether to perform the communication task, or not. 
If Alice decides to be inactive, the channel input is $\ket{0}^{\otimes n}$.
Otherwise, she applies an encoding map to prepare the channel input $A^n$.
She then transmits  $A^n$ via the
quantum channel $\mathcal{U}_{A\rightarrow BW}^{\otimes n}$. Bob receives $B^n$, % 
performs a decoding operation and prepares  ${\widehat{M}}$  such that Alice and Bob's state is $\approx \ket{\Phi}_{R\widehat{M}}$. Willie would like to detect the transmission. To this end, he performs a hypothesis-test measurement  on his output $W^n$ to estimate whether 
Alice is quiet (null hypothesis $H_0$) or 
transmitting (alternate hypothesis $H_1$).
    \label{Figure:EG_Model}}
\end{figure*}

Entanglement generation is closely related to quantum subspace transmission \cite{lloyd1997capacity,shor2002quantum,Devetak05}, i.e., sending quantum information. In particular, the ability to teleport a qubit state implies the ability to generate an  Einstein-Podolsky-Rosen (EPR) pair between Alice and Bob, i.e., a pair of entangled qubits.
Anderson et al.~\cite{anderson2024covert,anderson2025achievability} have recently developed lower bounds on the quantum covert rate using twirl modulation and depolarizing channel codes. Here, we give a more refined characterization and determine the capacity for covert entanglement generation in terms of the channel itself. 

Furthermore,
entanglement generation is intimately related to secrecy \cite{Watanabe:12p, singh2025capacities}. Due to the no-cloning theorem, the transmission of quantum information inherently ensures  secrecy
\cite{W2017}.  If the eavesdropper, Eve, could obtain any information about the quantum information that Alice is sending to Bob, then Bob would not be able to recover it. Otherwise, this would contradict the no-cloning theorem. 
Devetak \cite{Devetak05} introduced a coherent version of classical secrecy codes, which leverage their privacy properties to define subspaces where Alice can securely encode quantum information, ensuring Eve's inaccessibility. 
Another technique for entanglement generation is based on the decoupling approach \cite{hayden2008decoupling}, which leverages the principle that quantum information transmission is possible if the environment of the channel becomes ``decoupled'' from the transmitted  state \cite{dupuis2014one}.
In this sense, secrecy is both necessary and sufficient in order to establish entanglement generation.

\setcounter{page}{1}

Consider a covert entanglement generation setting described in  Figure~\ref{Figure:EG_Model}, where the adversary, per convention in covert communication literature, is named Warden Willie due to his task of detecting Alice's transmission rather than decoding it being fundamentally different from Eve's.\footnote{Warden Willie moniker is, in turn, borrowed from steganography literature \cite{fridrich09stego}.}
Alice makes a decision on whether to perform the communication task, or not. 
If Alice decides to be inactive, the channel input is $\ket{0}^{\otimes n}$, where $\ket{0}$ is the ``innocent'' state corresponding to a passive transmitter.
Otherwise, if she does perform the task,  she prepares a maximally entangled state 
$\Phi_{RM}$ locally, and
applies an encoding map  
on her ``quantum message'' $M$.
She then transmits the encoded system $A^n=(A_1,A_2,\ldots,A_n)$ using $n$ instances of the
quantum channel. 
At the channel output, Bob and Willie receive 
$B^n=(B_1,B_2,\ldots,B_n)$ and
$W^n=(W_1,W_2,\ldots,W_n)$,
respectively. 
Bob % 
performs a decoding operation on his received system, which recovers a state that is close to $% 
\Phi_{R\widehat{M}}$.
 Meanwhile, Willie receives $W^n$ and 
 performs a hypothesis test to determine whether Alice has transmitted information or not.

% 

\input{Table_I_General}

Our approach is fundamentally different from that in Anderson et al.~\cite{anderson2024covert,anderson2025achievability}.
First, we consider the combined setting of covert and secret communication of classical information via a classical-quantum (c-q) channel, and determine the covert secrecy capacity. This can be viewed as the classical-quantum generalization of the result by Bloch \cite[Sec. VII-C]{Bloch16}.  One might argue that, if covertness is achieved, secrecy becomes redundant, as Willie would not attempt to decode a message he does not detect. However, covertness is typically defined in a statistical sense, meaning that while the probability of detection is small, it is not necessarily zero. Thus, in those rare cases when Willie does detect some anomalous activity, secrecy ensures that he still cannot extract meaningful information. Covert secrecy is thus a problem of independent interest, not merely an auxiliary result for the main derivation.
We determine the covert secrecy capacity,
both with and without key assistance,
as well as the minimal key rate required to achieve the key-assisted capacity.
Then, we use Devetak's approach \cite[Sec.~IV]{Devetak05} of
constructing an entanglement-generation code from a secrecy code.
This method utilizes secrecy to establish entanglement generation. 
We note that, as opposed to Anderson et al.~\cite{anderson2025achievability},
our scheme does not require a pre-shared secret key.

 We show that approximately $ {\sqrt{n} C_\text{EG}}$ EPR pairs can be generated covertly. The optimal rate $ C_\text{EG}$, which is referred to as the covert capacity for entanglement generation, is
 \begin{align}
C_\text{EG}=\frac{{ \left[\RelEntropy{\sigma_1}{\sigma_0}
    -\RelEntropy{\omega_1}{\omega_0}
    \right]_+}}{\sqrt{\frac{1}{2}\chi^2(\omega_1||\omega_0)}},
\end{align}
where 
$[t]_+ \equiv \max(0,t)$ for every 
$t\in\mathbb{R}$;
$\sigma_0$ and $\omega_0$ are Bob's and Willie's respective outputs for the ``innocent'' input $\ket{0}$, whereas $\sigma_1$ and $\omega_1$ are the outputs associated with  inputs that are orthogonal to $\ket{0}$; 
$D(\rho||\sigma)$ is the quantum relative entropy and 
$\chi^2(\rho||\sigma)$ is the $\chi^2$-relative entropy defined in \eqref{eq:eta}, which can be interpreted as the second derivative of the quantum relative entropy 
(see 
\cite[Eq.~4]{9344627},
 \cite[Sec.~1.1.4]{Zlotnick:24z}).
 Furthermore, we demonstrate our results through the example of the excitation channel, and
  derive a closed-form formula for its covert entanglement-generation capacity. 
 In  many settings in quantum Shannon theory, a single-letter capacity formula is an open problem \cite{Devetak05}
 (see discussion on the importance of single-letterization in \cite{Pereg2023pCommunication}). 
Remarkably, we establish a single-letter formula for this fully quantum model.

The paper is organized as follows.
 Section~\ref{section_prerequisites} provides basic definitions. 
 In Section~\ref{Section:Covert-Secret Communication Over Classical-Quantum Channels}, we address covert communication of classical information under a secrecy constraint, both with and without key assistance. 
 In Section~\ref{Section:Entanglement Generation}, we present the model definitions and capacity result for covert entanglement generation. 
In Section \ref{sec:Proof Covert Entanglement Generation}, we prove the covert entanglement-generation capacity theorem.
 Section~\ref{Section:Discussion} concludes with a summary and discussion. 

The analysis of covert secrecy, along with technical lemmas for covert entanglement generation, is provided in the Supplementary Material. Sections~\ref{Section:Secrecy_Proof} and \ref{Section:Minimal_Key}  present the derivation of the covert secrecy capacity with key assistance and the corresponding minimum key rate required to achieve this capacity. Section~\ref{Appendix:Covert_Secrecy_Capacity_No_Key_Proof} establishes the capacity result for covert secrecy without assistance, and Section~\ref{Appendix:Quantum_Code_First_Approximation} contains the proof of a lemma used in the direct part of the covert entanglement-generation result.

\begin{figure*}[bt!]
\center
    \input{Secrecy_Figure_With_Public_Message}
    \caption{
   Covert secrecy for a classical-quantum channel.
    Suppose  % 
    Alice selects a classical {secret} message $m$ {and a classical public message $\ell$}.
 She makes a decision on whether to send % 
 {them} 
 to Bob, or  be inactive, in which case % 
 the channel input is $x^n=0^n$.
Otherwise, she encodes her message {pair $(m,\ell)$}  using her access to the pre-shared secret key and transmits a codeword
 {$x^n=f{(m,\ell,k)}$} via  $\mathcal{P}_{X\rightarrow BW}^{\otimes n}$. At the channel output, Bob uses the key and performs a decoding measurement  on his received system $B^n$, and obtains an estimate {$(\hat{m},\hat{\ell})$}. Willie attempts to detect and decode Alice's transmission and {secret} message by measuring his received system $W^n$.
}
    \label{Figure:Secrecy_Model}
\end{figure*}

\input{Table_II_Secrecy}

\section{Basic Definitions}
\label{section_prerequisites}
We use the  notation conventions in Table~\ref{tab:general-notation}. 
Calligraphic letters  $\mathcal{X}, \mathcal{Y}, \mathcal{Z}, \ldots$ denote finite sets.
Bold lowercase letters $\mathbf{x}, \mathbf{y}, \mathbf{z}, \ldots$ represent random variables,
while the
non-bold lowercase letters $x, y, z, \ldots$ stand for their 
values. We use $x^j = (x_1, x_2, \ldots, x_j)$ for a sequence of letters from the alphabet $\mathcal{X}$, and $[i:j]$ denotes the index set $\{i, i+1, \ldots, j\}$ where $j>i$. 
For every 
$t\in\mathbb{R}$, let $[t]_+$ denote the positive part of $t$, namely, $[t]_+ \equiv \max(0,t)$.
We use  standard asymptotic notation \cite[Ch. 3.1]{cormen2009introduction} for functions $g:\mathbb{N}\to \mathbb{R}$, % 
\begin{align}
\begin{array}{ll}
    O(g(n)) 
    &\equiv 
    \left\{f(n): \limsup\limits_{n 
    \to \infty}{\abs{\frac{f(n)}{g(n)}}} < \infty \right\} \,, \quad% 
    \\
    o(g(n)) 
    &\equiv 
    \left\{f(n): \lim\limits_{n 
    \to \infty}{\frac{f(n)}{g(n)}} = 0 \right\} \,,\quad% 
    \\
    \Omega(g(n)) 
    &\equiv 
    \left\{f(n): \liminf\limits_{n 
    \to \infty}{\frac{f(n)}{g(n)}} > 0 \right\} \,,\quad % 
    \\
    \omega(g(n)) 
    &\equiv 
    \left\{f(n): \lim\limits_{n 
    \to \infty}{\abs{\frac{f(n)}{g(n)}}} = \infty \right\} \,.
    \end{array}
\end{align}

The quantum state of system $A$  is described by a density operator $\rho$ on a finite-dimensional Hilbert space $\mathcal{H}_A$. We denote the dimension by either $d_A$ or $\mathrm{dim}(\mathcal{H}_A)$.
Let
$\mathscr{L}(\mathcal{H}_A)$ be the set of all operators 
$Q:\mathcal{H}_A\to\mathcal{H}_A$, and
$\mathscr{S}(\mathcal{H}_A)$ be the subset of all density operators, 
$\mathscr{S}(\mathcal{H}_A)\subset \mathscr{L}(\mathcal{H}_A)$.
We denote the symmetric maximally entangled state on $\mathcal{H}_{A}^{\otimes 2}$ by
\begin{align}
\ket{\Phi}_{A_1 A_2}\equiv \frac{1}{\sqrt{d_A}} \sum_{i=0}^{d_A-1} \ket{i}_{A_1}\otimes \ket{i}_{A_2}
\end{align}
for $\mathrm{dim}(\mathcal{H}_{A_1})=\mathrm{dim}(\mathcal{H}_{A_2})=d_A$.
The quantum Fourier transform unitary 
$\mathsf{F}:\mathcal{H}_A\to\mathcal{H}_A$
is defined by
\begin{align}
    \mathsf{F} \equiv \frac{1}{\sqrt{d_A}}\sum_{k=0}^{d_A-1}\sum_{j=0}^{d_A-1} e^{2\pi i k j /d_A}\ketbra{k}{j}
    \label{Equation:QFT}
\end{align}
and the Heisenberg-Weyl unitaries by
\begin{align}
    \mathsf{X} \equiv \sum_{j=0}^{d_A-1}\ketbra{j+1}{j} \,,
    \quad
    \mathsf{Z} \equiv \sum_{j=0}^{d_A-1}e^{2\pi i j /d_A}\ketbra{j}{j}
    \label{Equation:phase_shift_unitary}
\end{align}
with addition modulo $d_A$.
The Controlled-Not ($\mathrm{CNOT}$) gate is a unitary that performs modular addition on the target qudit based on the control qudit,
and can be expressed as
\begin{align}
    \mathsf{CNOT} \equiv \sum_{j=0}^{d_A-1}\ketbra{j} \otimes \mathsf{X}^j
    \,.% 
    \label{Equation:CNOT_fate_def}
\end{align}

The minimal and maximal eigenvalues of a Hermitian operator $Q\in\mathscr{L}(\mathcal{H})$ are denoted by $\lambda_{\min}(Q)$ and $\lambda_{\max}(Q)$, respectively.
The  Schatten $p$-norm of any operator $Q$ is defined as $\norm{Q}_p \equiv \left(\Tr{\abs{Q}^p}\right)^{\frac{1}{p}}$, where $\abs{Q} \equiv \sqrt{Q^\dagger Q}$.
The trace norm is the Schatten 1-norm, i.e., $\TrNorm{Q} \equiv \Tr{\abs{Q}}$, and the supremum norm is $\norm{Q}_{\infty} \equiv \sqrt{\lambda_{\max}(Q^\dagger Q)}$. % 
The normalized trace distance between $\rho$ and $\sigma$  is given by $\frac{1}{2}\TrNorm{\rho - \sigma}$, and their 
fidelity by 
$F(\rho,\sigma)\equiv \norm{\sqrt{\rho}\sqrt{\sigma}}_1^2$. For $\rho\in\mathscr{S}({\mathcal{H}})$ and $\sigma\in\mathscr{L}({\mathcal{H}})$ where $\sigma \geq 0$,
the quantum relative entropy $\RelEntropy{\rho}{\sigma}$ % 
is defined as follows:
$\RelEntropy{\rho}{\sigma}=\trace\left[\rho(\log \rho-\log \sigma) \right]$, if 
$\mathrm{supp}(\rho)\subseteq\mathrm{supp}(\sigma)$;
and $\RelEntropy{\rho}{\sigma}=\infty$, otherwise.
Throughout the paper, 
all exponents and logarithms are
in the natural basis. The quantum chi-square divergence can be defined by 
$\chi^2(\rho||\sigma)=\frac{\partial^2}{\partial \alpha^2} D(\alpha\rho+(1-\alpha)\sigma||\sigma)\Big|_{\alpha=0}$ (see \cite[Sec. 2.6]{HircheTomamichel:24p}). Explicitly,
given a spectral decomposition  of a full-rank operator,
$\sigma=\sum_i \lambda_i \Pi_i$, we have % 
\cite[Eq.~(4)]{9344627}
\begin{align}
\label{eq:eta}
\chi^2(\rho||\sigma)
&=\sum_{i\neq j} 
\frac{\log(\lambda_i)-\log(\lambda_j)}
{\lambda_i-\lambda_j}
\trace\left[ (\rho-\sigma)\Pi_i (\rho-\sigma)\Pi_j\right]
\nonumber\\
&\phantom{=}+
\sum_{i} 
\frac{1}
{\lambda_i}
\trace\left[ (\rho-\sigma)\Pi_i (\rho-\sigma)\Pi_i\right]\,.
\end{align}

The von Neumann entropy is defined as $H(\rho)\equiv -\trace[\rho \log(\rho)]$. Given a bipartite state $\rho_{AB}\in\mathscr{S}(\mathcal{H}_A\otimes \mathcal{H}_B)$,
the quantum mutual information is 
$% 
I(A;B)_\rho \equiv H(\rho_{A})+H(\rho_{B})-H(\rho_{AB}) % 
$. % 
The conditional quantum entropy is defined by $H(A|B)_\rho \equiv H(\rho_{AB})-H(\rho_{B})$, and the quantum conditional mutual information is defined accordingly.

A quantum channel $\mathcal{N}_{A\to B} : \mathscr{S}(\mathcal{H}_{A})\to \mathscr{S}(\mathcal{H}_{B})$ is a linear completely-positive and trace-preserving (CPTP) map. 
Every quantum channel has a Stinespring representation. Specifically, there exists an isometric map
$\mathcal{U}_{A\to BW}(\rho)= V\rho V^\dagger$, such that 
\begin{align}
\mathcal{N}_{A\to B}=\trace_W\circ \mathcal{U}_{A\to BW}
\end{align}
where the operator $V:\mathcal{H}_A\to \mathcal{H}_B\otimes \mathcal{H}_W$ is an isometry, i.e., $V^\dagger V=\identity_A$.
In general, the system $W$ is interpreted as the receiver's environment.

\begin{remark}
In previous literature on quantum covert communication, the expression on the right-hand side of \eqref{eq:eta} is referred to as the $\eta$-divergence \cite{SheikholeslamiB16,10886999} \cite{9344627} 
  \cite[Sec. 1.1.4]{Zlotnick:24z}.  Here, we observe that this is in fact identical to the so-called  quantum chi-square divergence. Further details are given in  Subsection~\ref{Discussion:Chi_Square}.
\end{remark}

\section{Classical Information With Secrecy}
\label{Section:Covert-Secret Communication Over Classical-Quantum Channels}

First, we establish an achievability result for covert transmission of \emph{classical} information with \emph{secrecy} over a quantum channel. 
Originally, the classical setting of covert secrecy via a classical wiretap channel with key assistance was considered by Bloch \cite[Sec. VII-C]{Bloch16}.
This addition makes the protocol not only undetectable, but also prevents the
warden from extracting information about the transmitted message. 
Then, we determine the minimal key rate required to 
achieve the covert secrecy capacity with key assistance, i.e., we show the key rate is optimal by establishing a converse result.
Furthermore, we derive the  covert secrecy capacity without assistance, which is later used to apply
Devetak's approach \cite{Devetak05} of
constructing an entanglement-generation code from a secrecy code.

Secrecy requires that Willie cannot recover the message. In covert communication, the goal is to ensure that Willie is not even aware that a transmission is occurring. 
However, secrecy is still important: per discussion of \eqref{eq:Willie_Perror} below, covertness guarantees a small, but non-zero probability of transmission detection by Willie. In the rare event that transmission is detected, secrecy denies Willie access to the information contained therein.
In this section, we impose both covertness and secrecy as fundamental constraints. This requirement later plays a crucial role in establishing entanglement generation, where covertly transmitting quantum states demands an additional layer of security beyond mere undetectability.

\subsection{Coding Definitions}
\subsubsection{Covert Secrecy Code}
\label{Subsubsection:Secrecy_code}
Alice wishes % 
to send a classical message to Bob with two security guarantees: covertness and secrecy, as illustrated in Figure~\ref{Figure:Secrecy_Model}.

Consider
a classical-quantum channel
$\mathcal{P}_{X\to BW}:\mathcal{X}\to\mathscr{S}(\mathcal{H}_B\otimes\mathcal{H}_W)$
that maps a classical input 
$x\in\mathcal{X}$ to a quantum state $\pi_{BW}^{(x)}=\mathcal{P}_{X\to BW}(x)$.
Hence, the reduced states of Bob and Willie are given by 
$\sigma_x\equiv \trace_W\left(\pi_{BW}^{(x)}\right)$
and
$\omega_x\equiv \trace_B\left(\pi_{BW}^{(x)}\right)$, respectively. 
Table~\ref{tab:secrecy-notation} summarizes the notation in this subsection.

\begin{definition}
A classical secrecy code $(\mathcal{M},\mathcal{L},\mathcal{K},f,\Lambda)$ for the classical-quantum channel $\mathcal{P}_{X\to BW}$ consists of: 
\begin{itemize}
\item
a secret message set $\mathcal{M}$, % 
\item
a public message set $\mathcal{L}$,
\item
a key set $\mathcal{K}$, 
\item
an encoding function
 $f:\mathcal{M}\times\mathcal{L}\times\mathcal{K}\to \mathcal{X}^{ n}$, and 
 \item
 a collection of decoding measurements $\{\Lambda_{B^n}^{(m,\ell|k)}, (m,\ell)\in\mathcal{M} \times \mathcal{L}\}$ on $\mathcal{H}_B^{\otimes n}$, for $k\in\mathcal{K}$.
\end{itemize}
\end{definition}

Figure~\ref{Figure:Secrecy_Model} depicts our setting. 
Suppose that Alice and Bob share a random key $k$ that is uniformly distributed over $\mathcal{K}$.
Alice selects % 
secret and public messages $m\in\mathcal{M}$ and $\ell\in\mathcal{L}$,
both uniformly distributed,
 that she desires to send to Bob. 

In the security setting of covert communication, 
Alice makes a decision on whether to perform the communication task, or not. Assume 
$\mathcal{X}\equiv \{0,1,\ldots, \abs{\mathcal{X}}-1\}$.
If Alice decides to be inactive, the channel input is the zero sequence, i.e., $x^n=(0,0,\dots,0)$.
Otherwise, if she does communicate, then % 
she transmits a codeword
 $x^n=f{(m,\ell,k)}$ of length $n$,
using her access to the key $k$.
The joint output state is thus
\begin{align}
\rho_{B^n W^n}^{(m,\ell, k)}
=\mathcal{P}_{X\to BW}^{\otimes n} \left(f{(m,\ell,k)}\right)
=\bigotimes_{i=1}^n \pi_{ BW}^{ \left(f_i{(m,\ell,k)}\right)}
\,.
\end{align}

At the channel output, Bob and Willie receive 
$B^n=(B_1,B_2,\ldots,B_n)$ and
$W^n=(W_1,W_2,\ldots,W_n)$,
respectively. 
Using the key, Bob performs a decoding measurement $\{\Lambda_{B^n}^{(m,\ell|k)}, (m,\ell)\in\mathcal{M} \times \mathcal{L}\}$
  on his received system $B^n$, and obtains an estimate $(\hat{m},\hat{\ell})$ as the measurement outcome.
  The conditional  probability of decoding error given a secret message $m$, a public message $\ell$, and a key $k$  is
\begin{align}
    P_e^{(n)}(m,\ell,k) = 1 - \trace\left( \Lambda_{B^n}^{(m,\ell|k)} \rho_{B^n}^{(m,\ell,k)} \right)
    \,.
\end{align}
Hence, the average error probability is
\begin{align}
    \overline{P}_e^{(n)} = \frac{1}{|\mathcal{M}||\mathcal{L}||\mathcal{K}|}\sum_{\substack{(m,\ell)\in\\\mathcal{M}\times \mathcal{L}}}\sum_{k\in\mathcal{K}} P_e^{(n)}(m,\ell,k)% 
    \,.
    \label{Equation:Average_Pe}
\end{align}
 Meanwhile, Willie receives $W^n$ in the following state:
 \begin{align}
\rho_{W^n}^{(m)}=\frac{1}{|\mathcal{L}||\mathcal{K}|}\sum_{\ell\in\mathcal{L}}\sum_{k\in\mathcal{K}}
\rho_{W^n}^{(m,\ell,k)}
\label{Equation:Willie_Message_Output_Secrecy}
 \end{align}
 where $\rho_{W^n}^{(m,\ell,k)}=\trace_{B^n}\left( \rho_{B^nW^n}^{(m,\ell,k)}\right)$ is Willie's reduced state when conditioned on the secret message $m$, public message $\ell$, and the key $k$.
 As we average over the secret message set, we obtain
 \begin{align}
\overline{\rho}_{W^n}=\frac{1}{|\mathcal{M}|}\sum_{m\in\mathcal{M}}
\rho_{W^n}^{(m)} \,.
\label{Equation:Willie_Average_Output_Secrecy}
 \end{align}

 Willie 
 performs a hypothesis test to determine whether Alice has transmitted information or not.
In addition, if Willie identifies the transmission, he may also try to recover the secret message.
Detection failure could arise from either falsely identifying a transmission when none occurs (false alarm) or failing to detect an actual transmission (missed detection).
Denoting these error probabilities  by $P_{\text{FA}}% 
$ and $P_{\text{MD}}% 
$, respectively, and 
assuming equally likely hypotheses, % 
Willie's probability of error is
\begin{align}
   e_W = \frac{P_{\text{FA}} + P_{\text{MD}}}{2}
    \,.
    \label{eq:Willie_Perror}
\end{align}
A detector is ineffective if it performs no better than random guessing, yielding $e_{W} = \frac{1}{2}$. The objective of covert communication is to construct a code that forces Willie’s detector to become asymptotically ineffective.
By quantum hypothesis testing (see \cite[Sec. 9.1.4]{W2017}), the minimal error probability for Willie is given by % 
\begin{align}
    e_{W,\min}
    &= \frac{1}{2}\left(1-\frac{1}{2}\TrNorm{ \overline{\rho}_{W^n} -\omega_0^{\otimes n}}\right) \nonumber\\
    &\geq \frac{1}{2}\left(1 - \sqrt{\frac{1}{2}\RelEntropy{\overline{\rho}_{W^n}}{\omega_0^{\otimes n}}}\right)
    \label{Equation:Willie_Min_Error}
\end{align}
where $\omega_0 $  is Willie's output corresponding to the innocent input $x=0$ (no transmission), and the inequality follows by the quantum Pinsker’s inequality (see \cite[ % 
Th. 11.9.1]{W2017}).
Therefore, we call a code covert if the quantum relative entropy $\RelEntropy{\overline{\rho}_{W^n}}{\omega_0^{\otimes n}}$ tends to zero as $n\to\infty$.
The covert secrecy code is required to satisfy three requirements: reliability, covertness, and secrecy. 
Formally, an~$(\abs{\mathcal{M}},\abs{\mathcal{L}}, \abs{\mathcal{K}}, n,\varepsilon,\delta_{\text{cov}},\delta_{\text{sec}})$ secrecy code for the classical-quantum covert communication channel $\mathcal{P}_{X\to BW}$ satisfies
the following conditions: % 
\begin{enumerate}[(i)]
\item
\emph{Decoding Reliability:}
The probability of  decoding  error is $\varepsilon$-small, i.e., 
\begin{equation}
    \overline{P}_e^{(n)} \leq \varepsilon
    \,.
    \label{Equation:Error_S}
\end{equation}
\item
\emph{Covertness Criterion:}
In order to render Willie's detection ineffective, we require that
the quantum relative entropy is $\delta_{\text{cov}}$-small, i.e., 
\begin{equation}
    \RelEntropy{\overline{\rho}_{W^n}}{\omega_0^{\otimes n}} 
    \leq \delta_{\text{cov}}
    \label{eq:covertness_criterion}
\end{equation}
where 
$\bar{\rho}_{W^n}$ is Willie's average state when Alice is active and uses the channel (see \eqref{Equation:Willie_Average_Output_Secrecy}), and $\omega_0 $  is Willie's output corresponding to the innocent input $x=0$ (no transmission). By \eqref{Equation:Willie_Min_Error}, this guarantees that 
Willie's minimal error probability $e_{W,\min}$ is close to that of random guessing. 

\item
\emph{Secrecy:}
There exists a constant state $\breve{\rho}_{W^n}$, which does not depend on the secret message $m$, such that
the average leakage distance is $\delta_{\text{sec}}$-small, i.e.,
\begin{equation}
    \frac{1}{\abs{\mathcal{M}}}\sum_{m\in\mathcal{M}}
    \TrNorm{
    \rho^{(m)}_{W^n} - \breve{\rho}_{W^n}}
    \leq \delta_{\text{sec}}
\end{equation}
where 
$\rho_{W^n}^{(m)}$ is Willie's state for a given secret 
message $m$, 
as in \eqref{Equation:Willie_Message_Output_Secrecy}.
\end{enumerate}

 \subsubsection{Covert Secrecy Rate and Capacity}
In traditional coding problems, the secrecy rate is defined as $R_{\text{S}}\equiv\frac{\log \abs{\mathcal{M}}}{n}$, i.e., the number of information bits per channel use. However, in the covert setting, the best achievable transmission rate is zero, since the number of information bits is sublinear in $n$, and scales as $\log \abs{\mathcal{M}}=O(\sqrt{n})$.
Instead, we define the covert rate as follows. 
The covert secrecy rate is characterized as 
\begin{equation}
    L_\text{S}\equiv \frac{\log\abs{\mathcal{M}}}{\sqrt{n\delta_{\text{cov}}}} 
    \,,
\label{Equation:L_S}
\end{equation}
 hence, 
$\abs{\mathcal{M}}=e^{\sqrt{n\delta_{\text{cov}}}L_\text{S}}$. 
Similarly, we define the public message  and key rates as
\begin{align}
    L_{\text{public}}&\equiv \frac{\log\abs{\mathcal{L}}}{ \sqrt{n\delta_{\text{cov}}}}
\,\text{ and }\;% 
    L_{\text{key}}\equiv \frac{\log\abs{\mathcal{K}}}{ \sqrt{n\delta_{\text{cov}}}}
    \,,
\end{align}
respectively,
and thus $\abs{\mathcal{L}}=e^{\sqrt{n\delta_{\text{cov}}}L_\text{public}}$ and $\abs{\mathcal{K}}=e^{\sqrt{n\delta_{\text{cov}}}
L_\text{key}}$ (see \cite{10886999}, \cite{WangEB22}).
We  now  define an achievable covert rate and the covert capacity for sending secret classical information. 
\begin{definition}[Achievable covert secrecy rate]
\label{Definition:Achievable_Secrecy_Rate}
     A covert secrecy rate $L_\text{S} > 0$ is achievable \emph{with key assistance}  if, for every $\varepsilon,\delta_{\text{cov}},\delta_{\text{sec}} > 0$, sufficiently large $n$, and some public message and key rates,  $L_\text{public}$ and $L_\text{key}$, there exists a $(e^{\sqrt{n\delta_{\text{cov}}}L_\text{S}},e^{\sqrt{n\delta_{\text{cov}}}L_\text{public}}, e^{\sqrt{n\delta_{\text{cov}}}L_\text{key}}, n, \varepsilon, \delta_{\text{cov}},\delta_{\text{sec}})$ code for covert and secret classical communication. 
     Similarly, a covert secrecy rate $L_\text{S} > 0$ is  achievable \emph{without assistance} if it is achievable for $L_\text{key}=0$.
\end{definition}
Equivalently, a covert secrecy rate $L_\text{S}$ is achievable if there exists a sequence of codes of length $n$ approaching this rate, such that the error probability, covertness divergence, and leakage distance all tend to zero as $n\to\infty$.

\begin{definition}[Covert secrecy capacity with key assistance]
    The covert secrecy capacity $C_\text{S}^{\text{key}}(\mathcal{P})$ with key assistance of a classical-quantum covert communication channel $\mathcal{P}_{X\rightarrow BW}$ is the supremum of all achievable rates with key assistance.
\end{definition}

The covert secrecy capacity $C_\text{S}(\mathcal{P})$ without assistance is defined accordingly.

\begin{remark}
Local randomness at the encoder is essential for ensuring secrecy. In our setting, this role is fulfilled by a uniform public message $\ell$, which can be interpreted as a source of local randomness that is recovered by the receiver. This feature also plays a useful role in the converse proof.
In the key-assisted setting, Alice and Bob additionally share secret common randomness, i.e., a pre-shared secret key, in which case the public message is no longer required.
\end{remark}

\subsection{Assumptions}
\label{Section:Scenarios}

For simplicity, we assume a binary input, i.e., $\mathcal{X}=\{0,1\}$.
We consider covert communication under the following assumptions, which delineate the non-trivial operating regime in which covert communication is neither trivially impossible nor meaningless. Equivalent conditions are standard throughout the covert communication literature, both for classical channels~\cite{Bloch16} and for quantum channels~\cite{TahmasbiBloch:19p, tahmasbi2020covertIEEE,SheikholeslamiB16, 10886999, GagatsosBB20}.
    
\begin{enumerate}[i)]
    \item \emph{Non-trivial detection:}
    $\mathrm{supp}(\omega_1) \subseteq \mathrm{supp}(\omega_0)$ and $\omega_1 \neq \omega_0$.
    The first condition ensures that Willie cannot perfectly detect a non-innocent transmission; 
    the second excludes the degenerate case where detection by Willie is impossible.

    \item \emph{No unfair advantage for Bob:}
    $\mathrm{supp}(\sigma_1) \subseteq \mathrm{supp}(\sigma_0)$.
    Together with the first assumption, this guarantees that neither Bob nor Willie can detect a non-zero transmission with certainty. 
\end{enumerate}

Since the support of $\omega_0$ contains all of Willie's output states under the above assumptions, we may assume without loss of generality that $\omega_0$ has full support. 
Otherwise, we can redefine the channel such that $\mathcal{H}_W=\mathrm{supp}(\omega_0)$.

\begin{remark}
    When Willie's support condition $\mathrm{supp}(\omega_1) \subseteq \mathrm{supp}(\omega_0)$ is violated, Willie can identify any non-innocent transmission with certainty, as the quantum relative entropy $D(\omega_1 \| \omega_0)$ diverges. Covert communication is then fundamentally impossible, and the covert capacity is trivially zero (see also~\cite[App.~G]{Bloch16} and~\cite[Sec.~IV-C]{10886999}, where this impossibility is established for the classical and classical-quantum settings, respectively).
    Whether this condition holds depends on the specific channel and the choice of innocent input.
    In particular, the condition is automatically satisfied whenever the output density operator $\omega_0$ has a full support.
This occurs, for example, in the excitation channel treated in Subsection~\ref{subsection:example_excitation_channel}.
Furthermore, the condition holds in the continuous-variable case of
bosonic channels with thermal background noise~\cite{bash2015quantum, GagatsosBB20}. 
    In contrast, for some standard channels, such as % 
    the erasure channel, % 
    the outputs $\omega_0$ and $\omega_1$ are orthogonal. 
    In the case of the qubit depolarizing channel, the condition only holds if Willie receives a noisy version of the environment \cite{ZlotnickBP23}.
Otherwise,
covert communication is impossible. A similar behavior is observed in the amplitude-damping channel model. % 
On the other hand, for a qubit-flip channel we have  $\omega_0=\omega_1$,
hence the warden cannot distinguish between the inputs and covert communication reduces to the ordinary non-covert setting.
We go back to these examples in the Summary and Discussion (see Subsection~\ref{Subsection:Covert_Communication_Scenarios}).
\end{remark}

\begin{remark}
    While Bob's objective is to recover information (either classical or quantum), if $\mathrm{supp}(\sigma_1) \nsubseteq \mathrm{supp}(\sigma_0)$, then Bob has an unfair advantage over Willie in the sense that he can identify a non-innocent input with certainty. This improves the information scale  to $\sim \sqrt{n}\log{n}$, surpassing the standard square root law (see~\cite[App.~G]{Bloch16} and~\cite[Sec.~IV-B]{10886999} for the classical-information setting).  For example, in the dephasing channel model, if Bob observes the outcome $\ket{1}$, he can conclude with certainty that the input is non-innocent.
\end{remark}

\subsection{Key-Assisted Covert Secrecy}
First, we consider covert communication of secret classical messages, when Alice and Bob are provided with key assistance.
That is, they have access to a pre-shared secret key, at a limited rate $L_{\text{key}}$ (see Definition~\ref{Definition:Achievable_Secrecy_Rate}).

\subsubsection{Capacity Theorem}
For simplicity, we assume a binary input, i.e., $\mathcal{X}=\{0,1\}$.
The capacity theorem for secret and covert communication with key assistance is given below.
\begin{theorem}
\label{Theorem:Covert_Secrecy_Capacity}
Let $\mathcal{P}_{X\to BW}$ be a classical-quantum covert communication channel.
    Consider covert communication of classical information with secrecy via this channel, without a public message (i.e., $L_{\text{public}}=0$). If  $\mathcal{P}_{X\to BW}$ satisfies
    \begin{align}
        \text{supp}(\sigma_1) &\subseteq \text{supp}(\sigma_0) 
        \,,\;% 
        \text{supp}(\omega_1) \subseteq \text{supp}(\omega_0)
        \,,\;
        \omega_1\neq \omega_0
        \end{align}
        then the covert secrecy capacity \textit{with key assistance} is 
    \begin{align}
        C_\text{S}^{\text{key}}(\mathcal{P})=\frac{D(\sigma_1|| \sigma_0)}{\sqrt{\frac{1}{2}\chi^2(\omega_1||\omega_0)}}
        \,.
    \end{align}
    \label{Theorem:Secrecy}
\end{theorem}
The proof of Theorem~\ref{Theorem:Covert_Secrecy_Capacity} % 
is given in Section~\ref{Section:Secrecy_Proof} in the Supplementary Material. In the analysis, we combine several methods from previous works on secret and covert communication. % 
We build upon covert communication results without secrecy % 
\cite{10886999}, along with the secrecy coding approach proposed by Bloch \cite{Bloch16} for classical channels. We employ binning and one-time pad encryption to guarantee security, and then
  use the quantum channel resolvability lemma due to Hayashi 
in the secrecy analysis \cite{hayashi2006quantum}. 

Our  result in Theorem~\ref{Theorem:Secrecy} generalizes Bloch's characterization of the covert secrecy capacity for classical channels~\cite{Bloch16} to the quantum setting. 
We complement this characterization with additional results that provide further insight into covert secrecy. 
Specifically, in Theorem~\ref{Theorem:Secrecy_Key_Converse} below, we determine the minimal secret-key rate required to achieve the key-assisted covert secrecy capacity. 
Later, in Theorem~\ref{Theorem:Covert_Secrecy_Capacity_No_Key}, we characterize the covert secrecy capacity without assistance as well.

\begin{remark}
\label{Remark:Secrecy_Rate}
We observe that  we achieve the same covert rate as without secrecy \cite{SheikholeslamiB16,10886999}, albeit with a larger classical key.
Specifically, covert communication without secrecy requires a key 
rate $L_{\text{key}}\equiv \frac{\log\abs{\mathcal{K}}}{ \sqrt{n\delta_{\text{cov}}}}$ of 
\begin{align}
L_{\text{key}}\sim \frac{% 
\left[D(\omega_1||\omega_0)-D(\sigma_1||\sigma_0)
\right]_+
}{
\sqrt{\frac{1}{2}\chi^2(\omega_1||\omega_0)} 
} \,.
\end{align}
Whereas, here, we establish secrecy using  
\begin{align}
L_{\text{key}}\sim \frac{D(\omega_1||\omega_0) }{
\sqrt{\frac{1}{2}\chi^2(\omega_1||\omega_0)} 
} \,.
\label{Equation:Lkey_Secrecy_approx}
\end{align}
In the subsection below, we establish that this key rate is optimal. That is, the rate above represents the minimal amount of shared secret key % 
in order to achieve the key-assisted covert secrecy capacity. See Theorem~\ref{Theorem:Secrecy_Key_Converse} below. % 
\end{remark}

\begin{remark}
\label{Remark:OTP_Key_Rate}
Consider covert communication with the key rate in \eqref{Equation:Lkey_Secrecy_approx}.
If $D(\omega_1||\omega_0)\geq D(\sigma_1||\sigma_0)$, secrecy can be obtained through straightforward one-time pad (OTP) encryption as the key is longer than the message. 
However, if $D(\omega_1||\omega_0)< D(\sigma_1||\sigma_0)$, then our key is shorter than the message, hence secrecy requires a more elaborate coding scheme. In this case, we use a similar coding approach as in the classical work \cite{Bloch16} on covert secrecy for classical channels, using rate splitting and combining binning with OTP.
We bound the leakage using the quantum channel resolvability lemma due to Hayashi \cite{hayashi2006quantum}.
\end{remark}

\subsubsection{Minimal Key Rate (Converse)}
\label{Subsection:Converse_On_Key_Rate}

We show the optimality of our key rate via a matching converse. Specifically, 
 we derive a lower bound on the pre-shared secret key rate required to achieve the covert secrecy capacity for the c-q channel, and observe that the asymptotic lower bound agrees with \eqref{Equation:Lkey_Secrecy_approx}. 
\begin{theorem}
\label{Theorem:Secrecy_Key_Converse}
Let $\mathcal{P}_{X\to BW}$ be a classical-quantum covert communication channel as in Theorem~\ref{Theorem:Covert_Secrecy_Capacity}.
    Consider further a sequence of c-q covert secrecy codes  $(\mathcal{M}_n,\mathcal{L}_n,\mathcal{K}_n,f_n,\Lambda_n)$
    that achieves an information rate
    \begin{align}
        L_\text{S} 
        \geq
        (1 - \vartheta_n)\frac{D(\sigma_1||\sigma_0)}{\sqrt{\frac{1}{2} \chi^2(\omega_1||\omega_0)}}\,,
        \,\text{ with }
        L_{\text{public}}=0 
        \label{Equation:key_converse_info_rate_bound_in_theorem}
    \end{align}
    such that
    \begin{align}
        &\overline{P}_e^{(n)} 
        \leq \varepsilon_n
        \,, % 
        \nonumber\\
        &\RelEntropy{\overline{\rho}_{W^n}}{\omega_0^{\otimes n}} 
        \leq \delta_n^{\text{cov}}
        \,, % 
        \nonumber\\
        &\frac{1}{\abs{\mathcal{M}}}\sum_{m\in\mathcal{M}}
        \TrNorm{\rho^{(m)}_{W^n} - \breve{\rho}_{W^n}} 
        \leq \delta_n^{\text{sec}}
    \end{align}
    where $\varepsilon_n, \delta_n^{\text{cov}}, \delta_n^{\text{sec}}$ and $\vartheta_n$ tend to zero as $n \rightarrow \infty$. % 
    Then,
    \begin{equation}
        L_\text{key} \geq \frac{\RelEntropy{\omega_1}{\omega_0}}{\sqrt{\frac{1}{2}\chi^2(\omega_1||\omega_0)}} - \lambda_n^\text{key}
    \end{equation}
    where $\lambda_n^\text{key}$
    tends to zero as $n\to\infty$.
\end{theorem}
The proof of Theorem~\ref{Theorem:Secrecy_Key_Converse} 
is given in % 
Subsection~\ref{sec:Secrecy_Key_Converse_proof} in the Supplementary Material.

\subsection{Unassisted Covert Secrecy}
Next, we consider covert communication of secret classical messages, without key assistance.
That is,  $L_{\text{key}}=0$ (see Definition~\ref{Definition:Achievable_Secrecy_Rate}).
This result will be useful in our entanglement generation analysis as well. 
\begin{theorem}
\label{Theorem:Covert_Secrecy_Capacity_No_Key}
Let $\mathcal{P}_{X\to BW}$ be a classical-quantum covert communication channel as defined in Theorem~\ref{Theorem:Covert_Secrecy_Capacity}. Then, the covert secrecy capacity \textit{without key assistance} is given by
\begin{align}
    C_\text{S}(\mathcal{P})=\frac{\left[\RelEntropy{\sigma_1}{\sigma_0}
    -\RelEntropy{\omega_1}{\omega_0}
    \right]_+}{\sqrt{\frac{1}{2}\chi^2(\omega_1||\omega_0)}}
        \,.
    \label{Equation:secrecy_capacity_no_key}
\end{align}
\end{theorem}
The proof of Theorem~\ref{Theorem:Covert_Secrecy_Capacity_No_Key} % 
is given in % 
Section~\ref{Appendix:Covert_Secrecy_Capacity_No_Key_Proof} in the Supplementary Material. % 

\begin{remark}
The capacity expression above coincides with the secret key generation capacity derived by Tahmasbi and Bloch \cite{TahmasbiBloch:19p, tahmasbi2020covertIEEE}.
In principle, achievability of our result could also be obtained by performing covert quantum key distribution (QKD) followed by one-time-pad encryption \cite{TahmasbiBloch:19p, tahmasbi2020covertIEEE}. Tahmasbi and Bloch's framework thus paves the way to an explicit and efficient communication protocol. Nonetheless, they \cite{TahmasbiBloch:19p, tahmasbi2020covertIEEE} assume the use of a public communication channel that is not required to be covert and is therefore not regarded as a proof of communication. 
In other words, Willie always expects classical
communication on the public channel \cite[Remark 1]{tahmasbi2020covertIEEE}.  Our capacity result does not rely on an auxiliary public channel.
\end{remark}

Roughly speaking, the covert encoding scheme is based on a sparse coding protocol, % 
with only a fraction of $\alpha_n$ non-zero transmissions \cite{10886999}\cite{tahmasbi21covertsignaling}. This fraction is chosen as 
$\alpha_n=\frac{\gamma_n}{\sqrt{n}}$, where $\gamma_n$ tends to zero.
The average state, 
\begin{align}
\omega_{\alpha_n}=(1-\alpha_n)\omega_0+\alpha_n\omega_1\,,
\end{align}
is referred to in \cite{10886999} as the ``quantum-secure covert state.'' Further discussion and properties are shown in \cite[Sec. II-E]{10886999} and in the Supplementary Material. 

Note that our communication model assumes a uniform message, and consequently, the error, secrecy, and covertness criteria are defined in a message-average sense (see Subsection~\ref{Subsubsection:Secrecy_code}). However, to employ covert secrecy codes for the construction of entanglement-generation codes, a stronger achievability result is required, in which both the reliability and secrecy criteria hold uniformly over all messages. This strengthened result is established via an expurgation argument.

\begin{lemma}[Expurgated covert secrecy]
\label{Lemma:Expergated_classical_code}
   Consider a covert memoryless classical-quantum channel such that
    $\text{supp}(\sigma_1) \subseteq \text{supp}(\sigma_0)$, 
    $\text{supp}(\omega_1) \subseteq \text{supp}(\omega_0)$, and $\omega_1\neq\omega_0$. 
    Let $\alpha_n = \frac{\gamma_n}{\sqrt{n}} $   with $\gamma_n = o(1) \cap \omega\left(\frac{(\log{n})^{\frac{7}{3}}}{n^{\frac{1}{6}}}\right)$. 
    Then, 
    for any $\zeta_n \in o(1) \cap \omega\left((\log n)^{-\frac{2}{3}}\right)$,  there exist 
    $\zeta_n^{(1)} \in \omega\left((\log n)^{-\frac{4}{3}}
    n^{-\frac{1}{3}}\right)$,
    $% 
    {\zeta}_n^{(2)} \in \omega\left((\log n)^{-2}\right)$, $\zeta_n^{(3)} \in
    \omega\left((\log n)^{-1}\right)$
    and
    a classical-quantum covert secrecy
    code with a deterministic codebook $\mathscr{C}=\{x^n(m,\ell)\}$,
    {for the transmission of a secret message
    $m\in\mathcal{M}$ and a 
    public message
    $\ell\in\mathcal{L}$},
    such that, 
    \begin{align}
        &\log{\abs{\mathcal{M}}} 
        \nonumber\\
        &=
        \gamma_n\sqrt{n}
        \left[(1 -2\zeta_n)\RelEntropy{\sigma_1}{\sigma_0}
        -% 
        (1 + \zeta_n)\RelEntropy{\omega_1}{\omega_0}
        \right]_+ \,,
        \nonumber\\
        &\log\abs{\mathcal{L}}=\left(1 + \frac{1}{2}\zeta_n\right)\gamma_n\sqrt{n}\RelEntropy{\omega_1}{\omega_0}
        \label{eq:Theorem3LogMLogL_Expurgation}
    \end{align}
    and
\begin{subequations}
\begin{align}
         &\max_{(m,\ell)\in\mathcal{M}\times \mathcal{L}}
         {P}_e^{(n)}(m,\ell) 
         \leq 
         e^{-\zeta_n^{(1)}\gamma_n\sqrt{n}}
         \,,
         \\
         &\abs{\RelEntropy{\overline{\rho}_{W^n}}{\omega_0^{\otimes n}} - \RelEntropy{\omega_{\alpha_n}^{\otimes n}}{\omega_0^{\otimes n}}} 
         \leq e^{-% 
         {\zeta}_n^{(2)}\gamma^{\frac{3}{2}}_n n^\frac{1}{4}} 
         \,,
          \label{eq:ExpurgatedCovertness}
        \\
        &\TrNorm{\rho^{(m)}_{W^n} - \omega_{\alpha_n}^{\otimes n}} 
        \leq 
        e^{-\zeta_n^{(3)}\gamma_n^{\frac{3}{2}} n^{\frac{1}{4}}}
        \,,\;\text{for all $m\in\mathcal{M}$}
        \label{Equation:expurgated_secrecy}
    \end{align}
\label{Equation:Expurgation}
\end{subequations}
for sufficiently large $n$.
\end{lemma}

The proof of Lemma~\ref{Lemma:Expergated_classical_code} 
is given in Subsection~\ref{Subsubsection:Expurgation} in the Supplementary Material.

\section{Covert Entanglement Generation}
\label{Section:Entanglement Generation}
We now turn to our main problem of interest, i.e., covert entanglement generation. Remarkably, we establish a single-letter formula in this fully quantum model. % 
 \subsection{Channel Model}
Consider a quantum channel $\mathcal{N}_{A\to B}$ with a Stinespring dilation, $\mathcal{U}_{A\to BW}$.
In our setting, we assume that an adversarial warden, Willie, holds $W$ (in the context of secrecy, the environment is sometimes viewed as an eavesdropper and referred to as Eve). 
The complementary channel from Alice to Willie is defined by 
\begin{align}
\mathcal{N}^c_{A\to W} \equiv \trace_B\circ \mathcal{U}_{A\to BW}
\,.
\end{align}
Therefore, the isometric channel $\mathcal{U}_{A\to BW}$ maps Alice's input state $\rho_A$ into a joint state of Bob and Willie, $\rho_{BW} = \mathcal{U}_{A\rightarrow BW}(\rho_A)$, while
the marginal channels ${\mathcal{N}}_{A\to B}$ and $\mathcal{N}^c_{A\to W}$ produce the reduced states 
$\rho_{B} = \mathcal{N}_{A\rightarrow B}(\rho_A)$
and
$\rho_{W} = \mathcal{N}^c_{A\rightarrow W}(\rho_A)$, respectively.

Suppose that Alice would like to generate entanglement with Bob % 
covertly, i.e. without the adversarial warden Willie % 
knowing whether Alice transmitted or not.   Assume that $\ket{0}$ is the ``innocent'' input, corresponding to the case where Alice is inactive, i.e., she is not using the channel in order to generate shared entanglement with Bob.
  Let $\{\ket{0},\ket{1},\ldots, \ket{d_A-1}\}$ be an orthonormal basis, and 
  denote Bob's output by
  \begin{align}
\sigma_x &\equiv \mathcal{N}_{A\to B}(\ketbra{x})
\intertext{and Willie's output by }
\omega_x &\equiv \mathcal{N}^c_{A\to W}(\ketbra{x})
\,,
  \end{align}
  for $x\in\{0,1,\ldots, d_A-1\}$.
In particular, $\sigma_0$ and $\omega_0$ are Bob and Willie's respective outputs for the innocent input $\ket{0}$.
We are interested in covert entanglement generation under the same assumptions as % 
in Subsection~\ref{Section:Scenarios}.

\begin{remark}
An important feature of our model is that it adopts a worst-case perspective on the warden's capabilities. Specifically, we model communication from Alice to Bob via a quantum channel $\mathcal{N}_{A\to B}$ and grant Willie access to the entire environment of that channel, with % 
no structural or operational restrictions on Willie beyond those dictated by the laws of physics.
This stands in contrast to classical formulations \cite{bash12sqrtlawisit,BashGT13,WangWZ15,Bloch16}, which can be interpreted as fixing a particular measurement and then analyzing the statistics of the resulting observations. By avoiding any a priori restriction on Willie’s measurement, our approach effectively endows him with maximal observational power. Consequently, the resulting guarantees are robust, as they hold against any physically admissible strategy available to the warden.
\end{remark}

\subsection{Coding Definitions}
\subsubsection{Entanglement-Generation Code}
The definition of a code for % 
entanglement generation over a quantum channel 
is given below.

\begin{definition}
A $(T,n)$-entanglement-generation code  consists of a Hilbert space $\mathcal{H}_M$ of dimension 
$\mathrm{dim}(\mathcal{H}_M)=T$, where $T$ is an integer,
 and
a collection of encoding and decoding maps, {$\mathcal{F}_{M\to A^n}:\mathscr{S}(\mathcal{H}_M)\to \mathscr{S}(\mathcal{H}_A^{\otimes n})$ and $\mathcal{D}_{B^n\to \widehat{M}}:\mathscr{S}(\mathcal{H}_B^{\otimes n})\to \mathscr{S}(\mathcal{H}_M)$},  respectively. We denote the entanglement-generation code by 
{$(\mathcal{F},\mathcal{D})$}.
\end{definition}

The setting is depicted in 
Figure \ref{Figure:EG_Model}. The goal is to generate entanglement between Alice and Bob, without being detected by Willie.
Suppose that
Alice prepares a maximally entangled state
$\ket{\Phi}_{RM}$ on 
$\mathcal{H}_M^{\otimes 2}$,
locally, where $R$ is a resource that she keeps, and $M$ is the resource that she would like to distribute to Bob. 
Table~\ref{tab:eg-notation} summarizes the notation in this subsection.

\input{Table_III_EG}

In the setting of covert entanglement generation, % 
Alice makes a decision on whether to perform the % 
task, or not. 
If Alice decides to be inactive, the channel input is $\ket{0}^{\otimes n}$.
Otherwise, if she does perform the task,  she applies an encoding map  {$\mathcal{F}_{M\rightarrow A^n}$} 
on her ``quantum message'' $M$, which results in a quantum state
{
\begin{align}
\tau_{RA^n}=(\mathrm{id}_R\otimes \mathcal{F}_{M\to A^n} )\left(\ketbra{\Phi}_{RM}\right)
\,.
\end{align}
}
She then transmits the encoded system $A^n=(A_1,A_2,\ldots,A_n)$ using $n$ instances of the
quantum channel $\mathcal{U}_{A\rightarrow BW}$. 
The joint average output state is thus
{
\begin{align}
\tau_{RB^n W^n}=(\mathrm{id}_R\otimes \mathcal{U}_{A\to BW}^{\otimes n} )\left(\tau_{RA^n}\right)
\,.
\end{align}
}
At the channel output, Bob and Willie receive 
$B^n% 
$ and
$W^n% 
$,
respectively. 
Bob  performs a decoding operation {$\mathcal{D}_{B^n\to\widehat{M}}$} on his received system $B^n$, which recovers a state 
\begin{align}
\tau_{R\widehat{M}}% 
=
(\mathrm{id}_R\otimes \mathcal{D}_{B^n\to\widehat{M}})\left(\tau_{RB^n}\right)
\label{Bob_decoded_state_prerequisites}
\end{align}
where
{$\tau_{RB^n}=\trace_{W^n}\left( \tau_{RB^nW^n}\right)$} is the reduced output state of Alice and Bob alone. % 
 Meanwhile, Willie receives $W^n$ in the reduced state
 \begin{align}
\tau_{W^n}=\trace_{RB^n}\left( \tau_{RB^nW^n}\right)
\,.
\label{Equation:Willie_Average_Output}
 \end{align}
 He 
 performs a hypothesis test to determine whether Alice has transmitted  or not.

A $(T,n,\varepsilon,\delta)$-code for covert entanglement generation satisfies
the following conditions. % 
\begin{enumerate}[(i)]
\item
\emph{Decoding Reliability:}
The fidelity between the resulting state and the maximally entangled state is $\varepsilon$-close to $1$,
\begin{equation}
    F\left( \tau_{R\widehat{M}} \,,\; \Phi_{RM}\right) \geq 1-\varepsilon
    \label{Equation:Fidelity_EG}
\end{equation}
where $\Phi\equiv \ketbra{\Phi}$.
\item
\emph{Covertness Criterion:}
 The quantum relative entropy  is $\delta$-small, i.e.,
\begin{equation}
    \RelEntropy{% 
    {\tau}_{W^n}}{\omega_0^{\otimes n}} 
    \leq \delta,
\end{equation}
where 
{$% 
{\tau}_{W^n}$} is Willie's average state when Alice is active and uses the channel (see \eqref{Equation:Willie_Average_Output}), and $\omega_0 \equiv \mathcal{N}^c_{A \rightarrow W} (\ketbra{0})$  is the output corresponding to the innocent input (no transmission).  By \eqref{Equation:Willie_Min_Error}, this guarantees that 
Willie's minimal error probability $e_{W,\min}$ for detection tends to that of a random guess.
\end{enumerate}
We note that, by the quantum Pinsker inequality, we also have
{
$\norm{ % 
{\tau}_{W^n}-\omega_0^{\otimes n} }_1\leq \delta'\equiv\sqrt{2\delta }$.
}
That is, the two outputs of Willie, when Alice is active or inactive, are 
$\delta$-indistinguishable.

\begin{remark}
\label{Remark:Stinespring_dilations_equivalence}
Since all Stinespring dilations are isometrically equivalent and the relative entropy is invariant with respect to an isometry, the characterization does not depend on the choice of dilation (see \cite[Ex. 5.2.5 and 11.8.6]{W2017}).
\end{remark}

\begin{remark}
    As can be seen in the proof for covert entanglement generation (see Section \ref{sec:Proof Covert Entanglement Generation}), Alice applies  an isometric encoder, hence we need not assume that Willie cannot access the encoder's environment, as assumed by Anderson et al.~\cite{anderson2025achievability}. % 
\end{remark}

\begin{remark}
Entanglement generation is intimately related to secrecy. Due to the no-cloning theorem, the transmission of quantum information inherently ensures  secrecy
\cite{W2017}.  If the eavesdropper, Eve, could obtain any information about the quantum information that Alice is sending to Bob, then Bob would not be able to recover it. Otherwise, this would contradict the no-cloning theorem. 
Devetak \cite{Devetak05} introduced a coherent version of classical secrecy codes, which leverage their privacy properties to define subspaces where Alice can securely encode quantum information, ensuring Eve's inaccessibility. 
In this sense, secrecy is both necessary and sufficient in order to establish entanglement generation.
We discuss the connection between classical secrecy and entanglement generation in Section~\ref{Section:Discussion}.
\end{remark}

\subsubsection{Entanglement-Generation Capacity}
In traditional coding problems, the entanglement rate is defined as $R_{\text{EG}}\equiv\frac{\log\left[ \mathrm{dim}(\mathcal{H}_M) \right]}{n}$, i.e., the number of EPR pairs per channel use. However, in the covert setting, the best achievable entanglement rate is zero, since the number of EPR % 
pairs is sublinear in $n$, and scales as $\log\left[ \mathrm{dim}(\mathcal{H}_M)\right]=O(\sqrt{n})$.
Instead, we define the covert  entanglement-generation rate  as 
\begin{equation}
    L_\text{EG}\equiv \frac{\log{T}}{\sqrt{n\delta}} 
\label{Equation:L_EG}
\end{equation}
where $T=\mathrm{dim}(\mathcal{H}_M)$ is the entanglement dimension, hence, 
$T=e^{\sqrt{n\delta}L_\text{EG}}$. We can now  define an achievable covert rate and the covert capacity for entanglement generation. 
\begin{definition}[Achievable covert entanglement-generation rate]
     A covert entanglement-generation rate $L_\text{EG} > 0$ is achievable if for every $\varepsilon,\delta > 0$ and sufficiently large $n$, there exists a $(e^{\sqrt{n\delta}L_\text{EG}}, n, \varepsilon, \delta)$ code for covert entanglement generation via the quantum channel $\mathcal{N}_{A\to B}$.
\end{definition}
Equivalently, a covert entanglement-generation rate $L_\text{EG}$ is achievable if there exists a sequence of codes of length $n$ approaching this rate, such that both the fidelity tends to one, and the covertness divergence tends to zero, as $n\to\infty$.
\begin{definition}[Covert entanglement-generation capacity]
    The covert entanglement-generation capacity $C_\text{EG}(\mathcal{N})$ of a quantum channel $\mathcal{N}_{A\rightarrow B}$ is the supremum of all achievable rates.
\end{definition}

\subsection{Capacity Theorem% 
}

Now we reach  our main result. % 
For simplicity, we assume $d_A=2$.
Recall that $\mathcal{U}_{A\to BW}$ is a Stinespring dilation for the quantum channel $\mathcal{N}_{A\to B}$. As pointed out in Remark~\ref{Remark:Stinespring_dilations_equivalence}, the choice of dilation is arbitrary.
\begin{theorem}
\label{Theorem:Entanglement_Generation}
Consider  covert entanglement generation via a quantum channel $\mathcal{U}_{A\to BW}$  that satisfies
    \begin{align}
        \text{supp}(\sigma_1) &\subseteq \text{supp}(\sigma_0) 
        \,,\;% 
        \text{supp}(\omega_1) \subseteq \text{supp}(\omega_0)
        \,,\;
        \omega_1\neq \omega_0 \,.
        \label{Equation:Support_Conditions}
        \end{align}
    Then, the covert entanglement-generation capacity is given by 
\begin{align}
C_\text{EG}(\mathcal{N})=\frac{\left[\RelEntropy{\sigma_1}{\sigma_0}
    -\RelEntropy{\omega_1}{\omega_0}
    \right]_+}{\sqrt{\frac{1}{2}\chi^2(\omega_1||\omega_0)}}
\,.
\label{Equation:Entanglement_Generation}
\end{align}
\end{theorem}
The proof of Theorem~\ref{Theorem:Entanglement_Generation} is given in Section~\ref{sec:Proof Covert Entanglement Generation}.

\begin{remark}
Recall that $\{\ket{x}\}_{x=0,1}$ is an orthonormal basis for the input Hilbert space $\mathcal{H}_A$. We may define a classical-quantum channel $\mathcal{P}_{X\to BW}$ by 
    \begin{align}
\mathcal{P}_{X\to BW}(x)\equiv \mathcal{U}_{A\to BW}(\ketbra{x}_A)
\label{Equation:P_channel_EG}
    \end{align}
    for $x\in\mathcal{X}$, where $\mathcal{X}\equiv \{0,1\}$.
In the achievability proof for the entanglement-generation capacity, we construct an entanglement generation code for the quantum channel $\mathcal{N}_{A\to B}$ from classical secrecy codes for the classical-quantum channel $\mathcal{P}_{X\to BW}$, following the approach by Devetak \cite{Devetak05}.
\end{remark}

\begin{remark}
    We obtain a single-letter formula for the entanglement-generation capacity, which is the same as % 
    for secret classical information without key assistance.
    That is, the number of EPR pairs that Alice can generate covertly with Bob,  $\log T$, % 
 is the same as the number of classical secret bits that can be sent covertly. 
 This arises from Devetak's approach \cite{Devetak05}, of constructing a code for entanglement generation using a classical secrecy code. 
 We discuss this further in Subsection~\ref{Section:Key_Assistance}.
 
\end{remark}

\begin{remark}
    Although the quantum (entanglement generation) rate equals the classical rate in value, this does not imply that ``there is no quantum advantage.'' Although in both cases we consider a quantum channel, the two tasks are fundamentally different: entanglement generation is more demanding than classical communication, with typically lower achievable rates in the non-covert setting. Furthermore, the entanglement-generation rate measures qubit pairs per transmission, whereas the classical rate measures bits per transmission. Thus, while their numerical values coincide in this setting, their operational meanings remain fundamentally distinct.
\end{remark}

\begin{remark}
\label{Remark:Secret_Key_Scaling}
    In a recent work on covert quantum communication \cite{anderson2025achievability}, a lower bound is established using Pauli twirl modulation with the aid of $O(\sqrt{n}\log{n})$ key bits, based on a method that employs a sparse signaling approach with the order of the number of channel uses $n$. 
    Here, on the other hand, we do not use a key. % 
\end{remark}

\begin{remark}
In \cite{anderson2024covert}, the authors study covert quantum communication over optical (bosonic) channels, which are infinite-dimensional. Their analysis is tailored to this setting and employs % 
coding constructions specific to qubit modulation over optical channels. By contrast, our framework focuses on finite-dimensional quantum channels and yields a single-letter characterization of the covert entanglement-generation capacity in terms of intrinsic channel quantities.
\end{remark}

\begin{remark}
\label{Remark:Anti_Degraded}
As without covertness, the entanglement-generation capacity is zero for  anti-degradable channels. Specifically, suppose that there exists a degrading channel $\mathcal{D}_{W\to B}$ such that 
$% 
    \mathcal{N}_{A\to B} = \mathcal{D}_{W \to B} \circ \mathcal{N}^c_{A\to W}
$.  % 
Then, by QRE monotonicity   
$% 
    \RelEntropy{\sigma_1}{\sigma_0}
    \leq \RelEntropy{\omega_1}{\omega_0} 
$, hence the numerator in \eqref{Equation:Entanglement_Generation} is zero  \cite[Th.~11.8.1]{W2017}. % 
In particular, this holds for entanglement-breaking channels.
\end{remark}

\subsection{Example: Qubit Excitation Channel}
\label{subsection:example_excitation_channel}
Here, we demonstrate our results through the example of the excitation channel, % 
the opposite process for the known amplitude-damping channel. % 
From a physical perspective, excitation-type channels arise naturally in models of energy exchange between a two-level quantum system and a thermal environment. In particular, the generalized amplitude-damping channel is widely studied in the quantum information literature, as it describes the interaction of a qubit with a finite-temperature bath and captures both relaxation and thermal excitation processes \cite{KhatriSharmaWilde2020, nielsen2010quantum}. Such channels are standard noise models for superconducting qubits and other low-temperature quantum hardware platforms \cite{chirolli2008decoherence,myatt2000decoherence,turchette2000decoherence}, and can be viewed as the qubit analogue of bosonic thermal channels used to model lossy communication in the presence of background noise \cite{zou2017protecting,KhatriSharmaWilde2020}. The excitation channel considered here represents one direction of such thermally induced population transfer and, thus, constitutes a physically meaningful non-unital noise model.

The excitation channel is represented by the following input-output relation:
\begin{align}
    \mathcal{N}_{A \to B}(\rho)=K_0\rho K_0^\dagger
+K_1\rho K_1^\dagger
\end{align}
with the Kraus operators 
\begin{align}
    K_0 = \sqrt{1-\gamma}\ketbra{0} + \ketbra{1}, 
    \quad
    K_1 = \sqrt{\gamma}\ketbra{1}{0}
\end{align}
for $\gamma \in (0,1]$.
Bob and Willie's respective outputs are 
\begin{align}
    \sigma_0 &= (1-\gamma)\ketbra{0} + \gamma\ketbra{1}
    &
    \sigma_1 &= \ketbra{1}
    \\
    \omega_0 &= (1-\gamma)\ketbra{0} + \gamma\ketbra{1}
    &
    \omega_1 &= \ketbra{0}
    \,.
\end{align}
By Theorem~\ref{Theorem:Entanglement_Generation},
the covert entanglement-generation capacity is given by
\begin{align}
    C_\text{EG}(\mathcal{N}) = 
    \begin{cases}
    \log{\left(\frac{1-\gamma}{\gamma}\right)}\sqrt{\frac{2(1-\gamma)}{\gamma}}
    &
    \text{if $\gamma\in \left(0,\frac{1}{2}\right)$}
    \\
    0& \text{if $\gamma\in \left[\frac{1}{2},1\right]$}
    \end{cases}
    \,.
\end{align}
\begin{figure}[bt]% 
\vspace*{0.2in}
\center
\scalebox{0.8}{
\input{C_EG_example_2_plot_vs_gamma_no_key}
}
\caption{Covert entanglement-generation capacity of the excitation channel as a function of the excitation probability $\gamma$. 
As $\gamma \to 0$, the covert constraint becomes trivial and the rate is unbounded on the $\mathcal{O}(\sqrt{n})$ scale. 
The red marker indicates the transition point at $\gamma = 0.5$, beyond which the channel is anti-degradable and entanglement generation is impossible. Consequently, the capacity is zero.
\label{Figure:C_EG(gamma)}}
\end{figure}
Figure~\ref{Figure:C_EG(gamma)} plots
the covert entanglement-generation capacity.
As 
$\gamma \to 0$, the channel becomes noiseless and covertness is trivial. In the scale of $O(\sqrt{n})$, the number of EPR pairs % 
is then linear in $n$, hence the rate 
is unbounded.
Intuitively, as $\gamma$ increases, it is more difficult for Bob to distinguish between the outputs, while Willie's output states become more distinguishable. This results in a monotonic non-increasing curve.
For $\gamma \geq \frac{1}{2}$, the channel is anti-degradable \cite{KhatriSharmaWilde2020}, in which case entanglement generation is impossible   and % 
the entanglement-generation capacity is  zero (see Remark~\ref{Remark:Anti_Degraded}).

\section{Proof of Theorem~\ref{Theorem:Entanglement_Generation} (Covert Entanglement Generation)}\label{sec:Proof Covert Entanglement Generation}

In the achievability proof, we  construct an entanglement-generation code from  classical secrecy codes following the approach by Devetak \cite{Devetak05} \cite[Sec. 24.4]{W2017}. The proof sketch is given below. At first, we assume the availability of covert LOCC assistance, i.e., free undetectable classical communication. The derivation involves a state approximation via Parseval’s relation and Uhlmann’s theorem, such that Alice and Bob's state is approximately decoupled from Willie. This leads to the generation of a tripartite GHZ state between Alice and Bob. % 
They then apply Fourier transform and phase shifts % 
to convert the GHZ state into bipartite entanglement. % 
Next, we show that  entanglement generation is achievable without any assistance, i.e., using quantum communication alone. % 
To show that the communication scheme is covert, we use the properties of Willie's state from the classical secrecy setting and  continuity properties of the quantum relative entropy. % 
Finally, we observe that the converse part immediately  follows  from the 
secrecy capacity result in Theorem~\ref{Theorem:Covert_Secrecy_Capacity_No_Key}.

\subsection{Analytic Tools}

\subsubsection{Uhlmann's theorem}
The theorem below shows the relationship between the purifications of states that are close to one another in trace distance. 
\begin{theorem}[Uhlmann’s theorem  {\cite[Lemma 2.2]{DevetakHarrowWinter:08p}}]
\label{Theorem:Uhlmann_theorem}
    For every pair of pure quantum states $\ket{\psi}_{AB}$ and $\ket{\theta}_{AC}$ such that their reduced states $\psi_A$ and $\theta_A$ satisfy
    \begin{align}
        \TrNorm{\psi_A - \theta_A} \leq \varepsilon
    \end{align}
    there exists a partial isometry $W_{B \to C}$ such that
    \begin{align}
        \TrNorm{\left(\identity_A \otimes W_{B \to C}\right)\psi_{AB}\left(\identity_A \otimes W_{B \to C}\right)^\dagger - \theta_{AC}} \leq 2\sqrt{\varepsilon}
        \,,
        \label{eq:Uhlmann_theorem_bound}
    \end{align}
    where $\psi_{AB}\equiv \ketbra{\psi}_{AB}$ and $\theta_{AC}\equiv \ketbra{\theta}_{AC}$.
\end{theorem}

\begin{remark}
    \label{Remark:Uhlmann_theorem_isometry_specific_construction}
    The partial isometry $W_{B\to C}$ in Theorem~\ref{Theorem:Uhlmann_theorem} can be constructed explicitly as follows~\cite{uhlmann1976transition}\cite[Sec.~9.2.2]{nielsen2010quantum}. Define the operator $Q: \mathcal{H}_B \to \mathcal{H}_C$ by
    \begin{align}
      Q \equiv \Tr_A\!\Big( \ket{\theta}_{AC}\!\bra{\psi}_{AB} \Big)
      \,.
    \end{align}
    Let $Q = W_{B\to C}\, \abs{Q}$ be the polar decomposition of $Q$, where $\abs{Q} = \sqrt{Q^\dagger Q}$  and $W_{B\to C}$ is a partial isometry that maps from the closed range of $Q$ to that of $Q^\dagger$ (see \cite[Th.~6.1.2]{kadison1986fundamentals}). Then $W_{B\to C}$ is the desired partial isometry. Indeed, one can verify that $\bra{\theta}_{AC} (\identity_A \otimes W_{B\to C}) \ket{\psi}_{AB} = \Tr{\abs{Q}} = \sqrt{F(\psi_A, \theta_A)}$, which is the maximum overlap achievable over all partial isometries from $B$ to $C$ (see \cite[Lemma~9.5]{nielsen2010quantum}).
    The bound in \eqref{eq:Uhlmann_theorem_bound} then follows from the Fuchs--van de Graaf inequalities \cite[Th. 9.3.1]{W2017}.
\end{remark}

\subsubsection{Continuity}
The following lemma shows that  the relative entropy $D(\rho_{W^n}||\omega_0^{\otimes n})$ is continuous in $\rho_{W^n}$.
\begin{lemma}[see {\cite[Lemma 13]{10886999}}]
\label{Lemma:bound_on_rel_entropy_diff}
    Let $\rho_{W^n}\in\mathscr{S}(\mathcal{H}_W^{\otimes n})$ be an arbitrary state, and $\omega_0, \omega_1$  density operators in $\mathscr{S}(\mathcal{H}_W)$ with supports satisfying $\text{supp}(\omega_1) \subseteq \text{supp}(\omega_0)$. Consider
    \begin{align}
        \omega_{\alpha_n} = (1-\alpha_n)\omega_0 + \alpha_n \omega_1
    \end{align}
    where $\alpha_n\in (0,1)$, and $\lim_{n \to \infty}{\alpha_n} = 0$.
    Assume that
    $% 
        \TrNorm{\rho_{W^n} - \omega_{\alpha_n}^{\otimes n}} \leq e^{-1}
    $ % 
    and 
    $\text{supp}(\rho_{W^n}) \subseteq \text{supp}(\omega_{0}^{\otimes n})$. 
    Then, % 
    \begin{align}
        &\abs{\RelEntropy{\rho_{W^n}}{\omega_0^{\otimes n}} - \RelEntropy{\omega_{\alpha_n}^{\otimes n}}{\omega_0^{\otimes n}}} 
        \leq \TrNorm{\rho_{W^n} - \omega_{\alpha_n}^{\otimes n}}
        \nonumber\\
        &\phantom{=}\times
        \left[n\log{\left(\frac{4\dim{d_W}}{(\lambda_{\min}(\omega_0))^3}\right) -\log\TrNorm{\rho_{W^n} - \omega_{\alpha_n}^{\otimes n}}}\right]
    \end{align}
    for sufficiently large $n$.
\end{lemma}

\subsection{Achievability Proof}

    \begin{proposition}
    \label{Proposition:Achievability_Covert_Entanglement_Generation_No_LOCC}
       Consider  covert entanglement generation via a quantum channel $\mathcal{U}_{A\to BW}$  that satisfies
    \begin{align}
        \text{supp}(\sigma_1) &\subseteq \text{supp}(\sigma_0) 
        \,,\;% 
        \text{supp}(\omega_1) \subseteq \text{supp}(\omega_0)
        \,,\;
        \omega_1\neq \omega_0 \,.
        \end{align}
     Let $\alpha_n = \frac{\gamma_n}{\sqrt{n}}$
        with $\gamma_n = o(1) \cap \omega\left(\frac{(\log{n})^{\frac{7}{3}}}{n^{\frac{1}{6}}}\right)$. Then, 
        for any $\zeta_n \in o(1) \cap \omega\left((\log n)^{-\frac{2}{3}}\right)$, 
        there exist
        $\Tilde{\zeta}_n \in
        \omega\left((\log n)^{-\frac{4}{3}}
        n^{-\frac{1}{3}}\right)$ and 
        a covert entanglement-generation code % 
        such that for  sufficiently large $n$,
        \begin{align}
            &\log{\left[\mathrm{dim}(\mathcal{H}_M)\right]} 
            \nonumber\\
            &= \gamma_n\sqrt{n}
        \left[(1 -2\zeta_n)\RelEntropy{\sigma_1}{\sigma_0}
        -% 
        (1 + \zeta_n)\RelEntropy{\omega_1}{\omega_0}
        \right]_+,
            \label{eq:Theorem2LogMLogK}
        \end{align}
        and
        \begin{align}
            &F\left(\Phi_{RM}\,,\; \tau_{R\widehat{M}} \right)  \geq 
            1 - e^{-\Tilde{\zeta}_n\gamma_n \sqrt{n}}
            \;, 
            \nonumber\\
            &\abs{\RelEntropy{% 
            {\tau}_{W^n}}{\omega_0^{\otimes n}} - \RelEntropy{\omega_{\alpha_n}^{\otimes n}}{\omega_0^{\otimes n}}} 
            \leq 
            e^{-\Tilde{\zeta}_n\gamma_n \sqrt{n}}
        \end{align}
        where 
        $\Phi_{RM}$ is the maximally entangled state Alice shares with her reference, $\tau_{R\widehat{M}}$ is Bob's decoded state defined in \eqref{Bob_decoded_state_prerequisites}, and $% 
        {\tau}_{W^n}$ is Willie's average output state defined in \eqref{Equation:Willie_Average_Output}.
    \end{proposition}

\begin{proof}[Proof of Proposition~\ref{Proposition:Achievability_Covert_Entanglement_Generation_No_LOCC}]
First, we show achievability while assuming covert LOCC assistance, i.e., free undetectable classical communication (we eliminate this assumption later).
In particular, Alice applies a quantum instrument,
\begin{align}
\mathcal{F}_{M\to A^n J}(\rho)=\sum_{j} \mathcal{F}_{M\to A^n }^{(j)}(\rho)\otimes \ketbra{j}_J
\end{align}
where $\{\mathcal{F}_{M\to A^n }^{(j)}\}$ is a collection of encoding channels and $J$ is a classical register that stores a value $j$. 
She transmits $A^n$ via the channel $\mathcal{N}_{A\to B}^{\otimes n}$.
In addition, she
 uses an undetectable classical link to send $j$ to Bob, free of cost.
 Bob then applies a controlled decoder $\mathcal{D}_{B^n\to\widehat{M}}^{(j)}$.
 Later, we observe that this assistance is unnecessary. 
Consider a quantum channel $\mathcal{N}_{A\to B}$ with a Stinespring dilation $\mathcal{U}_{A\to BW}$, corresponding to an isometry $V_{A\to BW}$.
    In the achievability proof, we use the properties of  classical secrecy codes in order to construct a code for entanglement generation.

 \subsubsection{Code Conversion}   

We convert the classical code into a quantum one. 
Recall that $\{\ket{x}\}_{x=0}^{d_A-1}$ is an orthonormal basis for the input Hilbert space $\mathcal{H}_A$.
    Let $\{x^n(m,\ell)\}$ be a codebook as in Lemma~\ref{Lemma:Expergated_classical_code} for the transmission of classical information via the classical-quantum channel $\mathcal{P}_{X\to BW}$, defined by% 
    \begin{align}
\mathcal{P}_{X\to BW}(x)\equiv \mathcal{U}_{A\to BW}(\ketbra{x}_A)
\label{Eq:P_U_relation}
    \end{align}
    for $x\in\mathcal{X}$, where $\mathcal{X}\equiv \{0,1\}$.

    Here, Alice creates a quantum codebook $\{\ket{\phi_m}_{A^n} \,,\; m=0,1,\ldots,T-1\}$ with the following ``quantum codewords,''
    \begin{align}
    \ket{\phi_m}_{A^n} 
    = \frac{1}{\sqrt{\abs{\mathcal{L}}}} \sum_{\ell\in\mathcal{L} } e^{it(m,\ell)} \ket{x^n(m,\ell)}_{A^n}
    \end{align}
    where the phases $\{t(m,\ell)\}$ are chosen later. 
    In order to encode, she applies an isometry $U$ that maps $\ket{m}_M$ to $\ket{\phi_m}_{A^n}$.
    First, suppose that Alice prepares a maximally entangled state % 
    \begin{align}
        \ket{\Phi}_{RM} \equiv \frac{1}{\sqrt{T}}\sum_{m=0}^{T-1}\ket{m}_R \otimes \ket{m}_M
    \end{align}
    where $R$ is the resource that she keeps, and $M$ is the ``quantum message'' that she would like to distribute to Bob.
    Then, Alice % 
    copies the value of $m$ in the $M$ register to another register $M'$ using a $\mathrm{CNOT}$ gate:
    \begin{align}
        \ket{\tau}_{RMM'} &= (\identity\otimes\mathsf{CNOT})(\ket{\Phi}_{RM}\otimes \ket{0}_{M'})
        \nonumber\\
        &= \frac{1}{\sqrt{T}}\sum_{m=0}^{T-1}\ket{m}_R \otimes \ket{m}_M \otimes \ket{m}_{M'}
        \label{Equation:Alice_copier}
    \end{align}
    (see \eqref{Equation:CNOT_fate_def}).
    Next,
    Alice applies the isometry $U$ to encode from $M'$ to $A^n$, hence
    \begin{align}
        \ket{\tau}_{RMA^n} &=(\identity\otimes\identity\otimes U)\ket{\tau}_{RMM'}
        \nonumber\\
        &= \frac{1}{\sqrt{T}}\sum_{m=0}^{T-1}\ket{m}_R\otimes\ket{m}_M\otimes\ket{\phi_{m}}_{A^n}
        \,.
        \label{Equation:Alice_quantum_encoder}
    \end{align}
    Alice transmits the systems $A^n$ through $n$ uses of the quantum channel, leading to the following state shared between the reference, Alice, Bob, and Willie:
    \begin{align}
        \ket{\tau}_{RMB^nW^n} = \frac{1}{\sqrt{T}}\sum_{m=0}^{T-1}\ket{m}_R\otimes\ket{m}_{M}\otimes\ket{\phi_{m}}_{B^nW^n}
    \end{align}
    where $\ket{\phi_{m}}_{B^nW^n}=V_{A\to BW}^{\otimes n}\ket{\phi_{m}}_{A^n}$ is the channel output given the input $\ket{\phi_{m}}_{A^n}$, and $V_{A\to BW}$ is associated with a Stinespring representation of the channel $\mathcal{N}_{A\to B}$.
   
    Based on Lemma~\ref{Lemma:Expergated_classical_code}, there exists a decoding POVM
    $\{\Lambda^{(m, \ell)}_{B^n}\}$ such  that Bob’s probability of decoding success in the covert-secret classical code satisfies
    \begin{align}
        \forall m,\ell:
        {\Tr{\Lambda^{(m, \ell)}_{B^n}\sigma_{x^n(m,\ell)}}} 
        \geq 1 - e^{-\zeta_n^{(1)}\gamma_n\sqrt{n}}.
        \label{Equation:Bob_decoding_success_for_quantum_code}
    \end{align}
    We may construct a coherent version of this POVM
    for Bob, % 
    using the following isometry 
    \begin{align}
       D_{B^n\to B^n \widehat{M} \widehat{L}}\equiv
        \sum_{m\in\mathcal{M},\ell\in\mathcal{L}}
        {\sqrt{\Lambda^{(m, \ell)}_{B^n}} \otimes \ket{m}_{\widehat{M}} \otimes \ket{\ell}_{\widehat{L}}}
        \label{Equation:Bob_coherent_POVM}
    \end{align}
    (see definition of coherent POVM in \cite[Sec. 5.4]{W2017}).
    The resulting shared state is
    \begin{align}
        &\ket{\tau}_{RMB^nW^n\widehat{M}\widehat{L}} =
        (\identity\otimes D_{B^n\to B^n \widehat{M} \widehat{L}} \otimes \identity)\ket{\tau}_{RMB^nW^n}
        \nonumber \\
        &= \frac{1}{\sqrt{T}}\sum_{m=0}^{T-1}\frac{1}{\sqrt{\abs{\mathcal{L}}}} \sum_{\ell\in\mathcal{L}}\sum_{\substack{m'\in\mathcal{M},\\ \ell'\in\mathcal{L}}}\ket{m}_R\otimes\ket{m}_{M}
        \nonumber\\
        &\phantom{=}\otimes \left(\sqrt{\Lambda^{(m',\ell')}_{B^n}} \otimes \identity_{W^n}\right)e^{it(m,\ell)}\ket{x^n(m,\ell)}_{B^nW^n} \nonumber\\
        &\phantom{=}\otimes \ket{m'}_{\widehat{M}} \otimes \ket{\ell'}_{\widehat{L}}
        \label{Equation:Bob_encoded_shared_state}
    \end{align}
    where % 
    $\ket{x^n(m,\ell)}_{B^nW^n}\equiv V^{\otimes n}_{A\to BW}\ket{x^n(m,\ell)}_{A^n}$. 

    \subsubsection{State Approximation}
    Next, we approximate Bob's decoded state as follows.
    Intuitively, Alice sends the ``message'' pair $(m,\ell)$ to Bob. In the approximated state, the pair is recovered at Bob's registers  $(\widehat{M},\widehat{L})$ as well.
    \begin{lemma}
        \label{Lemma:quantum_code_1st_approx}
        Let $\ket{\tau}_{RMB^nW^n\widehat{M}\widehat{L}}$ be the state produced by Bob's decoder  in 
        \eqref{Equation:Bob_encoded_shared_state}.
        Then, there exist phases $\{h(m,\ell)\}$ such that $\ket{\tau}_{RMB^nW^n\widehat{M}\widehat{L}}$ can be approximated by the following state, 
                \begin{align}
            &\ket{\eta}_{RMB^nW^n\widehat{M}\widehat{L}} 
            \nonumber\\
            &= \frac{1}{\sqrt{T}}\sum_{m=0}^{T-1}\frac{1}{\sqrt{\abs{\mathcal{L}}}} \sum_{\ell\in\mathcal{L}}\ket{m}_R\otimes\ket{m}_{M}
            \nonumber\\
            &\phantom{=}\otimes e^{ih(m,\ell)}\ket{x^n(m,\ell)}_{B^nW^n} \otimes \ket{m}_{\widehat{M}} \otimes \ket{\ell}_{\widehat{L}}
            \,.
            \label{Equation:quantum_protocol_state_after_1st_approx_lemma}
        \end{align}
       Specifically,
        \begin{align}
            &\TrNorm{\ketbra{\tau}_{RMB^nW^n\widehat{M}\widehat{L}} - \ketbra{\eta}_{RMB^nW^n\widehat{M}\widehat{L}}} 
            \nonumber\\
            &\leq 2\sqrt{2}e^{-\frac{1}{2}\zeta_n^{(1)}\gamma_n\sqrt{n}}
            \label{Equation:quantum_protocol_bound_from_reliability}
        \end{align}
        where
 $\zeta_n^{(1)}$ is as in 
 Lemma~\ref{Lemma:Expergated_classical_code}.
        \end{lemma}
        The proof of Lemma~\ref{Lemma:quantum_code_1st_approx} is given in % 
        Section~\ref{Appendix:Quantum_Code_First_Approximation} in the Supplementary Material.
        Notice that, in the  expression of the approximation on the right-hand side of \eqref{Equation:quantum_protocol_state_after_1st_approx_lemma}, the systems $M$ and $\widehat{M}$ are in the same state $\ket{m}$ in each term.
    Observe further that the state approximation can also be expressed as % 
    \begin{align}
        &\ket{\eta}_{RM\widehat{M}B^nW^n\widehat{L}} 
        \nonumber\\
        &= \frac{1}{\sqrt{T}}\sum_{m=0}^{T-1}\ket{m}_R\otimes\ket{m}_{M} \otimes \ket{m}_{\widehat{M}} \otimes \ket{{\eta}_m}_{B^nW^n\widehat{L}}
        \label{Equation:shared_state_after_first_approximation}
    \end{align}
    where 
    \begin{align}
        \ket{{\eta}_m}_{B^nW^n\widehat{L}} \equiv \frac{1}{\sqrt{\abs{\mathcal{L}}}} \sum_{\ell\in\mathcal{L}}e^{ih(m,\ell)}\ket{x^n(m,\ell)}_{B^nW^n} \otimes \ket{\ell}_{\widehat{L}}
        \,.
        \label{Equation:eta_m}
    \end{align}

    \subsubsection{Decoupling}
    We use the secrecy property of our covert secrecy code from Lemma~\ref{Lemma:Expergated_classical_code} in order to show that Willie shares negligible % 
    correlation with Alice and Bob.

    Let $\rho_{W^n}^{(m)}$ be Willie's output state in the \emph{secrecy} coding scheme, conditioned on a particular classical message $m$, while identifying the message set $\mathcal{M}$ with $\{0,1,\dots,T-1\}$. 
    We now show that the approximated state $\ket{{\eta}_m}_{B^nW^n\widehat{L}}$ in \eqref{Equation:eta_m} is a purification of $\rho_{W^n}^{(m)}$:
        \begin{align}
            &\Tr_{B^n\widehat{L}}\left\{\ketbra{{\eta}_m}_{B^nW^n\widehat{L}}\right\} % 
            \nonumber\\
            &= \frac{1}{\abs{\mathcal{L}}}\sum_{\ell,\ell'\in\mathcal{L}}\Tr_{B^n\widehat{L}}\left\{e^{i[h(m,\ell)-h(m,\ell')]}\right.
            \nonumber\\
            &\left.\phantom{=======} \times \ketbra{x^n(m,\ell)}{x^n(m,\ell')}_{B^nW^n} \otimes \ketbra{\ell}{\ell'}_{\widehat{L}}\right\} \nonumber \\
            &\stackrel{(a)}{=} \frac{1}{\abs{\mathcal{L}}}\sum_{\ell\in\mathcal{L}}\Tr_{B^n}\left\{\ketbra{x^n(m,\ell)}_{B^nW^n}\right\} \nonumber \\
            &\stackrel{(b)}{=} \frac{1}{\abs{\mathcal{L}}}\sum_{\ell\in\mathcal{L}}\Tr_{B^n}\left\{\mathcal{P}_{X\to BW}^{\otimes n}\left( x^n(m,\ell)\right)\right\} \nonumber \\
            &= \frac{1}{\abs{\mathcal{L}}}\sum_{\ell\in\mathcal{L}}\rho_{W^n}^{(m,\ell)} \nonumber \\
            &= \rho_{W^n}^{(m)}
            \label{Equation:Willie_approx_reduced_state_pem_m}
        \end{align}
    where $(a)$ follows from the cyclic property of the trace and $(b)$ holds by \eqref{Eq:P_U_relation}. % 

    Based on the secrecy property \eqref{Equation:expurgated_secrecy}, we have that $\rho_{W^n}^{(m)}$ is close in trace distance to $\omega_{\alpha_n}^{\otimes n}$:
\begin{align}
        &\TrNorm{\rho^{(m)}_{W^n} - \omega_{\alpha_n}^{\otimes n}} 
        \leq 
        e^{-\zeta_n^{(3)}\gamma_n^{\frac{3}{2}} n^{\frac{1}{4}}}
\end{align}
for all $m\in\{0,1,\dots,T-1\}$,
where $\zeta_n^{(3)}$ is as in Lemma~\ref{Lemma:Expergated_classical_code}.
    Then,  by Uhlmann's theorem, Theorem~\ref{Theorem:Uhlmann_theorem},  there exists an isometry $\Gamma^m_{B^n\widehat{L} \to \breve{W}^n}$ such that
    \begin{multline}
        \left\|\left(\identity_{W^n} \otimes \Gamma^m_{B^n\widehat{L} \to \breve{W}^n}\right)\ketbra{{\eta}_m}_{B^nW^n\widehat{L}}\right.
        \\
        \left.\times\left(\identity_{W^n} \otimes \Gamma^m_{B^n\widehat{L} \to \breve{W}^n}\right)^\dagger - \ketbra{\omega_{\alpha_n}}_{W \breve{W}}^{\otimes n}
        \right\|_1 % 
        \\
        \leq 
        2e^{-\frac{1}{2}\zeta_n^{(3)}\gamma_n^{\frac{3}{2}} n^{\frac{1}{4}}}
        \label{Equation:Uhlmann_Gamma}
    \end{multline}
    where 
         $\ket{\omega_{\alpha_n}}_{W \breve{W}}$ is a purification  of $\omega_{\alpha_n}$. 

    We note that the partial isometry $\Gamma^m_{B^n\widehat{L}\to \breve{W}^n}$ is constructed according to Remark~\ref{Remark:Uhlmann_theorem_isometry_specific_construction}, identifying 
    $\ket{\psi}_{AB}$ and $\ket{\theta}_{AC}$
    with 
    $  \ket{\eta_m}_{B^n W^n \widehat{L}}$ and $\ket{\omega_{\alpha_n}}_{W\breve{W}}^{\otimes n}$, respectively
    (with $A \leftarrow W^n$, $B \leftarrow B^n\widehat{L}$, and $C \leftarrow \breve{W}^n$). The isometry is obtained from the polar decomposition of the operator $Q_m: \mathcal{H}_{B^n\widehat{L}} \to \mathcal{H}_{\breve{W}^n}$, defined by
    \begin{align}
      Q_m \equiv \Tr_{W^n}\!\left( \ket{\omega_{\alpha_n}}_{W\breve{W}}^{\otimes n}\!\bra{\eta_m}_{B^n W^n \widehat{L}} \right)\,.
    \end{align}

    Suppose Bob performs the following decoupling controlled isometry on his systems $B^n$, $\widehat{M}$, and $\widehat{L}$,
    \begin{align}
        \Delta_{\widehat{M}B^n\widehat{L} \to \widehat{M}\breve{W}^n} \equiv \sum_m \ketbra{m}_{\widehat{M}} \otimes \Gamma^m_{B^n\widehat{L} \to \breve{W}^n}
        \,.
        \label{Equation:Bob_decoupling_isometry}
    \end{align}
    The state approximation becomes
    \begin{align}
        \ket{\eta}_{RM\widehat{M}\breve{W}^n W^n} 
        &=\frac{1}{\sqrt{T}}\sum_{m=0}^{T-1}\ket{m}_R\otimes\ket{m}_{M} \otimes \ket{m}_{\widehat{M}} 
        \nonumber\\
        &\phantom{=}\otimes \left(\identity_{W^n} \otimes \Gamma^m_{B^n\widehat{L} \to \breve{W}^n}\right)\ket{{\eta}_m}_{B^nW^n\widehat{L}}
        \,.
    \end{align} 
    Then, by \eqref{Equation:Uhlmann_Gamma}, we find that this state is close to a decoupled state:
        \begin{multline}
        \left\|\ketbra{\eta}_{RM\widehat{M}\breve{W}^nW^n} 
        \right.
        \\
        \left.- 
        \ketbra{\mathrm{GHZ}}_{RM\widehat{M}} \otimes \ketbra{\omega_{\alpha_n}}^{\otimes n}_{W\breve{W}}
        \right\|_1
        \\
        \leq 
        2e^{-\frac{1}{2}\zeta_n^{(3)}\gamma_n^{\frac{3}{2}} n^{\frac{1}{4}}}
        \label{Equation:quantum_protocol_bound_from_secrecy}
    \end{multline}
where
   \begin{align}
        \ket{\mathrm{GHZ}}_{RM\widehat{M}% 
        } = \frac{1}{\sqrt{T}}\sum_{m=0}^{T-1}\ket{m}_R\otimes\ket{m}_{M} \otimes \ket{m}_{\widehat{M}}  % 
    \end{align}
by the telescoping property of the trace distance (see \cite[Ex. 9.1.3]{W2017}).    Roughly speaking, this shows that the output state is $\approx {\mathrm{GHZ}}_{RM\widehat{M}} \otimes \omega_{\alpha_n}^{\otimes n}$.
    Thus,  $W^n\breve{W}^n$ is effectively decoupled from  $R$, $M$, and $\widehat{M}$, hence Bob can simply trace out  $\breve{W}^n$, leaving us with the GHZ state. % 

\subsubsection{From GHZ to Bipartite Entanglement\label{Step:GHZ2bi}}
Recall that we have assumed that Alice and Bob share a free classical link. 
    Alice and Bob can use this link 
    in order to convert the GHZ state into a maximally entangled bipartite state. Specifically, if Alice applies the Fourier transform unitary from
    \eqref{Equation:QFT}
    to  $M$ and then measures  in the computational basis, this results in a post-measurement state 
    \begin{align}
        &\ket{\upsilon}_{RM\widehat{M}} 
        \nonumber\\
        &= \left(\frac{1}{\sqrt{T}}\sum_{m=0}^{T-1}\exp{2\pi im j / T}\ket{m}_R \otimes \ket{m}_{\widehat{M}}\right) \otimes\ket{j}_{M}
    \end{align}
    where $j$ is the measurement outcome.
    Suppose 
    Alice sends Bob the measurement outcome $j$ over an undetectable classical channel. Then, if Bob applies $\mathsf{Z}^{-j}$ to % 
    $\widehat{M}$ % 
    (see \eqref{Equation:phase_shift_unitary}), this yields the desired state $\ket{\Phi}_{R\widehat{M}}$.
    Effectively, the Fourier transform transfers the GHZ correlations into phase information on Alice's system. When Alice measures and obtains outcome $j$, the joint state of Bob's and the reference systems ($\widehat{M}$ and $R$, respectively) collapses into a superposition state with phases that depend on $j$. Bob can then cancel this phase by applying the inverse phase-shift operation $Z^{-j}$, where the index $j$ is sent to him by Alice, thereby recovering the maximally entangled state.

   Overall, we bound the approximation error by 
    \begin{align}
        &\TrNorm{\Phi_{R\widehat{M}} - \tau_{R\widehat{M}}} % 
        \nonumber\\
        &\stackrel{(a)}{\leq} \TrNorm{\Phi_{R\widehat{M}} - \eta_{R\widehat{M}}} + \TrNorm{\eta_{R\widehat{M}} - \tau_{R\widehat{M}}} \nonumber \\
        &\stackrel{(b)}{\leq} 
        \left\|\ketbra{\mathrm{GHZ}}_{RM\widehat{M}} \otimes \ketbra{\omega_{\alpha_n}}^{\otimes n}_{W\breve{W}}
        \right.
        \nonumber\\
        &\left.\phantom{=}- \ketbra{\eta}_{RM\widehat{M}\breve{W}^nW^n}
        \right\|_1 % 
        \nonumber\\
        &\phantom{\leq}+ 
        \TrNorm{\ketbra{\eta}_{RMB^nW^n\widehat{M}\widehat{L}} - \ketbra{\tau}_{RMB^nW^n\widehat{M}\widehat{L}}} \nonumber \\
        &\stackrel{(c)}{\leq} 
        2e^{-\frac{1}{2}\zeta_n^{(3)}\gamma_n^{\frac{3}{2}} n^{\frac{1}{4}}}
        + 2\sqrt{2}e^{-\frac{1}{2}\zeta_n^{(1)}\gamma_n\sqrt{n}} \nonumber \\
        &\leq 
        e^{-\frac{1}{4}\zeta_n^{(1)}\gamma_n \sqrt{n}}
    \end{align}
    for sufficiently large $n$, where $(a)$ follows from the triangle inequality, % 
    $(b)$  from the monotonicity of the trace distance, and $(c)$  from the approximation bounds  \eqref{Equation:quantum_protocol_bound_from_reliability} and \eqref{Equation:quantum_protocol_bound_from_secrecy}.
    Finally, by the Fuchs-van de Graaf Inequalities \cite[Th. 9.3.1]{W2017}, % 
   the bound on the trace distance implies a corresponding bound on the fidelity, % 
    \begin{align}
        F\left(\Phi_{RM}\,,\; \tau_{R\widehat{M}} \right) \geq 1 - e^{-\Tilde{\zeta}_n\gamma_n \sqrt{n}}
        \,,
        \label{Equation:final_fidelity}
    \end{align}
    with $\Tilde{\zeta}_n \in
    \omega\left((\log n)^{-\frac{4}{3}}
    n^{-\frac{1}{3}}\right)$. 
    This completes the achievability derivation based on covert LOCC assistance.
\subsubsection{Eliminating the Assistance\label{Step:NoAssistance}}% 
    Now, we show that we do not actually need classical communication assistance for covert entanglement generation.

    Let $\mathsf{G}_{M \to M M'}$ be an isometry that ``copies'' the content of $M$ to $M'$ in the computational basis, implemented by simply adding an ancilla $M'$ initialized in the $\ket{0}_{M'}$ state, and then applying   
    $\mathsf{CNOT}_{M M'\to M M'}$ (a repetition encoder).
    First, we notice 
    that all of Alice's encoding actions can be represented as a single encoding map of the form
    \begin{align}
        &\Tilde{\mathcal{F}}^j_{M \to A^n} 
        \equiv 
        \bra{j}_M  \mathsf{F}_M U_{M' \to A^n} \mathsf{G}_{M \to M M'} \nonumber \\
        &= (\bra{j}_M \otimes \identity_{A^n}) \left(\sum_{\Tilde{j},\Tilde{m}=0}^{T-1}e^{2\pi i \Tilde{j} \Tilde{m} /T}\ketbra{\Tilde{j}}{\Tilde{m}}_M\otimes \identity_{A^n}\right) 
        \nonumber\\
        &\phantom{=}\times
        \left(\identity_M\otimes \sum_{m'=0}^{T-1} \ket{\phi_{m'}}_{A^n}\bra{m'}_{M'}\right) 
        \nonumber\\
        &\phantom{=}\times
        \left(
        \sum_{m=0}^{T-1}\ketbra{m}_M \otimes \ket{m}_{M'}
        \right)
    \end{align}
    where 
    $\mathsf{F}_M$ is the quantum Fourier transform, and $\bra{j}_M$ represents the projection onto a particular measurement outcome $j$.
    When scaling Alice's encoding map by $\sqrt{T}$, it can be simplified to the following isometry
    \begin{align}
        \hat{\mathcal{F}}^j_{M \to A^n} 
        = \sqrt{T}\Tilde{\mathcal{F}}^j_{M \to A^n} 
        &=\sum_{m=0}^{T-1}e^{2\pi i j m /T} \ket{\phi_m}_{A^n}\bra{m}_M
        \;.
    \end{align}
    On the other hand, Bob has a corresponding decoding map of the form
    \begin{align}
        \hat{\mathcal{D}}^j_{B^n \to \widehat{M}} \equiv \mathsf{Z}^{-j}_{\widehat{M}}\left(\trace_{\breve{W}^n}\circ\Delta_{\widehat{M}B^n\widehat{L} \to \widehat{M}\breve{W}^n}\right)D_{B^n\to B^n \widehat{M} \widehat{L}}
    \end{align}
    where $D_{B^n\to B^n \widehat{M} \widehat{L}}$ is Bob's coherent POVM in \eqref{Equation:Bob_coherent_POVM}, $\left(\trace_{\breve{W}^n}\circ\Delta_{\widehat{M}B^n\widehat{L} \to \widehat{M}\breve{W}^n}\right)$ is the concatenation of Bob's decoupler isometry in \eqref{Equation:Bob_decoupling_isometry} and the tracing out of the system $\breve{W}^n$ afterwards, and $\mathsf{Z}_{\widehat{M}}$ is the Heisenberg-Weyl phase-shift unitary. % 
    Because all of the components of Bob's decoding map $\hat{\mathcal{D}}^j_{B^n \to \widehat{M}}$ are isometries, then the decoding map itself is also an isometry.
    Next, we define the following channel,
    \begin{align}
        \mathcal{E}^j_{M \to \widehat{M}} 
        &\equiv \hat{\mathcal{D}}^j_{B^n \to \widehat{M}} \circ \mathcal{N}_{A \to B}^{\otimes n} \circ \hat{\mathcal{F}}^j_{M \to A^n}
        \,.
    \end{align}
    We can then represent Bob's estimated shared state with the resource as
    \begin{align}
        \tau_{R\widehat{M}} = \left(\mathrm{id}_R \otimes 
        \frac{1}{T}\sum_{j=0}^{T-1}
        \mathcal{E}^j_{M \to \widehat{M}}
        \right)\left(\ketbra{\Phi}_{RM}\right)
        \,.
    \end{align}
    Thus, the fidelity between the actual state and the estimated state can be expressed as
    \begin{align}
        &F\left(\Phi_{RM}\,,\; \tau_{R\widehat{M}} \right)
        = \bra{\Phi}_{RM}\tau_{R\widehat{M}}\ket{\Phi}_{RM} \nonumber \\
        &= \frac{1}{T}\sum_{j=0}^{T-1}\bra{\Phi}_{RM}\left(\mathrm{id}_R \otimes   \mathcal{E}^j_{M \to \widehat{M}}
        \right)\left(\ketbra{\Phi}_{RM}\right)\ket{\Phi}_{RM} \nonumber \\
        &\geq 1 - e^{-\Tilde{\zeta}_n\gamma_n \sqrt{n}}
    \end{align}
    where the inequality follows \eqref{Equation:final_fidelity}. Hence, we deduce that at least one of the encoder–decoder pairs $\left(\hat{\mathcal{F}}^j_{M \to A^n}, \hat{\mathcal{D}}^j_{B^n \to \widehat{M}}\right)$ achieves arbitrarily high fidelity. Consequently, Alice and Bob can prearrange to employ this particular scheme.
    \subsubsection{Covertness}
    We now proceed to show covertness. % 
     Willie's received state is given by
    \begin{align}
        {\tau}_{W^n} = \Tr_{RMB^n\widehat{M}\widehat{L}}\left\{\ketbra{\tau}_{RMB^nW^n\widehat{M}\widehat{L}}\right\}
    \end{align}
    (see \eqref{Equation:Bob_encoded_shared_state}).
    We note that the phase that we have added in Steps~\ref{Step:GHZ2bi} and \ref{Step:NoAssistance} does not affect Willie's reduced state, and can thus be ignored.
    We have shown that the actual output state can be approximated by the state
    $\ket{\eta}_{RMB^nW^n\widehat{M}\widehat{L}}$ in Lemma~\ref{Lemma:quantum_code_1st_approx}.
    Furthermore, the reduced state $\eta_{W^n}$ is identical to Willie's state in the classical secrecy setting:
    \begin{align}
        \eta_{W^n} 
        &\equiv \Tr_{RM\widehat{M}B^n\widehat{L}}\left\{\ketbra{\eta}_{RM\widehat{M}B^nW^n\widehat{L}}\right\} \nonumber \\
        &\stackrel{(a)}{=} \frac{1}{T}\sum_{m=0}^{T-1}\rho_{W^n}^{(m)} \nonumber \\
        &\stackrel{(b)}{=} \overline{\rho}_{W^n}
        \label{Equation:Willie_approx_state_equal_to_average_state}
    \end{align}
   where $(a)$ % 
   follows from  \eqref{Equation:shared_state_after_first_approximation} and \eqref{Equation:Willie_approx_reduced_state_pem_m}, and $(b)$ from the definition of $ \overline{\rho}_{W^n}$ in \eqref{Equation:Willie_Average_Output_Secrecy}.
    Therefore,
    \begin{align}
        \TrNorm{\eta_{W^n}% 
        - \omega_{\alpha_n}^{\otimes n}} 
        &= \TrNorm{\frac{1}{T}\sum_{m=0}^{T-1}\rho^{(m)}_{W^n} - \omega_{\alpha_n}^{\otimes n}} \nonumber \\
        &\leq \frac{1}{T}\sum_{m=0}^{T-1}\TrNorm{\rho^{(m)}_{W^n} - \omega_{\alpha_n}^{\otimes n}} \nonumber \\
        &\leq 
        e^{-\zeta_n^{(3)}\gamma_n^{\frac{3}{2}} n^{\frac{1}{4}}}
        \label{Equation:Willie_average_state_bound_for_quantum_covertness}
    \end{align}
    by the triangle inequality  and the secrecy property % 
    \eqref{Equation:expurgated_secrecy}.
    Thus, by the triangle inequality % 
    \begin{align}
        \TrNorm{% 
        {\tau}_{W^n} - \omega_{\alpha_n}^{\otimes n}} 
        &\leq \TrNorm{% 
        {\tau}_{W^n} - \eta_{W^n}} + \TrNorm{\eta_{W^n} - \omega_{\alpha_n}^{\otimes n}} \nonumber \\
        &\leq 2\sqrt{2}e^{-\frac{1}{2}\zeta_n^{(1)}\gamma_n\sqrt{n}} + 
        e^{-\zeta_n^{(3)}\gamma_n^{\frac{3}{2}} n^{\frac{1}{4}}}
        \nonumber \\
        &\leq 
        e^{-\frac{1}{4}\zeta_n^{(1)}\gamma_n \sqrt{n}}
        \label{Equation:Willie_actual_state_bound_for_quantum_covertness}
    \end{align}
    for sufficiently large $n$,
    where % 
    the second inequality follows from % 
    \eqref{Equation:quantum_protocol_bound_from_reliability}, \eqref{Equation:Willie_average_state_bound_for_quantum_covertness} and trace monotonicity.
    By the continuity property in  Lemma~\ref{Lemma:bound_on_rel_entropy_diff}, we now have
    \begin{align}
        &\abs{\RelEntropy{% 
        {\tau}_{W^n}}{\omega_0^{\otimes n}} - \RelEntropy{\omega_{\alpha_n}^{\otimes n}}{\omega_0^{\otimes n}}} 
        \nonumber\\
        &\leq \TrNorm{% 
        {\tau}_{W^n} - \omega_{\alpha_n}^{\otimes n}}
        \nonumber\\
        &\phantom{=}\times\left(n\log{\left(\frac{4\dim{d_W}}{(\lambda_{\min}(\omega_0))^3}\right) -\log{\left(\TrNorm{% 
        {\tau}_{W^n} - \omega_{\alpha_n}^{\otimes n}}\right)}}\right) 
        \nonumber \\
        &\leq e^{-\frac{1}{4}\zeta_n^{(1)}\gamma_n \sqrt{n}}\left(n\log{\left(\frac{4\dim{d_W}}{(\lambda_{\min}(\omega_0))^3}\right) + \frac{1}{4}\zeta_n^{(1)}\gamma_n \sqrt{n}}\right) \nonumber \\
        &\leq e^{-\frac{1}{8}\zeta_n^{(1)}\gamma_n \sqrt{n}}
    \end{align}
    where the second inequality follows from \eqref{Equation:Willie_actual_state_bound_for_quantum_covertness} as the function $f(x) = -x\log{x}$ is monotonically increasing for $x \in (0, e^{-1})$. % 
    Thus, for the sequence $\Tilde{\zeta}_n$, we have
    \begin{align}
        \abs{\RelEntropy{% 
        {\tau}_{W^n}}{\omega_0^{\otimes n}} - \RelEntropy{\omega_{\alpha_n}^{\otimes n}}{\omega_0^{\otimes n}}} 
        \leq e^{-\Tilde{\zeta}_n\gamma_n \sqrt{n}}
        \label{eq:EG_covertness_result}
    \end{align}
    as stated in Proposition~\ref{Proposition:Achievability_Covert_Entanglement_Generation_No_LOCC}. 
     \end{proof}

\label{Subsection:EG_rate_analysis}
We have thus shown achievability of the following entanglement-generation rate, 
\begin{align}
    &\frac{\log{\left[\mathrm{dim}(\mathcal{H}_M)\right]}}{\sqrt{n D(% 
    \tau_{W^n}||\omega_0^{\otimes n})}}
    \nonumber\\
    &=
    \frac{\gamma_n\sqrt{n}
    \left[(1 -2\zeta_n)\RelEntropy{\sigma_1}{\sigma_0}
    -% 
    (1 + \zeta_n)\RelEntropy{\omega_1}{\omega_0}
    \right]_+}{\sqrt{n D(% 
    \tau_{W^n}||\omega_0^{\otimes n})}}
    \,.
    \label{Equation:EG_Rate_Achievability}
\end{align}
Observe that
\begin{align}
    \RelEntropy{% 
    {\tau}_{W^n}}{\omega_0^{\otimes n}}  
    \leq \RelEntropy{\omega_{\alpha_n}^{\otimes n}}{\omega_0^{\otimes n}} + e^{-\Tilde{\zeta}_n\gamma_n \sqrt{n}}
    \label{Equation:EG_Covertness_Divergence}
\end{align}
by \eqref{eq:EG_covertness_result}, and 
\begin{align}
    D(\omega_{\alpha_n}^{\otimes n}||\omega_0^{\otimes n})
    = \frac{1}{2}\gamma_n^2\chi^2(\omega_1||\omega_0) + 
    {O}\left(\frac{\gamma_n^3}{\sqrt{n}}\right)
    \label{Equation:EG_Covert_State_Divergence}
\end{align}
by \cite[Lemma 5]{9344627}
(see Lemma~\ref{Lemma:avg_QRE_to_chi_square} in the Supplementary Material).
Therefore, by plugging
\eqref{Equation:EG_Covertness_Divergence}-\eqref{Equation:EG_Covert_State_Divergence} into \eqref{Equation:EG_Rate_Achievability}, we deduce that
the asymptotic covert entanglement-generation rate satisfies
\begin{align}
    L_{\text{EG}} % 
    \geq \frac{\left[D(\sigma_1||\sigma_0)-D(\omega_1||\omega_0)\right]_+}{\sqrt{\frac{1}{2} \chi^2(\omega_1||\omega_0)}}
    \,.
\end{align}

\subsection{Converse Proof}

The converse part is a straightforward consequence 
from our covert secrecy result.
Suppose Alice prepares a maximally entangled state, locally,
\begin{align}
    \ket{\Phi}_{M \breve{M}}=\frac{1}{\sqrt{T}}
    \sum_{m=0}^{T-1} \ket{m}_{M}\otimes \ket{m}_{\breve{M}}
\end{align}
where $T=2^{\sqrt{n\delta_n} L_{\text{EG}}}$, 
$M$ is a resource that she keeps, and $\breve{M}$ is the resource that she would like to distribute to Bob.
 If Alice is active, then  she applies an encoding map 
$\mathcal{F}^{(n)}_{\breve{M}\rightarrow A^n}$ 
on her ``quantum message'' $\breve{M}$, which results in % 
\begin{align}
    \tau_{MA^n} =\left(\mathrm{id}_M\otimes \mathcal{F}^{(n)}_{\breve{M}\to A^n} \right)\left(\ketbra{\Phi}_{M\breve{M}}\right)
    \,,
\end{align}
and % 
transmits the encoded system $A^n$. % 
The   output state is thus
\begin{align}
    \tau_{M B^n W^n} =\left(\mathrm{id}_M\otimes \mathcal{U}_{A\to BW}^{\otimes n} \right)\left(\tau_{M A^n}\right)
    \,.
\end{align}
Bob and Willie receive 
$B^n% 
$ and
$W^n% 
$,
respectively. 
Bob then % 
performs a decoding operation $\mathcal{D}^{(n)}_{B^n\to\widehat{M}}$, % 
which recovers a state 
\begin{align}
    \tau_{M\widehat{M} W^n}% 
    =
    \left(\mathrm{id}_M\otimes \mathcal{D}^{(n)}_{B^n\to\widehat{M}} \otimes \mathrm{id}_{W^n}\right)\left(\tau_{MB^nW^n}\right)
    \,.
\end{align}
Meanwhile, Willie receives $W^n$ in the reduced state
\begin{align}
    {\tau}_{W^n} = 
    \Tr_{MB^n}\left( \tau_{MB^nW^n}\right)
    \,.
\end{align}

Consider a sequence of covert entanglement-generation codes such that the reliability and covertness criteria hold, i.e., % 
\begin{align}
    &\norm{ \tau_{M\widehat{M} W^n} - \ketbra{\Phi}_{M \breve{M}}\otimes \tau_{W^n} }_1 \leq \varepsilon_n
    \label{Equation:Converse_EG_Distance}
\end{align}
and
\begin{align}
&\RelEntropy{\tau_{W^n}}{\omega_0^{\otimes n}} \leq \delta_n
\label{Equation:Converse_EG_Covertness_criterion}
\end{align}
where $\varepsilon_n$ and $\delta_n$ tend to zero as $n\to\infty$.

First, we observe that by trace monotonicity,
 \begin{align}
    \TrNorm{\tau_{\widehat{M}} -  \frac{1}{T}\identity_{\breve{M}}} \leq \varepsilon_n, 
    \label{Equation:EG_converse_rate_reducing_to_classical_bounds}
\end{align}
which reduces to the classical error  requirement for the transmission % 
of a classical message with zero public information rate.
The condition in 
\eqref{Equation:Converse_EG_Covertness_criterion}
fulfils the covertness criterion.
Furthermore, by tracing out Bob's system, we obtain
\begin{align}
\TrNorm{\tau_{M W^n} - \frac{1}{T}\identity_{M}\otimes \tau_{W^n}} \leq \varepsilon_n, 
\end{align}
which implies weak secrecy, i.e.,
\begin{align}
    I(M;W^n)_\tau 
    &\leq 3% 
    {\varepsilon_n}\log{T} + 2(1+% 
    {\varepsilon_n})h_2\left(\frac{% 
    {\varepsilon_n}}{1+% 
    {\varepsilon_n}} \right)
    \nonumber\\
    &\equiv \beta_n \log T
\end{align}
by the Alicki-Fannes-Winter inequality \cite[Lemma 2]{Winter:16p}, where 
$h_2(x)$ is the binary entropy
and 
$\beta_n$ tends to zero as $n\to\infty$. % 
We note that  the converse part of the covert secrecy capacity holds under  weak secrecy as well, as can be seen in Subsection~\ref{Subsection:Converse:Covert_Secrecy_Capacity_No_Key} in the Supplementary Material 
(see \eqref{Eq:Weak_Secrecy}).

We deduce that the same encoder-decoder pair enables the transmission of classical messages correctly, covertly and secretly.
Hence, if $L_{\text{EG}}$ is achievable for entanglement generation, then 
the same rate is achievable for the transmission of secret classical messages.
Therefore, the entanglement-generation capacity is upper bounded by the classical secrecy capacity, i.e., % 
\begin{align}
    C_{\text{EG}}(\mathcal{N}) \leq
    C_{\text{S}}(\mathcal{P})=
    \frac{ \left[\RelEntropy{\sigma_1}{\sigma_0}
    -\RelEntropy{\omega_1}{\omega_0}
    \right]_+}{\sqrt{\frac{1}{2}\chi^2(\omega_1||\omega_0)}}
\end{align}
where the equality holds by 
Theorem~\ref{Theorem:Covert_Secrecy_Capacity_No_Key}.
\qed % 

\section{Summary and Discussion}
\label{Section:Discussion}

\subsection{Summary}
We consider  entanglement generation through covert communication, where the transmission itself should be hidden to avoid detection by an adversary.
The communication setting is depicted in  Figure~\ref{Figure:EG_Model}.
Alice makes a decision on whether to perform the communication task, or not. 
If Alice decides to be inactive, the channel input is $\ket{0}^{\otimes n}$, where $\ket{0}$ is the innocent state corresponding to a passive transmitter.
Otherwise, if she does perform the task,  she prepares a maximally entangled state 
$\Phi_{RM}$ locally, and
applies an encoding map  
on her ``quantum message'' $M$.
She then transmits the encoded system $A^n$ using $n$ instances of the
quantum channel. 
At the channel output, Bob and Willie receive 
$B^n$ and
$W^n$,
respectively. 
Bob 
performs a decoding operation on his received system, which recovers a state that is close to $% 
\Phi_{R\widehat{M}}$.
 Meanwhile, Willie receives $W^n$ and 
 performs a hypothesis test to determine whether Alice has transmitted information or not.

Our approach is fundamentally different from that in Anderson et al.~\cite{anderson2024covert,anderson2025achievability}.
First, we consider the combined setting of covert and secret communication of classical information via a classical-quantum channel (see Section~\ref{Section:Covert-Secret Communication Over Classical-Quantum Channels}).
One might argue that if covertness is achieved, secrecy becomes redundant, as Willie would not attempt to decode a message he does not detect. However, covertness is typically defined in a statistical sense, meaning that while the probability of detection is small, it is not necessarily zero. Thus, in rare cases where Willie does detect some anomalous activity, secrecy ensures that he still cannot extract meaningful information. Covert secrecy is thus a problem of independent interest, not merely an auxiliary result for the main derivation.

We determine the covert secrecy capacity with key assistance in Theorem~\ref{Theorem:Covert_Secrecy_Capacity}, and without assistance in Theorem~\ref{Theorem:Covert_Secrecy_Capacity_No_Key}. The key-assisted result can be viewed as the classical-quantum generalization of the result by Bloch \cite[Sec. VII-C]{Bloch16}.  
In addition, we determine the % 
minimal key rate required to achieve the key-assisted capacity.
Then, we use Devetak's approach \cite{Devetak05} of
constructing an entanglement-generation code from a secrecy code (without assistance).
This method utilizes secrecy to establish entanglement generation. As a result, we achieve the same covert entanglement-generation rate as the secret  classical information rate 
without assistance.
We note that, as opposed to Anderson et al.~\cite{anderson2025achievability}, 
our scheme does not require a pre-shared secret key.

 We show that approximately $ {\sqrt{n} C_\text{EG}}$ EPR pairs can be generated covertly. The optimal rate $ C_\text{EG}$, which is referred to as the covert capacity for entanglement generation, is given by the formula 
 \begin{align}
C_\text{EG}(\mathcal{N})=\frac{{ \left[\RelEntropy{\sigma_1}{\sigma_0}
    -\RelEntropy{\omega_1}{\omega_0}
    \right]_+}}{\sqrt{\frac{1}{2}\chi^2(\omega_1||\omega_0)}},
\end{align}
where $\sigma_0$ and $\omega_0$ are Bob and Willie's respective outputs for the ``innocent'' input $\ket{0}$, whereas $\sigma_1$ and $\omega_1$ are the outputs associated with  inputs that are orthogonal to $\ket{0}$; 
$D(\rho||\sigma)$ is the quantum relative entropy and 
$\chi^2(\rho||\sigma)$ is the $\chi^2$-relative entropy, which can be interpreted as the second derivative of the quantum relative entropy 
(see 
\cite[Eq.~4]{9344627},
 \cite[Sec.~1.1.4]{Zlotnick:24z}).
 We also note that, although this work focuses on finite-dimensional quantum channels, a companion paper \cite{anderson2025covertJSAC} extends our approach and analysis to infinite-dimensional bosonic channels. The framework developed here provides the conceptual starting point for that analysis: covert classical communication with secrecy is first established and then converted into a covert entanglement-generation protocol via a coherent coding argument. 
 This yields a lower bound on covert entanglement-generation capacity and demonstrates achievable scaling on the order of $\sqrt{n}$.
    While the resulting behavior is similar, the companion paper considers a specific infinite-dimensional model, namely the lossy thermal-noise bosonic channel, whereas our results apply to general finite-dimensional quantum channels.
We also emphasize that our formulation adopts a worst-case perspective on the warden's capabilities. By modeling communication as a quantum channel from Alice to Bob and granting Willie access to the entire environment, we impose no restrictions on his observation strategy beyond the laws of quantum mechanics. This contrasts with classical approaches, such as Bloch's model \cite{Bloch16}, which can be viewed as fixing a particular measurement and analyzing the resulting outcomes. By avoiding such assumptions, our framework effectively endows Willie with maximal power, and the resulting guarantees therefore hold against any physically admissible detection strategy.

We now discuss covert communication in the wider context of quantum Shannon theory, and consider 
the transmission of either classical or quantum information, with or without secrecy or covertness.
We use the capacity notation in Table~\ref{table:LcapacityNotation}, where the first column corresponds to covert communication, and the second column includes  non-covert capacities. The rows are ordered in decreasing value.

\begin{table}[tb]
\caption{Capacity Notation}
\begin{center}
\begin{tabular}{l|cc}
&$\;$   Covert $\;$ 
&       Non-covert $\;$    
\\	[0.2cm]   
\hline 
\\			[-0.2cm]
Classical messages				
&	$C_\text{Cl}(\mathcal{N})$
&   $C_\text{Cl}^0(\mathcal{N})$ 	
\\[0.3cm] 
Secret messages			
&	$C_\text{S}(\mathcal{N})$
&   $C_\text{S}^0(\mathcal{N})$ 	
\\[0.3cm]
Entanglement generation				
&	$C_\text{EG}(\mathcal{N})$
&   $C_\text{EG}^0(\mathcal{N})$
\\[0.3cm] 
Quantum information			
&	$C_\text{Q}(\mathcal{N})$
&   $C_\text{Q}^0(\mathcal{N})$ 					
\end{tabular}
\end{center}
\label{table:LcapacityNotation}
\end{table}

\subsection{Quantum Communication}
\label{Section:Q_Information} 
    Transmission of quantum information generalizes entanglement generation. In the general task of quantum communication, also known as quantum subspace transmission, the transmitter encodes an arbitrary state $\Psi_{RM}$, and the receiver decodes such that $\rho_{R\widehat{M}}\approx \Psi_{RM}$.
The reliability requirement in this case takes the following form:
\begin{equation}
    F( \rho_{R\widehat{M}} \,,\; \Psi_{RM}) \geq 1-\varepsilon
\,,\; \text{for all $\ket{\Psi}_{RM}\in\mathcal{H}_R\otimes \mathcal{H}_M$}
\label{Equation:Fidelity_Q_Information}
\end{equation}
with arbitrary reference space, $\mathcal{H}_R$.    
    In particular, the reduced states satisfy $\rho_{\widehat{M}}\approx \Psi_{M}$.
    In other words, Alice teleports a ``message state'' $\rho_M$ to Bob.
Entanglement generation (EG) involves a weaker requirement, as the reliability requirement need only hold for the state $\ket{\Psi}_{RM}=\ket{\Phi}_{RM}$, where      
    \begin{equation}
          \ket{\Phi}_{RM} \equiv \frac{1}{\sqrt{T}}\sum_{m=0}^{T-1} \ket{m}_R\otimes\ket{m}_M
    \end{equation}
    (cf.~\eqref{Equation:Fidelity_EG} and
\eqref{Equation:Fidelity_Q_Information}).     The quantum covert rate $L_\text{Q}$ is then defined in terms of the information dimension, $T=\mathrm{dim}(\mathcal{H}_M)$. Anderson et al.~\cite{anderson2024covert,anderson2025achievability} have recently developed lower bounds on the quantum covert rate using depolarizing channel codes. Here, we give a more refined characterization for covert entanglement generation in terms of the channel itself.

\subsection{Quantum Covert Capacity}
\label{Section:Q_Information_Rate}
In traditional coding problems, i.e., without the covertness requirement, the achievable number of entangled qubits pairs that can be generated is linear in the blocklength, $n$. Hence, the entanglement-generation rate is defined as $R_{\text{EG}}\equiv\frac{\log T}{n}$, where $T$ is the entanglement dimension.
An entanglement-generation rate $R_\text{EG} > 0$ is achievable (without covertness) if, for every $\varepsilon> 0$ and sufficiently large $n$, there exists a $(e^{n R_\text{EG}}, n, \varepsilon, \infty)$ code for entanglement generation. 

Similarly, the entanglement-generation capacity is the supremum of all achievable rates without covertness, and it is denoted by $C_{\text{EG}}^0(\mathcal{N})$, where the superscript '0` indicates that the communication is \emph{not} covert.
The quantum capacity $C_{\text{Q}}^0(\mathcal{N})$,  without covertness, is defined in a similar manner, but with respect to the coding task in the previous subsection 
(see Subsection~\ref{Section:Q_Information} % 
above). See the last two rows in Table~\ref{table:LcapacityNotation}.
As explained in Subsection~\ref{Section:Q_Information} above, entanglement generation involves a weaker requirement, restricted to $\ket{\Psi}_{RM}=\ket{\Phi}_{RM}$ alone.
    Based on these operational definitions, it is clear that the capacities satisfy
    \begin{align}
    {C}_\text{Q}^0(\mathcal{N})\leq {C}_\text{EG}^0(\mathcal{N})
    \,.
    \end{align}
    Furthermore, in non-covert communication, it is well known that     
    \begin{align}
    {C}_{\text{Q}}^0(\mathcal{N})={C}_{\text{EG}}^0(\mathcal{N})
    \end{align} 
        (see \cite{bjelakovic2009entanglement}).
    That is, the quantum capacity and the entanglement-generation capacity are identical.     The identity holds whether key assistance is available or not.
    In this sense, entanglement generation is equivalent to quantum subspace transmission. 
        This is analogous to the classical result that the common-randomness capacity is the same as the capacity for message transmission \cite{AhlswedeCsiszar:93p,Ahlswede2021}.

    In covert communication, we have 
    \begin{align}
    {C}_\text{Q}(\mathcal{N})\leq {C}_\text{EG}(\mathcal{N})
    \end{align}
    based on the same operational considerations.
    We observe that under covert LOCC, equality is achieved via the teleportation protocol. % 
    Nevertheless, we conjecture that equality holds in general in covert communication, % 
    i.e., ${C}_\text{Q}(\mathcal{N})={C}_\text{EG}(\mathcal{N})$. 

\subsection{Classical Communication}
In the transmission of classical information,  the coding rate is the fraction of classical information bits per channel use. In this task, the transmitter sends a classical message $\mathbf{m}\in\mathcal{M}$, and the receiver performs a decoding measurement which produces a (random) estimate $\widehat{\mathbf{m}}$ for the original message. The performance can be characterized in terms of the classical rate, $R_\text{Cl}\equiv\frac{\log|\mathcal{M}|}{n}$, and the
probability of error, $\Pr(\widehat{\mathbf{m}}\neq \mathbf{m})$. The classical capacities
${C}_\text{Cl}(\mathcal{N})$ and ${C}_\text{Cl}^0(\mathcal{N})$ are defined accordingly, with and without covertness, respectively. See the first row in Table~\ref{table:LcapacityNotation}.
While the tasks of classical and quantum information transmission are fundamentally different, we can compare the rates. 
Specifically, if we limit the error requirement \eqref{Equation:Fidelity_Q_Information} to the computational basis, taking   
        ${\Psi}_M=\ketbra{m}$ for $m\in \{0,\ldots, T-1\}$,   
    the task reduces to the transmission of classical information, hence,
    \begin{equation}
        {C}_\text{Q}(\mathcal{N})\leq {C}_\text{Cl}(\mathcal{N})
        \,.
    \label{Equation:Classical_vs_Quantum}
    \end{equation}
    We note that the inequalities above do not mean that classical communication is better than quantum communication. Furthermore, we are only comparing the rate values. Since the quantum rate is a fraction of information qubits per channel use, % 
    while the classical rate is the fraction of classical bits, % 
    it is not a fair comparison of performance, but rather a useful observation. 
     The bound in \eqref{Equation:Classical_vs_Quantum} should be expected, as the ability to send qubits implies the ability to send classical bits. This principle is often presented as a resource inequality \cite{DevetakHarrowWinter:08p}, % 
 \begin{align}
 [q\to q]\geq [c\to c]
 \end{align}
 where $[c\to c]$ and $[q\to q]$ represent the resources of a noiseless bit channel and a noiseless qubit channel, respectively. 

\subsection{Secret Communication}
Another relevant task is the
 transmission of classical information, subject to a secrecy constraint, preventing information from being leaked to the eavesdropper. 
Suppose that Alice encodes a classical message ${m}\in\mathcal{M}$ into an input state $\rho_{A^n}^{(m)}$, and transmits $A^n$ through the channel.
This results in an output state 
$\rho_{B^n W^n}^{(m)}=\mathcal{U}_{A\to BW}^{\otimes n}\left(\rho_{A^n}^{(m)}\right)$.
At the channel output, Bob and the eavesdropper receive 
$B^n=(B_1,B_2,\ldots,B_n)$ and
$W^n=(W_1,W_2,\ldots,W_n)$,
respectively. 
The receiver performs a decoding measurement on $B^n$ which produces an estimate $\widehat{\mathbf{m}}$. The performance can be characterized in terms of the secrecy rate, $R_\text{S}\equiv\frac{\log|\mathcal{M}|}{n}$, in addition to the
probability of error and the leakage. 
An $(\abs{\mathcal{M}},n,\varepsilon,\delta)$-code for the transmission of secret classical information satisfies
the following conditions: % 
\paragraph{Decoding Reliability}
The probability of decoding error  is $\varepsilon$-small, i.e.,
$% 
    \Pr(\widehat{\mathbf{m}}\neq \mathbf{m}) \leq \varepsilon
$. % 
\paragraph{Secrecy Criterion}
The indistinguishability distance is $\delta$-small.
That is, there exists a state $\breve{\rho}_{W^n}\in\mathscr{S}(\mathcal{H}_W^{\otimes n})$ that does not depend on the message $m$, such that
$% 
    \norm{ {\rho}^{(m)}_{W^n}-\breve{\rho}_{W^n}}_1 
    \leq \delta
  $ % 
  for all $m\in\mathcal{M}$, % 
where ${\rho}^{(m)}_{W^n}= \left( \mathcal{N}^{c}_{A \rightarrow W} \right)^{\otimes n} \left(\rho_{A^n}^{(m)}\right)$ is Willie's  state given that Alice has sent a particular message $m\in\mathcal{M}$.

 The secrecy capacities
${C}_\text{S}(\mathcal{N})$ and ${C}_\text{S}^0(\mathcal{N})$ are defined accordingly, with and without covertness, respectively.  See the second row in Table~\ref{table:LcapacityNotation}.
The quantum capacity and the secrecy capacity satisfy
${C}_\text{Q}(\mathcal{N})\leq {C}_\text{S}(\mathcal{N})$ and ${C}^0_\text{Q}(\mathcal{N})\leq {C}^0_\text{S}(\mathcal{N})$, since the transmission of quantum information entails ``built in'' secrecy following the no-cloning theorem \cite{nielsen2010quantum}.

Here, we have used the approach of Devetak \cite{Devetak05}, and constructed a code for entanglement generation using a secrecy code for the transmission of secret classical information. Devetak \cite{Devetak05} has shown an achievable quantum rate that equals the private information of a c-q channel $\mathcal{P}_{X\to BW}$, specified by  
\eqref{Equation:P_channel_EG}.

\subsection{Key Assistance}
\label{Section:Key_Assistance}

The covert capacity is typically defined under the assumption that Alice and Bob are provided with a shared secret key, before communication begins.  
Consider secret communication without covertness. In this case, 
the key rate is defined as $R_{\text{key}}\equiv\frac{\log|\mathcal{K}|}{n}$.
The best known achievable secrecy rate with a key of rate $R_{\text{key}}$ is 
        \begin{align}
            \underline{R}_\text{S} % 
            =&\max_{\rho_{XA}} \min\Bigl\{{I}(X;B)_\rho-{I}(X;E)_\rho+R_{\text{key}}, {I}(X;B)_\rho\Bigl\}
            \,.
        \end{align}
        Therefore, if the key rate is sufficiently large, we may achieve
        \begin{align}
            R_\text{S}=&\max_{\rho_{XA}} {I}(X;B)_\rho=\chi(\mathcal{N})
        \end{align}
        where $\chi(\mathcal{N})$ denotes the Holevo information of the channel $\mathcal{N}_{A\to B}$, as without secrecy.
        We note that the key rate
        $R_\text{key}=\chi(\mathcal{N}^c)$ is sufficient to achieve this. Unlike in one-time pad encryption, the key can be shorter than  the message.

        On the other hand, the entanglement-generation capacity does not increase (see \cite{lalou2025quantum}).
    Similarly, we show that entanglement can be generated at the same rate as for secret classical information without key assistance,
    i.e., $C_\text{EG}(\mathcal{N})=C_\text{S}(\mathcal{P})$, where $\mathcal{P}_{X\to BE}$ is as in \eqref{Equation:P_channel_EG}. % 

\subsection{Chi-Square Divergence}
\label{Discussion:Chi_Square}
 The chi-square divergence is important in various information theory settings, and it is closely related to other information measures \cite{Petz1998-gk,Temme2010-yl,george2025unified}.
In previous literature on quantum covert communication, the expression on the right-hand side of \eqref{eq:eta} is referred to as the $\eta$-divergence. Specifically, given a spectral decomposition of a full-rank operator,
$\sigma=\sum_i \lambda_i \Pi_i$,
the $\eta$-divergence is defined as 
\begin{align}
\label{eq:eta_1}
\eta(\rho||\sigma)
&\equiv \sum_{i\neq j} 
\frac{\log(\lambda_i)-\log(\lambda_j)}
{\lambda_i-\lambda_j}
\trace\left[ (\rho-\sigma)\Pi_i (\rho-\sigma)\Pi_j\right]
\nonumber\\
&\phantom{=}+
\sum_{i} 
\frac{1}
{\lambda_i}
\trace\left[ (\rho-\sigma)\Pi_i (\rho-\sigma)\Pi_i\right]
\end{align}
(see \cite[Eq.~(4)]{9344627}). Here, we observe that $\eta(\rho||\sigma)$ is in fact identical to the so-called  quantum chi-square divergence.

Consider two probability mass functions $p$ and $q$ with equal support on $\mathcal{X}$.
The classical $f$-divergence between $p$ and $q$  is defined by 
\begin{align}
D_f(p||q)=\sum_{x\in\mathcal{X}} q(x) f\left( \frac{p(x)}{q(x)} \right)
\end{align}
with respect to  a twice differentiable convex function $f:[0,\infty)\to\mathbb{R}$ such that $f(1)=0$. 
The $f$-divergence also has the following integral representation, due to Sason and Verd\'u \cite{7552457}:
\begin{align}
&D_f(p||q)
\nonumber\\
&=\int_1^\infty \left( f''(\gamma)E_\gamma(p||q)
+\gamma^{-3} f''(\gamma^{-1})E_\gamma(p||q)
\right)\,d\gamma
\label{Equation:f_Divergence}
\end{align}
where $E_\gamma$ is known as the hockey-stick divergence, $E_\gamma(p||q)=\sum_{x\in\mathcal{X}} \max\{p(x)-\gamma q(x),0\}$ for $\gamma\geq 1$.
The classical relative entropy $D(p||q)$ is a special case corresponding to 
$f(x)=x\log x$, 
the classical Hellinger distance $H_\alpha(p||q)$ is obtained for $f(x)=\frac{x^\alpha-1}{\alpha-1}$, the classical R\'enyi divergence is given by
$D_\alpha(p||q)=\frac{1}{\alpha-1}\log\left(1+(\alpha-1)H_\alpha(\rho||\sigma) \right)$, and the classical chi-square divergence is the Hellinger distance of order $\alpha=2$:
\begin{align}
\chi^2(p||q)=H_2(p||q)=\sum_{x\in\mathcal{X}}  \frac{\left( p(x)-q(x) \right)^2}{q(x)}
\,.
\end{align}
In principle, one could define the quantum $\chi^2$-divergence as 
\begin{align}
\widehat{\chi}^2(\rho||\sigma)=\trace\left[\sigma^{-1} \left( \rho-\sigma \right)^2\right]
\end{align}
as in \cite[Eq.~(2)]{9344627}). However, this definition depends on a specific choice of division by $\sigma$ (see discussion in \cite[Sec. V-A]{george2025unified}). We do \emph{not} use this definition here. 
Instead, we use the following convention.

Following \cite[Sec. 1]{HircheTomamichel:24p}, 
 the quantum $f$-divergence $D_f(\rho||\sigma)$ is defined by a similar integral representation as in \eqref{Equation:f_Divergence}, 
where the quantum hockey-stick divergence $E_\gamma(\rho||\sigma)$ is the sum of the positive eigenvalues of $\rho-\gamma \sigma$.
In analogy to the classical case, the quantum relative entropy $D(\rho||\sigma)$ is a special case corresponding to 
$f(x)=x\log x$, 
the quantum Hellinger distance $H_\alpha(\rho||\sigma)$ is obtained for $f(x)=\frac{x^\alpha-1}{\alpha-1}$, and the quantum chi-square divergence  is the Hellinger distance of order $\alpha=2$:
\begin{align}
\chi^2(\rho||\sigma)=H_2(\rho||\sigma)
\,.
\end{align}
We note that if the operators $\rho$ and $\sigma$ commute, then
$\chi^2(\rho||\sigma)=\widehat{\chi}^2(\rho||\sigma)$.

On the one hand, Hirche and Tomamichel \cite{HircheTomamichel:24p} established that
the chi-square divergence is related to the quantum relative entropy by
\begin{align}
\chi^2(\rho||\sigma)=\frac{\partial^2}{\partial \alpha^2} D(\alpha\rho+(1-\alpha)\sigma||\sigma)\Big|_{\alpha=0} 
\end{align}
(see \cite[Th. 2.8]{HircheTomamichel:24p}).
On the other hand, Tahmasbi and Bloch \cite{9344627} showed that the quantum relative entropy satisfies
\begin{align}
D(\alpha \rho+(1-\alpha)\sigma||\sigma)=
\frac{1}{2}\alpha^2 \eta(\rho||\sigma)+O(\alpha^3)
\end{align}
(see \cite[Lemma~5]{9344627}).
Therefore,
\begin{align}
\chi^2(\rho||\sigma)&=\frac{\partial^2}{\partial \alpha^2} D(\alpha\rho+(1-\alpha)\sigma||\sigma)\Big|_{\alpha=0}
=\eta(\rho||\sigma) \,.
\end{align}
Based on our observation, there is no longer a need to treat $\eta(\rho||\sigma)$ as a new information measure, and it suffices to consider the well known chi-square divergence. 

\subsection{Covert Communication Scenarios}
\label{Subsection:Covert_Communication_Scenarios}

In Subsection~\ref{Section:Scenarios}, we have presented our assumptions on Bob and Willie's output states:
$\mathrm{supp}(\omega_1) \subseteq \mathrm{supp}(\omega_0)$, $\omega_1\neq \omega_0$, and
$\mathrm{supp}(\sigma_1) \subseteq \mathrm{supp}(\sigma_0)$.
This is usually the regime of interest throughout the covert communication literature \cite{Bloch16, tahmasbi2020covertIEEE,SheikholeslamiB16, 10886999, GagatsosBB20}.

 We now consider the scenarios when these assumptions do not hold:
    \begin{itemize}
        \item When the support condition on Willie's side is violated, i.e., $\mathrm{supp}(\omega_1) \nsubseteq \mathrm{supp}(\omega_0)$, there exists a measurement that perfectly distinguishes $\omega_0$ from $\omega_1$. 
        As the quantum relative entropy $D(\omega_1 \| \omega_0)$ diverges, Willie can detect any non-innocent transmission with certainty and no coding scheme can hide the transmission. Therefore, the covert capacity is zero. This is a fundamental physical impossibility rather than a limitation of the analysis. Further discussions and analysis can be found in Bloch~\cite[App.~G]{Bloch16} and Bullock et al.~\cite[Sec.~IV-C]{10886999}.
    
        \item When $\omega_1 = \omega_0$, Willie's outputs are identical regardless of Alice's input, and detection is impossible. In this case, the covertness constraint imposes no restriction, thus the number of EPR pairs scales linearly in $n$. % 
    
        \item A more interesting case is when the support condition on Bob's side is violated, i.e., $\mathrm{supp}(\sigma_1) \nsubseteq \mathrm{supp}(\sigma_0)$. Following the analysis in~\cite[App.~G]{Bloch16} and~\cite[Sec.~IV-B]{10886999} for classical information, Bob's enhanced distinguishability allows the number of covert bits to scale as $\sim \sqrt{n}\log{n}$, surpassing the standard square root law. We expect an analogous improvement for the entanglement-generation setting, though a rigorous derivation is left for future work.
    \end{itemize}

    Several standard channel models, such as the qubit depolarizing, dephasing,  
    erasure, amplitude-damping and qubit-flip channel,  violate the required % 
    assumptions. % 
    In such cases, the induced output states at Bob or Willie may have non-nested supports, leading to either certain detection or trivial indistinguishability.
    Consequently, for these channels the resulting covert capacity is either zero (when Bob's distinguishability is no greater than Willie's)
    or trivial, % 
    and therefore they do not yield a meaningful finite covert rate within our framework.
    We now elaborate on why the standard channels mentioned above generally fail to satisfy the required % 
    conditions. % 

    \subsubsection{Depolarizing Channel}

    Consider a qubit depolarizing channel from Alice to Bob, that transmits the input qubit perfectly
    with probability $1 - q$, and outputs a completely mixed state with probability $q$. 
    The channel is expressed as
    \begin{align}
        \mathcal{N}_{A \to B}(\rho) = (1-q)\rho + q\frac{\identity}{2}
    \end{align}
    where $q \in (0, 1)$, and $\dim{\mathcal{H}_A} = \dim{\mathcal{H}_B} = 2$.
    The depolarizing channel for both Bob and Willie $\mathcal{U}_{A \to BW}$ is given by the Stinespring dilation
    $% 
        \mathcal{U}_{A \to BW}(\rho) = V\rho V^\dagger
    $, % 
    where $V: \mathcal{H}_A \to \mathcal{H}_B \otimes \mathcal{H}_W $ is an isometry defined in the following canonical representation,
    \begin{multline}
        V \equiv 
        \sqrt{1-\frac{3q}{4}} \identity \otimes \ket{0} 
        + \sqrt{\frac{q}{4}} \mathsf{X} \otimes \ket{1} 
        \\
        + \sqrt{\frac{q}{4}} \mathsf{Y} \otimes \ket{2}
        + \sqrt{\frac{q}{4}} \mathsf{Z} \otimes \ket{3}
        \,,
    \end{multline}
       where $\mathsf{X}, \mathsf{Y}, \mathsf{Z}$ are the Pauli operators. % 
    Willie's output states
    do \emph{not}
    satisfy $\text{supp}(\omega_1) \subseteq \text{supp}(\omega_0)$ \cite[App. A]{ZlotnickBP23}, hence both $\RelEntropy{\omega_1}{\omega_0}$ and $\chi^2(\omega_1||\omega_0)$ tend to $+\infty$.
    Therefore, both the covert secrecy capacity and entanglement-generation capacity are zero, as covert communication is impossible.

   \subsubsection{Erasure Channel}
Consider the qubit erasure channel, specified by the isometry
    \begin{align}
        \label{eq:erasure_isometry}
        V
        &=
        \sqrt{1-\varepsilon}
  \identity_{A\to B}
        \otimes
        \ket{e}_W+ \ket{e}_B \otimes \identity_{A\to W}
    \end{align}
    where $\ket{e}$ is an erasure state that is orthogonal to the input qubit space. % 
    From Bob's perspective, he receives the original qubit state with probability $1-\varepsilon$, and receives an erasure with probability $\varepsilon$. If Bob receives the qubit, then Willie receives an erasure, and vice versa.
    Consequently, $\mathrm{supp}(\omega_0)=\mathrm{span}\{\ket{0},\ket{e}\}$,
    whereas, $\mathrm{supp}(\omega_1)=\mathrm{span}\{\ket{1},\ket{e}\}$.
    Therefore, Willie can detect a non-innocent transmission, and covert communication is impossible.

    \subsubsection{Amplitude-Damping Channel}

    Consider the amplitude-damping channel with a damping parameter
    $\gamma \in (0,1)$, for which % 
    \begin{align}
        V \ket{0}_A 
        &= \ket{0}_B \ket{0}_W, 
        \nonumber\\
        V \ket{1}_A 
        &= \sqrt{1-\gamma}\,\ket{1}_B \ket{0}_W
        + \sqrt{\gamma}\,\ket{0}_B \ket{1}_W.
    \end{align}
     Here, $\omega_0$ has the support of $\ketbra{0}$,
    whereas, $\omega_1$ has a full support.
    As before,  covert communication is thus impossible.
    On the other hand, the excitation channel in Subsection~\ref{subsection:example_excitation_channel} describes the reverse process of amplitude damping, and admits a positive covert capacity.

    \subsubsection{Quantum Bit-Flip Channel}
    Consider the qubit-flip channel, 
    specified by the isometry
    \begin{align}
        V=\sqrt{1-p}\,\identity\otimes \ket{0}+\sqrt{p}\,\mathsf{X}\otimes \ket{1}
    \end{align}
    for $p\in(0,1)$.
    Willie's output states are identical, i.e. $\omega_0=\omega_1$, hence Willie cannot distinguish between the inputs and covert communication reduces to the ordinary non-covert setting.

{
We note that although the innocent input is typically chosen as $\ketbra{0}$, corresponding to the ground state or a ``quiet'' (no-transmission) input to the channel, this choice is not fundamental. In some channel models, a different state may serve as the natural innocent input, e.g., $\ketbra{+}$. In such cases, the induced output states $(\sigma_0,\omega_0)$ and $(\sigma_1,\omega_1)$ change accordingly, and the capacity calculations for specific channel examples must be adapted to this choice. Nevertheless, our general capacity results depend only on the resulting output states and not on the particular representation of the input states. Therefore, the capacity theorems derived in this work remain valid for any choice of innocent input.

}

\subsection{Future Directions}

As quantum networks mature, with recent demonstrations of quantum key distribution
over deployed fiber infrastructure \cite{li2025microsatellite} and entanglement distribution
across metropolitan networks \cite{zheng2026large}, covert entanglement generation
becomes operationally relevant for scenarios where the very act of establishing shared
entanglement must remain undetectable.
Our capacity results quantify the fundamental cost of covertness for entanglement
distribution over noisy quantum channels, and, thus, provide design benchmarks for such
systems. Furthermore, the combined setting of  covertness and secrecy studied here parallels
recent investigations of joint covert communication and physical layer security in classical
wireless systems \cite{nguyen2026covert,li2025intelligent}, where similar trade-offs between
undetectability and confidentiality arise.

 Several open questions remain in the study of covert and secret communication. Communicating quantum information requires the transmission of an arbitrary quantum state, not necessarily symmetric. Our analysis focuses on entanglement generation, and does not cover this scenario. The primary challenge lies in ensuring covertness. In the case of a maximally entangled state, the reduced input state resembles a uniformly distributed message, which facilitates covertness. However, for an arbitrary quantum state, the resulting distribution may be non-uniform, in which case standard derivations of covert communication no longer apply. Developing new covertness analyses that accommodate non-uniform input distributions, or identifying encoding strategies that reconcile arbitrary state transmission with covertness constraints, remains an important open problem.
    
    One of the key questions in covert communications is determining the resources necessary to achieve covertness \cite{wang2024resource}.
    We characterize the covert information rate  $L_\text{S}$, when we are either using key assistance at a positive rate $L_\text{key}>0$, or no key at all.
    Another future direction is to
    obtain a more complete characterization of the
 trade-off of resources. Specifically, one may investigate  the region of achievable rate tuples $(L_\text{S},L_{\text{public}},L_\text{EG}, L_\text{key})$ to better understand the trade-off between the secret information, public information, entanglement generation, and key rates. % 
 The characterization of  trade-offs between resources is significant for the design and implementation of covert communication systems \cite{shi2019resource}. Recent result for classical communication \cite{bounhar25covertcapkeytradeoff} is a valuable starting point for this work.
    
    Furthermore, 
    in our model, we assume  Willie receives the complementary output to that of Bob, i.e., the channel from Alice to Bob and Willie is isometric.
    A more general scenario, however, is the non-isometric case, where Willie does not possess the full receiver environment. In this case, Willie's detection capabilities are strictly weaker, which may allow for higher covert communication rates, and more channels would be applicable for covert entanglement generation. However, analyzing this scenario requires new bounding techniques, as the isometric assumption plays a key role in our current coding schemes and converse bounds. Characterizing the covert capacity region under non-isometric channels is therefore a meaningful direction for future work.
    
    Moreover, in scenarios where Bob has an unfair advantage, Bullock et al.~\cite[Sec.~IV-B]{10886999} prove that a pre-shared key is unnecessary, as in our setting, while the number of information bits scales as $\sim \sqrt{n}\log{n}$, surpassing the standard square root law. We expect an analogous improvement for the entanglement-generation setting, though a rigorous derivation is left for future work. 
    
    Finally, as observed in~\cite{anderson2024covert, anderson2025achievability}, if an SRL-based covert classical communication channel is available, then covert quantum communication can be achieved via teleportation. This, however, requires covertness analysis for the overall protocol, including entanglement generation, quantum measurement, and classical communication. Furthermore, the total covertness of a combined setting of two different covert communication types remains an open question. Investigating the feasibility and efficiency of such a scheme could have significant implications for quantum networks and secure communications.

\section*{Acknowledgment}
The authors would like to thank Marco Tomamichel
(National University of Singapore) for motivating them to consider covert entanglement generation, highlighting its necessity
for covert entanglement-assisted communication. They also
thank Ian George (National University of Singapore) for
pointing out that $\eta(\rho||\sigma)$ is a quantum chi-square divergence.
Furthermore, they thank Matthieu Bloch (Georgia Tech) for
sharing the converse proof for the key rate in the classical model of covert communication with secrecy and key
assistance. They also acknowledge the helpful discussions
with Evan J.\ D.\ Anderson (formerly at The University of
Arizona, currently at the University of Maryland) and Michael
S.\ Bullock (formerly at The University of Arizona, currently
at the University of Massachusetts, Amherst).

\clearpage

\setcounter{section}{0}
\setcounter{equation}{0}
\setcounter{theorem}{0}
\setcounter{definition}{0}
\setcounter{remark}{0}
\setcounter{example}{0}
\setcounter{figure}{0}
\setcounter{table}{0}
\setcounter{footnote}{0}

\renewcommand{\thesection}{S.\arabic{section}}
\renewcommand{\theequation}{S.\arabic{equation}}
\renewcommand{\thetheorem}{S.\arabic{theorem}}
\renewcommand{\thedefinition}{S.\arabic{definition}}
\renewcommand{\theremark}{S.\arabic{remark}}
\renewcommand{\theexample}{S.\arabic{example}}
\renewcommand{\thefigure}{S.\arabic{figure}}
\renewcommand{\thetable}{S.\arabic{table}}

\makeatletter
\renewcommand{\theHsection}{S.\arabic{section}}
\renewcommand{\theHequation}{S.\arabic{equation}}
\renewcommand{\theHtheorem}{S.\arabic{theorem}}
\renewcommand{\theHdefinition}{S.\arabic{definition}}
\renewcommand{\theHremark}{S.\arabic{remark}}
\renewcommand{\theHexample}{S.\arabic{example}}
\renewcommand{\theHfigure}{S.\arabic{figure}}
\renewcommand{\theHtable}{S.\arabic{table}}
\renewcommand{\theHsubsection}{S.\arabic{section}.\arabic{subsection}}
\renewcommand{\theHsubsubsection}{S.\arabic{section}.\arabic{subsection}.\arabic{subsubsection}}
\makeatother

\title{Supplementary Material for ``Covert Entanglement Generation and Secrecy''
\thanks{ ... same \thanks contents as the supplementary's title ... }}
\author{\IEEEauthorblockN{Ohad Kimelfeld, Boulat A. Bash and Uzi Pereg}}
\maketitle

\makeatletter
\twocolumn[
  \begin{@twocolumnfalse}
    \begin{center}
      {\fontsize{24}{28}\selectfont
       Supplementary Material for\\[0.25em]
       ``Covert Entanglement Generation and Secrecy''\par}
      \vspace{1.5em}
    \end{center}
  \end{@twocolumnfalse}
]
\makeatother

We provide the analysis for our covert secrecy results and technical lemmas for the entanglement-generation derivation. Sections~\ref{Section:Secrecy_Proof} and \ref{Section:Minimal_Key}  provide the derivation for the covert secrecy capacity
with key assistance, and
 the minimal key rate required to achieve 
 this capacity.
In Section~\ref{Appendix:Covert_Secrecy_Capacity_No_Key_Proof}, we show the capacity proof for  covert secrecy without assistance. 
Finally, in Section~\ref{Appendix:Quantum_Code_First_Approximation}, we give the proof for our state approximation lemma,  Lemma~\ref{Lemma:quantum_code_1st_approx}.

\section{Proof of Theorem~\ref{Theorem:Covert_Secrecy_Capacity} % 
(Covert Secrecy Capacity With Key Assistance)}
\label{Section:Secrecy_Proof}
We show that covert secrecy capacity with key assistance is the same as without the secrecy requirement
(see Remark~\ref{Remark:Secrecy_Rate}).
Here, we do not include a public message, i.e., $L_{\text{public}}=0$.
The converse part is thus immediate, since the secrecy capacity is bounded from above by the capacity without secrecy.

To show achievability for covert secrecy via a  classical-quantum channel, we combine several methods from previous works on secret and covert communication. % 
We build upon covert communication results without secrecy % 
\cite{10886999} along with the secrecy coding approach proposed by Bloch \cite{Bloch16} for classical channels. Specifically,  we employ binning and one-time pad encryption to guarantee security, and then
we  use the quantum channel resolvability lemma due to Hayashi 
in the secrecy analysis
\cite{hayashi2006quantum}. 

\subsection{Analytic Tools}
We provide the basic analytic tools for our analysis below.

\subsubsection{Operators}
Consider a Hermitian operator 
$P\in\mathscr{L}(\mathcal{H})$ with a spectral decomposition
$P=\sum_j \lambda_j\Pi_j$, where $\Pi_j$ are orthogonal projectors. The projector onto the non-negative eigenspace is defined as
$% 
\{P\geq 0\}=\sum_{j:\lambda_j\geq 0} \Pi_j % 
$. % 
Furthermore, the pinching of an operator $Q\in\mathscr{L}(\mathcal{H})$ with respect to $P$ is defined as the following map,
\begin{align}
\mathcal{E}_P \,:\; Q\,\mapsto\;
\sum_j \Pi_j Q\Pi_j
\,.
\end{align}
The lemma below bounds the number of distinct eigenvalues for an $n$-fold product of $P$.
\begin{lemma}[see {\cite[Lemma 3.7]{hayashi2006quantum}}]
\label{Lemma:omega_alpha_n_number_of_distinct_eigenvalues}
     Let $P\in\mathscr{L}(\mathcal{H})$ be a Hermitian operator. % 
    The % 
    $n$-fold product operator
     $P^{\otimes n}$ has at most  $(n+1)^d$ distinct eigenvalues, where $d=\mathrm{dim}(\mathcal{H})$.
\end{lemma}

\begin{lemma}[H\"older's inequality (see {\cite[Sec. 12.2.1]{W2017}})]
\label{Lemma:H\"older inequality}
    For every two operators $Q_1,Q_2 \in \mathscr{L}(\mathcal{H})$, % 
    \begin{align}
        \abs{\Tr{Q_1^\dagger Q_2}} \leq \norm{Q_1}_p \norm{Q_2}_q
    \end{align}
    for $p$ and $q$ such that $p^{-1} + q^{-1} = 1$, $p,q\geq 1$, where $p,q \in \mathbb{R}\cup \{\infty\}$ and $\infty^{-1}\equiv 0$.
\end{lemma}
We use H\"older's inequality with $p=1$ and $q=\infty$.

\subsubsection{Quantum resolvability}
Resolvability quantifies the amount of randomness that is required in order to simulate a particular state through a given ensemble. Consider a probability distribution $p(x)$ on a state ensemble, $\{ \omega_x\}_{x\in\mathcal{X}}$. Denote the average state by 
\begin{align}
\omega_p=\sum_{x\in\mathcal{X}}p(x)\omega_x \,.
\end{align}
Following the definition by Hayashi \cite{hayashi2006quantum}, let
    \begin{align}
        \phi(s, p) \equiv \log{\sum_{x \in \mathcal{X}}p(x)\Tr{\omega_x^{1-s}\omega_{p}^s}}
        \label{Equation:phi_p}
    \end{align}
    for $s \leq 0$ % 
    (see \cite[Eq. (9.53)]{hayashi2006quantum}). 
    This definition is closely related to the R\'enyi relative entropy, as 
    $\phi(s,p)=\log\mathbb{E}_\mathbf{x}\left\{ e^{(s-1)\bar{D}_s(\omega_p||\omega_\mathbf{x})}\right\}$,
where $\bar{D}_s(\cdot||\cdot)$ is the Petz-R\'enyi relative entropy of order $s$. For a binary probability distribution
$p=(1-\alpha,\alpha)$ on $\{\omega_0,\omega_1\}$, we denote
$% 
\omega_\alpha=(1-\alpha)\omega_0+\alpha\omega_1
$ % 
instead of $\omega_p$, and similarly, $\phi(s, \alpha)$ instead of $\phi(s, p)$.
\begin{lemma}[Quantum resolvability {\cite[Lemma 9.2]{hayashi2006quantum}
 \cite% 
{TahmasbiBloch:19p}
}]
\label{Lemma:Quantum_channel_resolvability}
Let $\{p(x),\omega_x\}_{x\in\mathcal{X}}$ be a given ensemble in $\mathscr{S}(\mathcal{H})$, with an average $\omega_p$.
Consider a random codebook
$\{\mathbf{c}(m),m\in\mathcal{M}\}$, where each codeword $\mathbf{c}(m)$ is drawn independently at random according to a probability mass function $p: \mathcal{X}\to [0,1]$.
Then,
\begin{align}
&\mathbb{E}_{\mathscr{C}}
\TrNorm{\frac{1}{\abs{\mathcal{M}}}
\sum_{m\in\mathcal{M}} 
\omega_{\mathbf{c}(m)}-\omega_p}
\nonumber\\
&\leq
2\sqrt{\exp{\beta s +\phi(s,p)}}+
\sqrt{\frac{ e^{\beta} \nu } {\abs{\mathcal{M}}}}
\end{align}
for all $s\leq 0$ and $\beta \in \mathbb{R}$,
where the expectation is with respect to the random codebook $\mathscr{C}=\{\mathbf{c}(m)\}$, and $\nu$ is the number of distinct eigenvalues of $\omega_p$.
\end{lemma}

The lemma below provides an upper bound on $\phi(s, \alpha)$.
\begin{lemma}[see {\cite[Lemma 8]{10886999}}]
\label{Lemma:phi_s_alpha_bound}
    Let $s_0 < 0$ be an arbitrary constant, $\mathbf{x}\sim \text{Bernoulli}(\alpha)$
    for $\alpha\in [0,1]$, and
    \begin{align}
    \label{Equation:omega_alpha}
    \omega_{\alpha} &=
    (1-\alpha)\omega_0+\alpha \omega_1
    \,.
\end{align}   
    Then, 
    \begin{align}
        \phi(s, \alpha) \leq -\alpha s D(\omega_1||\omega_0) + \vartheta_1 \alpha s^2 -\vartheta_2 s^3
        \,,
    \end{align}
    for all 
    $ s \in [s_0,0]$ 
    and some constants $\vartheta_1, \vartheta_2 > 0$, independent of $s$ and $\alpha$.
\end{lemma}

Finally, roughly speaking, the covert encoding scheme is based on a sparse coding protocol, % 
with only a fraction of $\alpha_n$ non-zero transmissions \cite{10886999}\cite{tahmasbi21covertsignaling}. This fraction is chosen as 
$\alpha_n=\frac{\gamma_n}{\sqrt{n}}$, where $\gamma_n$ tends to zero.
The state $\omega_{\alpha_n}$ is referred to in \cite{10886999} as the ``quantum-secure covert state.'' Further discussion and properties are shown in \cite[Sec. II-E]{10886999}. In particular, covertness is established in \cite{10886999} using the properties $D((1-\alpha_n)\omega_0+\alpha_n\omega_1||\omega_0)\approx \frac{1}{2}\alpha_n^2 \chi^2(\omega_1||\omega_0) $ and $D(\omega_{\alpha_n}^{\otimes n}||\omega_0^{\otimes n})
=nD(\omega_{\alpha_n}||\omega_0)\approx \frac{1}{2}\gamma_n^2\chi^2(\omega_1||\omega_0)\in O(\gamma_n^2) $.
We formalize this in the lemma below.

\begin{lemma}[see {\cite[Lemma 5]{9344627}}]
\label{Lemma:avg_QRE_to_chi_square}
Let $\rho$ and $\sigma$ be two density operators such that $\sigma$ is invertible. Then, for small $\alpha > 0$,
\begin{align}
    D((1-\alpha)\sigma + \alpha\rho||\sigma)
    = 
    \frac{1}{2}\alpha^2 \chi^2(\rho||\sigma)
    + % 
    {O}(\alpha^3)
\end{align}
where $\chi^2(\rho||\sigma)$ is defined as in \eqref{eq:eta}.
\end{lemma}

\subsection{Achievability Proof}% 
\label{Subsection:random_codebook_analysis}
We now show achievability of secret and covert classical-quantum communication (see Subsection~\ref{Subsubsection:Secrecy_code}).
We use the notation in Tables~\ref{tab:general-notation}-\ref{tab:secrecy-notation}.
\begin{proposition}% 
\label{Proposition:classical_covert_and_secret}
    Consider a covert memoryless classical-quantum channel such that
    $\text{supp}(\sigma_1) \subseteq \text{supp}(\sigma_0)$, 
    $\text{supp}(\omega_1) \subseteq \text{supp}(\omega_0)$, and $\omega_1\neq\omega_0$. Let $\alpha_n = \frac{\gamma_n}{\sqrt{n}}$
    with $\gamma_n \in o(1) \cap \omega\left(\frac{(\log{n})^{\frac{7}{3}}}{n^{\frac{1}{6}}}\right)$.
    Then, 
    for any $\zeta_n \in o(1) \cap \omega\left((\log n)^{-\frac{2}{3}}\right)$, 
    there exist
    $\zeta_n^{(1)} \in \omega\left((\log n)^{-\frac{4}{3}}
    n^{-\frac{1}{3}}\right), \zeta_n^{(2)} \in \omega\left((\log n)^{-2}\right), \zeta_n^{(3)} \in
    \omega\left((\log n)^{-1}\right)$
    and a classical-quantum covert secrecy code, % 
    such that, for $n$ sufficiently large:
    \begin{align}
        \log{\abs{\mathcal{M}}} 
        &= (1 -\zeta_n)\gamma_n\sqrt{n}\RelEntropy{\sigma_1}{\sigma_0}
        \,, 
        \nonumber\\
        \log{\abs{\mathcal{K}}} 
        &= (1 + \zeta_n)\gamma_n\sqrt{n}\RelEntropy{\omega_1}{\omega_0}
        \label{eq:Theorem1LogMLogK}
    \end{align}
    and 
\begin{subequations}
    \label{Equation:Random_Code_Bounds}
    \begin{align}
         &\overline{P}_e^{(n)} \leq e^{-\zeta_n^{(1)}\gamma_n\sqrt{n}},
         \label{eq:Theorem1Errors}
         \\
         &\abs{\RelEntropy{\overline{\rho}_{W^n}}{\omega_0^{\otimes n}} - \RelEntropy{\omega_{\alpha_n}^{\otimes n}}{\omega_0^{\otimes n}}} \leq e^{-\zeta_n^{(2)}\gamma^{\frac{3}{2}}_n n^\frac{1}{4}}, 
         \label{eq:Theorem1Covertness}
        \\
        &\frac{1}{\abs{\mathcal{M}}}\sum_{m\in\mathcal{M}}\TrNorm{\rho^{(m)}_{W^n} - \omega_{\alpha_n}^{\otimes n}}
        \leq  e^{-\zeta_n^{(3)}\gamma_n^{\frac{3}{2}} n^{\frac{1}{4}}}
\,.
        \label{eq:Theorem1Secrecy}
    \end{align}
    \end{subequations}
\end{proposition}

\begin{remark}
The constraints on the parameters $\gamma_n$, $\zeta_n$, $\zeta_n^{(1)}$, and $\zeta_n^{(2)}$ are inherited from the setting without secrecy \cite{10886999}, since our proof for covert secrecy achieves reliability and covertness under the same constraints.
The constraint on $\zeta_n^{(3)}$ results from our secrecy analysis (see \eqref{Equation:Resolvability_Omega_alpha_1}-\eqref{Equation:Resolvability_Omega_alpha_2} below).
We employ a more restrictive constraint on $\gamma_n$ compared to previous works without secrecy, \cite[Eq.~(25)]{10886999} and \cite{Bloch16}, where
$\gamma_n \in o(1) \cap \omega({(\log n)^{\frac{4}{3}}} / {n^{\frac{1}{6}}})$ and $\gamma_n \in o(1) \cap \omega( {\log n} / {\sqrt{n}})$, respectively.
This stronger constraint is needed to ensure that  covertness  is preserved after the expurgation step in Lemma~\ref{Lemma:Expergated_classical_code} (without key assistance), which is required for the construction of the entanglement-generation code via Devetak's approach~\cite{Devetak05}, and that the asymptotic covert rate is achieved.
\end{remark}

\begin{proof}[Proof of Proposition~\ref{Proposition:classical_covert_and_secret}]
We divide the proof into two cases, following similar steps as in the classical work by Bloch \cite% 
{Bloch16} (see Remark~\ref{Remark:OTP_Key_Rate}). 
If 
$(1 - \zeta_n)\RelEntropy{\sigma_1}{\sigma_0} \leq (1 + \zeta_n)\RelEntropy{\omega_1}{\omega_0}$, then we know from \cite[Th. 1]{10886999} that we can transmit covertly $\log{\abs{\mathcal{M}}} = (1 -\zeta_n)\gamma_n\sqrt{n}\RelEntropy{\sigma_1}{\sigma_0}$ message bits with the help of $\log{\abs{\mathcal{K}}} = \gamma_n\sqrt{n}\left[(1 + \zeta_n)\RelEntropy{\omega_1}{\omega_0} - (1 - \zeta_n)\RelEntropy{\sigma_1}{\sigma_0}\right]$ key bits. 
Secrecy of the message bits can be ensured via one-time pad (OTP) encryption, which requires an additional
$(1 - \zeta_n)\gamma_n\sqrt{n}\RelEntropy{\sigma_1}{\sigma_0}$ key bits. 
Consequently, the total number of key bits is
$\log{\abs{\mathcal{K}}} = (1 + \zeta_n)\gamma_n\sqrt{n}\RelEntropy{\omega_1}{\omega_0}$. 
It remains to consider the case where 
\begin{align}
(1 - \zeta_n)\RelEntropy{\sigma_1}{\sigma_0} > (1 + \zeta_n)\RelEntropy{\omega_1}{\omega_0}
\,.
\label{Equation:Secrecy_Case_2}
\end{align}
In this case, covertness without secrecy can be achieved without a key, i.e. $\log{\abs{\mathcal{K}}} = 0$ (see \cite[Th. 1]{10886999}). Therefore, we modify the random coding argument % 
as follows. 
\subsection*{Code Construction}
We use rate splitting and consider a message that consists of two components, $m_1$ and $m_2$, that are encrypted using binning and OTP, respectively.
\paragraph*{Classical codebook generation}
Generate $\abs{\mathcal{M}_1} \abs{ \mathcal{M}_2}$ codewords $c(m_1,m_2) \in \{0, 1\}^n$, independently at random, for $(m_1,m_2)\in\mathcal{M}_1\times \mathcal{M}_2$, % 
each drawn i.i.d.~ from the 
$\text{Bernoulli}(\alpha_n)$ distribution.  
Reveal the codebook in public, to Alice, Bob, and Willie.

\paragraph*{Encoder}
Given the message $m=(m_1,m_2)$, 
transmit $x^n=c(m_1,m_2)$ over the channel.

\paragraph*{Decoder}
Define the projectors
\begin{equation}
    {\Pi}_{x^n} \equiv \{\mathcal{E}_{\sigma_0^{\otimes n}}\left(
    \sigma_{x^n}
    \right) - e^{a_n}\sigma_0^{\otimes n} \geq 0\}
\label{Equation:Codeword_Projection}
\end{equation}
for $x^n\in\{0,1\}^n$, where $a_n\in o(\sqrt{n})$ is set below.
Bob performs the ``square-root measurement,'' specified by the following POVM,
\begin{align}
    {\Upsilon}_{m_1,m_2} 
    &= \Sigma^{-\frac12} {\Pi}_{c(m_1,m_2)} 
    \Sigma^{-\frac12}
\label{Equation:Upsilon}
\end{align}
for $(m_1,m_2)\in\mathcal{M}_1\times\mathcal{M}_2$,
where $\Sigma\equiv \sum_{m_1',m_2'}{\Pi}_{c(m_1',m_2')}$.

\subsection*{Error and covertness analysis}
The error and covertness derivations follow the previous results without secrecy, due to Bullock et al.~\cite{10886999}, as we briefly explain below.
Suppose that Alice is using the channel to transmit classical information.
Bob's output state given the message $m=(m_1,m_2)$, is 
$% 
    \sigma_{\mathbf{c}(m_1,m_2)} = \bigotimes_{i=1}^n \sigma_{\mathbf{c}_i(m_1,m_2)} 
$. % 
We denote the random codebook by 
$\mathscr{C}=\{\mathbf{c}(m_1,m_2)\}$. 
The average error probability  with respect to this codebook  % 
is 
\begin{equation}
    \overline{P}_e^{(n)}(\mathscr{C}) = \frac{1}{\abs{\mathcal{M}}}\sum_{m\in\mathcal{M}}{\left(1 - \Tr{{\Upsilon}_{m}\sigma_{\mathbf{c}(m)}}\right)} 
\end{equation}
as without secrecy, where we take $m=(m_1,m_2)$ and 
$\mathcal{M}\equiv \mathcal{M}_1\times \mathcal{M}_2$.
Therefore, by the same error analysis as in covert communication without secrecy 
\cite[Th. 1]{10886999}, the expected error probability is bounded as in \eqref{eq:Theorem1Errors},
 for
 \begin{equation}
     \log{\abs{\mathcal{M}}} = \log{\abs{\mathcal{M}_1}} + \log{\abs{\mathcal{M}_2}} = (1 -\zeta_n)\gamma_n\sqrt{n}\RelEntropy{\sigma_1}{\sigma_0}
     \label{Equation:Secrecy_Message_Size}
 \end{equation}
 choosing $a_n\in o(\sqrt{n})$ in \eqref{Equation:Codeword_Projection} to be  as in 
\cite[Eq. (65)]{10886999}.
We now consider the covertness.
If Alice has used the channel, 
Willie's average state is
$% 
    \overline{\rho}_{W^n} = \frac{1}{\abs{\mathcal{M}_1}\abs{\mathcal{M}_2}}\sum_{m_1,m_2}% 
    {\omega_{\mathbf{c}(m_1,m_2)}}
$, % 
where
$% 
     \omega_{\mathbf{c}(m_1,m_2)} = \bigotimes_{i=1}^{n}\omega_{\mathbf{c}_i(m_1,m_2)} 
$. % 
By quantum channel resolvability \cite% 
{hayashi2006quantum} (see Lemma~\ref{Lemma:Quantum_channel_resolvability}), we have for any $s_n \leq 0$ and $\beta_n \in \mathbb{R}$:
\begin{equation}
    \mathbb{E}_{\mathscr{C}}\left[\TrNorm{\overline{\rho}_{W^n}- \omega_{\alpha_n}^{\otimes n}}\right] 
    \leq 2\sqrt{e^{\beta_n s_n + n\phi(s_n,\alpha_n)}} + \sqrt{\frac{e^{\beta_n} \nu_n}{\abs{\mathcal{M}_1}\abs{\mathcal{M}_2}}}  \label{eq:Theorem1ProofCovertnessBound}
\end{equation}
where $\nu_n$ is the number of distinct eigenvalues of $\omega_{\alpha_n}^{\otimes n}$ 
as defined in \eqref{Equation:omega_alpha}, setting $\alpha = \alpha_n$, and $\phi(s_n,\alpha_n)$ as  in % 
\eqref{Equation:phi_p}.
We observe that the bound on the right-hand-side of \eqref{eq:Theorem1ProofCovertnessBound} has appeared 
in the derivation without secrecy due to Bullock et al.~\cite{10886999} (see Proof of Lemma 14 therein, Eq. (86)).
The difference, however, is that $\mathcal{M}_2$ in \cite{10886999} is not a message component, but rather a pre-shared key. 
Consequently, 
$\log\abs{\mathcal{M}}=\log\abs{\mathcal{M}_1}+\log\abs{\mathcal{M}_2}$ scales as $\sim \sqrt{n} D(\omega_1||\omega_0)$ in \cite{10886999}, whereas here, $\log\abs{\mathcal{M}}$ scales as $\sim \sqrt{n} D(\sigma_1||\sigma_0)$, with respect to Bob's outputs $\sigma_x$, instead of Willie's outputs $\omega_x$
(see \eqref{Equation:Secrecy_Message_Size}). 
Nonetheless, based on our assumption in \eqref{Equation:Secrecy_Case_2} and \eqref{Equation:Secrecy_Message_Size}, we have
\begin{align}% 
\log\left({\abs{\mathcal{M}_1}\abs{ \mathcal{M}_2}}\right) \geq (1 + \zeta_n)\gamma_n\sqrt{n}\RelEntropy{\omega_1}{\omega_0}
\,.
\end{align}
Therefore, the covertness bound \eqref{eq:Theorem1Covertness} immediately follows from \cite{10886999}.

\subsection*{Secrecy}
The main novelty of our analysis is in the derivation of secrecy, which was not considered in \cite{10886999}.
To ensure secrecy, we combine the OTP and binning techniques. 

\subsubsection{One-time pad}% 
The message component $m_2$ is encrypted % 
by a one-time pad code, requiring a key of length % 
\begin{align}
    \log\abs{\mathcal{K}}=\log\abs{\mathcal{M}_2}=(1 + \zeta_n)\gamma_n\sqrt{n}\RelEntropy{\omega_1}{\omega_0}
    \,.
    \label{Equation:SecrecyKeySize}
\end{align}
That is, the encoder replaces $m_2$ by $m_2\oplus k$, where $k$ is the pre-shared key.
Therefore, Willie's average output state is
\begin{align}
\rho_{W^n}^{(m_1,m_2)}=\frac{1}{\abs{\mathcal{K}}}\sum_{k\in\mathcal{K}}
\omega_{\mathbf{c}(m_1,m_2\oplus k)}
&=\frac{1}{\abs{\mathcal{K}}}\sum_{k\in\mathcal{K}}
\omega_{\mathbf{c}(m_1, k)}
\,.
\label{Equation:Secrecy_Conditional_Willie}
\end{align} % 
This ensures perfect secrecy for the message component $m_2$.

\subsubsection{Binning}% 
We divide the random code $\mathscr{C}$ into $\abs{\mathcal{M}_1}$ bins,
$\mathscr{C}_{m_1}$ for 
$m_1\in\mathcal{M}_1$,
each of size $\abs{\mathcal{M}_2}=\abs{\mathcal{K}}$. Note that here we treat $m_2$ as a ``junk'' variable with respect to $m_1$.
That is, $m_2$ is the codeword index within each bin $\mathscr{C}_{m_1}$.

Let $(m_1,m_2)\in
\mathcal{M}_1\times\mathcal{M}_2$. By   \eqref{Equation:Secrecy_Conditional_Willie},
\begin{align}
    &% 
    \TrNorm{\rho_{W^n}^{(m_1,m_2)} - \omega_{\alpha_n}^{\otimes n}}% 
    =% 
    \TrNorm{\frac{1}{\abs{\mathcal{K}}}\sum_{k\in\mathcal{K}}
\omega_{\mathbf{c}(m_1, k)} - \omega_{\alpha_n}^{\otimes n}}% 
\,.
\end{align}
In order to establish secrecy, we apply the quantum channel resolvability from Lemma~\ref{Lemma:Quantum_channel_resolvability}, 
\begin{align}
&\mathbb{E}_{\mathscr{C}}\left[\TrNorm{\frac{1}{\abs{\mathcal{K}}}\sum_{k\in\mathcal{K}}
\omega_{\mathbf{c}(m_1, k)} - \omega_{\alpha_n}^{\otimes n}}\right]
\nonumber\\
    &\leq 2\sqrt{% 
    \exp\left[
    \Tilde{\beta}_n \Tilde{s}_n + n{\phi}(\Tilde{s}_n,\alpha_n)
    \right]
    } + \sqrt{\frac{e^{\Tilde{\beta}_n} \nu_n}{\abs{\mathcal{K}}}} 
    \,.
\label{Equation:Secrecy_Resolvability}
\end{align}
In simple terms, quantum channel resolvability shows that the target state $\omega_{\alpha_n}^{\otimes n}$ can be well approximated by a uniform average over the key set $\mathcal{K}$ of the actual states, and provides a bound on the resulting trace distance between the two states, which vanishes as $\abs{\mathcal{K}}$ increases.

Consider the first term on the right-hand side of  \eqref{Equation:Secrecy_Resolvability}. 
By 
Lemma~\ref{Lemma:phi_s_alpha_bound}, 
we have
\begin{align}
    &\Tilde{\beta}_n \Tilde{s}_n + n{\phi}(\Tilde{s}_n, \alpha_n) 
    \nonumber\\
    &\leq \Tilde{\beta}_n \Tilde{s}_n + n\left(-\alpha_n \Tilde{s}_n D(\omega_1||\omega_0) + \vartheta_1\alpha_n \Tilde{s}^{2}_n - \vartheta_2 \Tilde{s}^{3}_n\right)
    \,.
\label{Equation:SecrecyCondition1stTermInequality}
\end{align}
Then,  set
\begin{align}
    \Tilde{\beta}_n &= \left(1 + \frac{\zeta_n}{2}\right)\alpha_n n D(\omega_1||\omega_0)
    \,,
\label{Equation:SecrecyCondition1stTermBeta}
\\
    \Tilde{s}_n &= 
    -\sqrt{\frac{\alpha_n\zeta_n D(\omega_1||\omega_0)}{4\vartheta_2}}
    \,.
    \label{Equation:s_n_tilde}
\end{align}
We note that $\alpha_n\cdot\zeta_n \in o(n^{-\frac{1}{2}})$ for % 
$\alpha_n$ and $\zeta_n$ as in Proposition ~\ref{Proposition:classical_covert_and_secret}. Therefore, 
 for large enough $n$, the condition
$\Tilde{s}_n \geq  \Tilde{s}_0$ in Lemma~\ref{Lemma:phi_s_alpha_bound} holds, for any choice of a negative constant $\Tilde{s}_0 < 0$.

Substituting \eqref{Equation:SecrecyCondition1stTermBeta}-\eqref{Equation:s_n_tilde} into the right-hand side of \eqref{Equation:SecrecyCondition1stTermInequality} yields
\begin{align}
    &\Tilde{\beta}_n \Tilde{s}_n + n{\phi}(\Tilde{s}_n, \alpha_n) 
    \nonumber\\
    &\leq \Tilde{s}_n\alpha_n n \left(\frac{\zeta_n}{2} D(\omega_1||\omega_0) + \vartheta_1 \Tilde{s}_n - \vartheta_2 \Tilde{s}^2_n\alpha^{-1}_n \right)
    \nonumber\\
    &= \Tilde{s}_n \alpha_n n \left(\frac{\zeta_n}{4}D(\omega_1||\omega_0) + \vartheta_1\Tilde{s}_n\right) \nonumber \\
    &= -\sqrt{\frac{\alpha_n\zeta_n D(\omega_1||\omega_0)}{4\vartheta_2}}\alpha_n n \left(\frac{\zeta_n}{4}D(\omega_1||\omega_0) + \vartheta_1\Tilde{s}_n\right) 
\end{align}
and thus, for sufficiently large $n$,
\begin{align}
    &\exp\left[% 
    \Tilde{\beta}_n \Tilde{s}_n + n{\phi}(\Tilde{s}_n,\alpha_n)
    \right]
    \nonumber\\
    &\leq 
    \exp\left[ % 
    -\sqrt{\frac{\alpha_n\zeta_n D(\omega_1||\omega_0)}{4\vartheta_2}}\alpha_n n \left(\frac{\zeta_n}{4}D(\omega_1||\omega_0) + \vartheta_1\Tilde{s}_n\right)
    \right] % 
    \nonumber \\
    &\leq % 
    \exp\left[-\left(\frac{1}{16}\sqrt{\frac{(D(\omega_1||\omega_0))^3}{\vartheta_2}}\right)\alpha_n^{\frac{3}{2}}\zeta_n^{\frac{3}{2}}n
    \right] % 
    \,.
    \label{Equation:SecrecyCondition1stTermExponent}
\end{align}
Next, we bound the second term on the right-hand side of \eqref{Equation:Secrecy_Resolvability}. As we set $\Tilde{\beta}_n $ according to \eqref{Equation:SecrecyCondition1stTermBeta}, we can write this term as % 
\begin{align}
    \frac{e^{\Tilde{\beta}_n} \nu_n}{\abs{\mathcal{K}}} 
    = \frac{e^{\left(1 + \frac{\zeta_n}{2}\right)\alpha_n n D(\omega_1||\omega_0)}\nu_n}{\abs{\mathcal{K}}}
    \,.
\end{align}
According to Lemma ~\ref{Lemma:omega_alpha_n_number_of_distinct_eigenvalues},  the number of distinct eigenvalues of $\omega_{\alpha_n}^{\otimes n}$ is bounded by $\nu_n\leq(n+1)^{\dim{\mathcal{H}_W}}$. Hence,
\begin{align}
    \frac{e^{\Tilde{\beta}_n} \nu_n}{\abs{\mathcal{K}}}
    \leq \frac{e^{\left(1 + \frac{\zeta_n}{2}\right)\alpha_n n D(\omega_1||\omega_0)}(n+1)^{d_W}}{\abs{\mathcal{K}}}
\end{align}
where $d_W\equiv \dim{\mathcal{H}_W}$.
Thus, for the key size in \eqref{Equation:SecrecyKeySize}, we have 
\begin{align}
    \frac{e^{\Tilde{\beta}_n} \nu_n}{\abs{\mathcal{K}}} 
    &\leq \frac{e^{\left(1 + \frac{\zeta_n}{2}\right)\alpha_n n D(\omega_1||\omega_0)}(n+1)^{\dim{\mathcal{H}_W}}}{e^{(1 + \zeta_n)\alpha_n n D(\omega_1||\omega_0)}} \nonumber \\
    &= % 
    \exp\left[-\frac{\zeta_n}{2}\alpha_n n D(\omega_1||\omega_0) + d_W\log(n+1)% 
    \right]
    \,.
    \label{Equation:SecrecyCondition2ndTermExponent}
\end{align}
Substituting \eqref{Equation:SecrecyCondition1stTermExponent} and \eqref{Equation:SecrecyCondition2ndTermExponent} into the secrecy bound \eqref{Equation:Secrecy_Resolvability} yields
\begin{align}
    &\mathbb{E}_{\mathscr{C}}\left[\TrNorm{
    \frac{1}{\abs{\mathcal{K}}}\sum_{k\in\mathcal{K}}
\omega_{\mathbf{c}(m_1, k)}
    - \omega_{\alpha_n}^{\otimes n}}\right] 
    \nonumber\\
    &\leq 2% 
    \exp\left[
    -\frac{1}{32}\left(\sqrt{\frac{(D(\omega_1||\omega_0))^3}{\vartheta_2}}\right)\alpha_n^{\frac{3}{2}}\zeta_n^{\frac{3}{2}}n
    \right] % 
    \nonumber\\&
    + \exp\left[% 
    -\frac{1}{2}\left(\frac{\zeta_n}{2}\alpha_n n D(\omega_1||\omega_0) + d_W\log(n+1)\right)\right]
    \,.
    \label{Equation:Resolvability_Omega_alpha_1}
\end{align}
Thus, we conclude that there exists a sequence $\zeta_n^{(3)} \in \omega((\log n)^{-1})$, for large enough $n$, such that
\begin{align}
    \mathbb{E}_{\mathscr{C}}\left[\TrNorm{\frac{1}{\abs{\mathcal{K}}}\sum_{k\in\mathcal{K}}
\omega_{\mathbf{c}(m_1, k)}
    - \omega_{\alpha_n}^{\otimes n}}\right]
    &\leq e^{-3\zeta_n^{(3)}\alpha_n^{\frac{3}{2}}n} \nonumber \\
    &= e^{-3\zeta_n^{(3)}\gamma_n^{\frac{3}{2}}n^{\frac{1}{4}}}
    \label{Equation:Resolvability_Omega_alpha_2}
\end{align}
where the equality holds since $\alpha_n=\frac{\gamma_n}{\sqrt{n}}$ in Proposition ~\ref{Proposition:classical_covert_and_secret}.
Finally, from linearity we can bound also the average leakage,
\begin{align}
    &\mathbb{E}_{\mathscr{C}}\left[\frac{1}{\abs{\mathcal{M}_1}}\sum_{m_1\in\mathcal{M}_1}\TrNorm{\frac{1}{\abs{\mathcal{K}}}\sum_{k\in\mathcal{K}}
\omega_{\mathbf{c}(m_1, k)}
    - \omega_{\alpha_n}^{\otimes n}}\right]
    \nonumber\\
    &=\frac{1}{\abs{\mathcal{M}_1}}\sum_{m_1\in\mathcal{M}_1} \mathbb{E}_{\mathscr{C}}\left[\TrNorm{\frac{1}{\abs{\mathcal{K}}}\sum_{k\in\mathcal{K}}
\omega_{\mathbf{c}(m_1, k)}
    - \omega_{\alpha_n}^{\otimes n}}\right] \nonumber \\
    &\leq e^{-3\zeta_n^{(3)}\gamma_n^{\frac{3}{2}}n^{\frac{1}{4}}}
    \,.
\end{align}
By the standard Markov inequality argument, it follows that  there exists a deterministic codebook that satisfies
the error, covertness and secrecy requirements in \eqref{Equation:Random_Code_Bounds}.
This completes the proof of 
Proposition~\ref{Proposition:classical_covert_and_secret}.

Finally, we observe that we have shown achievability of approximately the same 
number of information bits
as without secrecy \cite{SheikholeslamiB16,10886999}, albeit with a larger key (see Remark~\ref{Remark:Secrecy_Rate}).
Proposition~\ref{Proposition:classical_covert_and_secret}
 proves the existence of a reliable, covert and secret communication scheme such
that \eqref{eq:Theorem1Covertness} holds. Thus,
\begin{align}
    \RelEntropy{\overline{\rho}_{W^n}}{\omega_0^{\otimes n}} 
    \leq \RelEntropy{\omega_{\alpha_n}^{\otimes n}}{\omega_0^{\otimes n}} + e^{-\zeta_n^{(2)}\gamma^{\frac{3}{2}}_n n^\frac{1}{4}}
    \,.
\end{align}
Furthermore, by Lemma~\ref{Lemma:avg_QRE_to_chi_square} (see \cite[Lemma 5]{9344627}), we have
\begin{align}
    D(\omega_{\alpha_n}^{\otimes n}||\omega_0^{\otimes n})
    = \frac{1}{2}\gamma_n^2\chi^2(\omega_1||\omega_0) + 
    {O}\left(\frac{\gamma_n^3}{\sqrt{n}}\right)
\end{align}
where $\alpha_n = \frac{\gamma_n}{\sqrt{n}}$.
Therefore, the secrecy rate satisfies % 
\begin{align}
    &\frac{\log{\abs{\mathcal{M}}}}{\sqrt{n D(\overline{\rho}_{W^n}||\omega_0^{\otimes n})}}
    \nonumber\\
    &\geq % 
    \frac{(1 -\zeta_n)\gamma_n\sqrt{n}\RelEntropy{\sigma_1}{\sigma_0}}{\sqrt{\frac{1}{2}n\gamma_n^2\chi^2(\omega_1||\omega_0) + ne^{-\zeta_n^{(2)}\gamma^{\frac{3}{2}}_n n^\frac{1}{4}} +  
    {O}\left(\sqrt{n}\gamma_n^3\right)}}
    \,.
\end{align}
In the limit of $n\to\infty$, we achieve the covert secrecy rate
    \begin{align}
    L_\text{S}
    \geq \frac{\RelEntropy{\sigma_1}{\sigma_0}
    }{\sqrt{\frac{1}{2}\chi^2(\omega_1||\omega_0)}}
\end{align}
with key assistance. 

 The converse part for the key-assisted capacity theorem follows immediately, since the secrecy capacity is bounded from above by the capacity without secrecy. 
\end{proof}

\section{Minimal Key Rate}
\label{Section:Minimal_Key}

\subsection{Technical Lemmas}% 
\label{Appendix:key_converse_logM_upper_bound}
{
Before we dive into the proof, we need the following two lemmas for our converse analysis: First, we show that % 
the number of information bits, 
$\log{\abs{\mathcal{M}}}$, satisfies the following upper bound.
\begin{lemma}
    \label{Lemma:key_converse_logM_upper_bound}
    Let $\text{supp}(\sigma_1) \subseteq \text{supp}(\sigma_0)$. For decoding error probability $\varepsilon_n$, 
    the number of information bits
    can be upper bounded by
    \begin{align}
        \log\abs{\mathcal{M}} 
        &\leq \frac{1}{1 - \varepsilon_n}\left(n\mu_nD(\sigma_1||\sigma_0) + 1\right)
    \end{align}
    where $\mu_n \in o(1)$ is the average probability of transmitting a non-innocent symbol.
\end{lemma}
Second, we show that % 
the mutual information between the classical input to the channel and Willie's output state satisfies the following lower bound.
\begin{lemma}
    \label{Lemma:key_converse_I(x^n;W^n)_upper_bound}
    Let $\text{supp}(\omega_1) \subseteq \text{supp}(\omega_0)$. For covertness constraint $\delta_n^{\text{cov}}$,  
    the mutual information between the classical input to the channel and Willie's output state can be bounded from below by
    \begin{align}
        I(\mathbf{x}^n;W^n)_{\overline{\rho}} \geq n\mu_n D(\omega_1 \| \omega_0) - 2\delta_n^{\text{cov}}
    \end{align}
    provided that $\lim_{n \to \infty} \sqrt{n}\mu_n = 0$, where $\mu_n \in o(1)$ is the average probability of transmitting a non-innocent symbol.
\end{lemma}
}

\subsubsection{Proof of Lemma~\ref{Lemma:key_converse_logM_upper_bound}}
We start with the proof of the first lemma, % 
which establishes an upper bound on the number of information bits that can be transmitted reliably in our setting, expressed in terms of Bob’s distinguishability, the average number of non-innocent transmitted symbols, and the decoding error bound.
We follow similar steps % 
to those used in the proof of \cite[Th. 2]{10886999}.

Let $\mathbf{m}$ denote a uniform message,
 $\mathbf{k}$  the shared secret key,
and $\mathbf{x}^n\in \{0,1\}^n$  the channel
input.
Fano’s inequality \cite[Th. 10.7.3]{W2017} yields
\begin{align}
    \log{\abs{\mathcal{M}}}
    &\leq I(\mathbf{m} ; B^n \mathbf{k})_\rho + 1 + \varepsilon_n\log{\abs{\mathcal{M}}}
    \label{equation:key_converse_logM_upper_bound_first_bound_4}
    \,.
\end{align}
Since the message and the pre-shared secret key are independent, this can be written as
\begin{align}
    \log{\abs{\mathcal{M}}}
    &\leq I(\mathbf{m} ; B^n | \mathbf{k})_\rho + 1 + \varepsilon_n\log{\abs{\mathcal{M}}}
    \nonumber\\
    &\leq I(\mathbf{m} \mathbf{k}; B^n)_\rho + 1 + \varepsilon_n\log{\abs{\mathcal{M}}}
    \label{equation:key_converse_logM_upper_bound_second_bound_3}
    \,.
\end{align}
Then, by the non-negativity of the c-q conditional mutual information, we obtain
\begin{align}
    \log{\abs{\mathcal{M}}}
    &\leq I(\mathbf{m} \mathbf{k}; B^n)_\rho + I(\mathbf{x}^n ; B^n | \mathbf{m} \mathbf{k})_\rho + 1 + \varepsilon_n\log{\abs{\mathcal{M}}}
    \nonumber\\
    &= I(\mathbf{m} \mathbf{k} \mathbf{x}^n; B^n)_\rho + 1 + \varepsilon_n\log{\abs{\mathcal{M}}}
    \nonumber\\
    &= I(\mathbf{x}^n; B^n)_\rho + 1 + \varepsilon_n\log{\abs{\mathcal{M}}}
    \label{equation:key_converse_logM_upper_bound_third_bound_1}
\end{align}
where the last equality follows because $H(B^n | \mathbf{x}^n \mathbf{m} \mathbf{k})_\rho = H(B^n | \mathbf{x}^n)_\rho$ by the definition of the model.
By subadditivity of von Neumann entropy, it follows that % 
\begin{subequations}
    \begin{align}
        \log{\abs{\mathcal{M}}}
        &\le \sum_{j=1}^n
        H(B_j)_\rho % 
        - H(B^n|\mathbf{x}^n)_\rho % 
        + 1 + \varepsilon_n \log{\abs{\mathcal{M}}}
        \label{equation:key_converse_logM_upper_bound_forth_bound_1}
        \\
        &= \sum_{j=1}^n
        \left(
       H(B_j)_\rho-H(B_j|\mathbf{x}_j)_\rho% 
        \right)
        + 1 + \varepsilon_n \log{\abs{\mathcal{M}}}
        \label{equation:key_converse_logM_upper_bound_forth_bound_2}
        \\
        &= \sum_{j=1}^n
        I\left( \mathbf{x}_j; B_j \right)_\rho
        + 1 + \varepsilon_n \log{\abs{\mathcal{M}}}
        \label{equation:key_converse_logM_upper_bound_forth_bound_3}
    \end{align}
\end{subequations}
where \eqref{equation:key_converse_logM_upper_bound_forth_bound_2} holds since $\rho_{B^n}^{x^n}=\bigotimes_{j=1}^n\sigma_{x_j}$ is a product state for a given $x^n$. % 
Let $\mathbf{t}$ be a uniformly distributed index over $\{1,\ldots,n\}$, and observe that 
\begin{align}
    \frac{1}{n}\sum_{j=1}^n I\left( \mathbf{x}_j; B_j \right)_\rho
    &=I\left( \mathbf{x}_{\mathbf{t}}; B_{\mathbf{t}}| {\mathbf{t}} \right)_\rho
    \nonumber\\
    &\leq I\left( \tilde{\mathbf{x}}; B \right)_\rho
\end{align}
where  $\tilde{\mathbf{x}}\equiv(\mathbf{x}_{\mathbf{t}},\mathbf{t})$ and 
$\rho_{B}^{\tilde{x}}=\rho_B^{x_j,j}=\sigma_{x_j}$.
Hence,
\begin{align}
    \log{\abs{\mathcal{M}}}
    &\leq \frac{1}{1 - \varepsilon_n} \left( n \, I\left( \tilde{\mathbf{x}}; B \right)_\rho + 1 \right)
    \,.
    \label{equation:key_converse_logM_upper_bound_sixth_bound_1}
\end{align}
Expanding the mutual information yields
\begin{align}
    &I\left( \tilde{\mathbf{x}}; B \right)_\rho
    \nonumber\\
    &=H(\rho_B)-(1-\mu_n)H(\sigma_0)-\mu_n H(\sigma_1)
    \nonumber\\
    &=H(\rho_B)+\trace\left\{\rho_B \log \sigma_0\right\}-\trace\left\{\rho_B \log \sigma_0\right\}
    \nonumber\\
    &\phantom{=} -(1-\mu_n)H(\sigma_0)-\mu_n H(\sigma_1)
    \nonumber\\
    &=- D(\rho_B \| \sigma_0)-(1-\mu_n)\trace\left\{\sigma_0 \log \sigma_0\right\}
    \nonumber\\
    &\phantom{=} -\mu_n \trace\left\{\sigma_1 \log \sigma_0\right\}-(1-\mu_n)H(\sigma_0)-\mu_n H(\sigma_1)
    \nonumber\\
    &=\mu_n D(\sigma_1 \| \sigma_0) - D(\rho_B \| \sigma_0)
    \nonumber\\
    &\le \mu_n D(\sigma_1 \| \sigma_0)
    \label{equation:key_converse_Holevo_info_upper_bound}
\end{align}
where $\rho_B=\sum_{\tilde{x}} p_{\tilde{\mathbf{x}}}(\tilde{x}) \rho_B^{\tilde{x}}$ is Bob's average output state,  $\mu_n \equiv 1 - p_{\tilde{\mathbf{x}}}(0)$ is the average probability of transmitting
a non-innocent symbol, and the inequality follows from the QRE being non-negative. 
Finally, combining \eqref{equation:key_converse_logM_upper_bound_sixth_bound_1} and \eqref{equation:key_converse_Holevo_info_upper_bound} yields
\begin{align}
    \log{\abs{\mathcal{M}}}
    &\le \frac{1}{1 - \varepsilon_n}
    \left(
    n \mu_n D(\sigma_1 \| \sigma_0) + 1
    \right)
    \label{equation:key_converse_logM_upper_bound_final_bound}
\end{align}
where $D(\sigma_1 \| \sigma_0) < \infty$ since $\mathrm{supp}(\sigma_1) \subseteq \mathrm{supp}(\sigma_0)$.
\qed

\subsubsection{Proof Of Lemma~\ref{Lemma:key_converse_I(x^n;W^n)_upper_bound}}
\label{Appendix:key_converse_I(x^n;W^n)_lower_bound}
We now provide the proof of the second lemma, % 
which establishes a lower bound on the mutual information between the classical channel input and Willie’s output state, expressed in terms of Willie’s distinguishability, the average number of non-innocent transmitted symbols, and the covertness bound.
We follow similar steps % 
 to those used in the proof of \cite[Th. 2]{10886999}.

First, the covertness bound yields
\begin{align}
    &I(\mathbf{x}^n;W^n)_{\overline{\rho}}
    \nonumber\\
    &\geq I(\mathbf{x}^n;W^n)_{\overline{\rho}} % 
    + \RelEntropy{\overline{\rho}_{W^n}}{\omega_0^{\otimes n}} 
    - \delta_n^{\text{cov}}
    \nonumber\\
    &= H(W^n)_{\overline{\rho}} -  H(W^n|\mathbf{x}^n)_{\overline{\rho}}
    + D\left(\overline{\rho}_{W^n} || \omega_0^{\otimes n} \right)
    - \delta_n^{\text{cov}}
    \,.
    \label{equation:key_converse_mut_info_lower_bound_first_bound_3}
\end{align}
Expanding the QRE in \eqref{equation:key_converse_mut_info_lower_bound_first_bound_3}, we obtain
\begin{subequations}
    \begin{align}
        &I(\mathbf{x}^n;W^n)_{\overline{\rho}}
        \nonumber\\
        &\geq H(W^n)_{\overline{\rho}} - H(W^n|\mathbf{x}^n)_{\overline{\rho}}
        -H\left(\overline{\rho}_{W^n}\right)
        \nonumber\\
        &\phantom{=}
        -\Tr{\sum_{x^n} p_{\mathbf{x}^n}(x^n)\, \omega_{x^n} \log{\omega_0^{\otimes n}}}
        - \delta_n^{\text{cov}}
        \label{equation:key_converse_mut_info_lower_bound_second_bound_1}
        \\
        &= - H(W^n|\mathbf{x}^n)_{\overline{\rho}}
        -\sum_{x^n} p_{\mathbf{x}^n}(x^n) \Tr{\omega_{x^n} \log{\omega_0^{\otimes n}}}
        - \delta_n^{\text{cov}}
        \label{equation:key_converse_mut_info_lower_bound_second_bound_2}
        \\
        &= 
        - \sum_{j=1}^n H(W_j|\mathbf{x}_j)_{\overline{\rho}}
        - \sum_{x^n} p_{\mathbf{x}_j}(x) \sum_{j=1}^n \Tr{\omega_{x_j}\log \omega_0}
        - \delta_n^{\text{cov}}
        \label{equation:key_converse_mut_info_lower_bound_second_bound_3}
        \\
        &= \sum_{x} \sum_{j=1}^n p_{\mathbf{x}_j}(x)
        \left(
        - H(\omega_x)
        - \Tr{\omega_x \log \omega_0}
        \right)
        - \delta_n^{\text{cov}}
        \label{equation:key_converse_mut_info_lower_bound_second_bound_4}
    \end{align}
\end{subequations}
where 
\eqref{equation:key_converse_mut_info_lower_bound_second_bound_3} follows because all the states there are product states. 
We continue:
\begin{subequations}
    \begin{align}
        &I(\mathbf{x}^n;W^n)_{\overline{\rho}}
        \nonumber\\
        &\geq n\sum_{x} p_{\tilde{\mathbf{x}}}(x)
        \left(
        - H(\omega_x)
        - \Tr{\omega_x \log \omega_0}
        \right)
        - \delta_n^{\text{cov}}
        \label{equation:key_converse_mut_info_lower_bound_third_bound_1}
        \\
        &= -nH(W|\tilde{\mathbf{x}})_{\overline{\rho}}
        - n\Tr{{\overline{\rho}}_W \log \omega_0}
        - \delta_n^{\text{cov}}
        \label{equation:key_converse_mut_info_lower_bound_third_bound_2}
        \\
        &\geq -nH(W|\tilde{\mathbf{x}})_{\overline{\rho}}
        - n\Tr{{\overline{\rho}}_W \log \omega_0}
        -n\RelEntropy{{\overline{\rho}}_W}{\omega_0}
        - \delta_n^{\text{cov}}
        \label{equation:key_converse_mut_info_lower_bound_third_bound_3}
        \\
        &= -nH(W|\tilde{\mathbf{x}})_{\overline{\rho}}
        +nH\left({\overline{\rho}}_W\right)
        - \delta_n^{\text{cov}}
        \label{equation:key_converse_mut_info_lower_bound_third_bound_4}
        \\
        &= nI\!\left( \tilde{\mathbf{x}}; W  \right)_{\overline{\rho}} - \delta_n^{\text{cov}}
        \label{equation:key_converse_mut_info_lower_bound_third_bound_5}
    \end{align}
\end{subequations}
having defined  $\tilde{\mathbf{x}}\equiv(\mathbf{x}_{\mathbf{t}},\mathbf{t})$ for a uniformly distributed index $\mathbf{t}$ over $\{1,\ldots,n\}$,  and 
${\rho}_{W}^{\tilde{x}}={\rho}_W^{x_j,j}=\omega_{x_j}$,
where 
\eqref{equation:key_converse_mut_info_lower_bound_third_bound_3} follows from the non-negativity of QRE, and \eqref{equation:key_converse_mut_info_lower_bound_third_bound_5} follows from the definition of mutual information.
Similarly to \eqref{equation:key_converse_Holevo_info_upper_bound}, expanding the mutual information yields
\begin{align}
    I\!\left( \tilde{\mathbf{x}}; W  \right)_{\overline{\rho}}
    &=H(\overline{\rho}_W)-(1-\mu_n)H(\omega_0)-\mu_n H(\omega_1)
    \nonumber\\
    &= \mu_n D(\omega_1 \| \omega_0) - D(\overline{\rho}_W \| \omega_0)
\end{align}
where $\mu_n \equiv 1 - p_{\tilde{\mathbf{x}}}(0)$ is the average probability of transmitting
a non-innocent symbol.
Then, combining with \eqref{equation:key_converse_mut_info_lower_bound_third_bound_5}, we obtain
\begin{align}
    I(\mathbf{x}^n;W^n)_{\overline{\rho}} \geq n\mu_n D(\omega_1 \| \omega_0) - nD(\overline{\rho}_W \| \omega_0) - \delta_n^{\text{cov}}
    \,.
    \label{equation:key_converse_mut_info_lower_bound_forth_bound}
\end{align}
Next, we want to upper bound $D(\overline{\rho}_W \| \omega_0)$. From the covertness condition, we have,
\begin{subequations}
    \label{Equation:QRE_tilde_omega_upper_bound}
    \begin{align}
        \delta_n^{\text{cov}}
        &\geq \RelEntropy{\overline{\rho}_{W^n}}{\omega_0^{\otimes n}}
        \label{Equation:QRE_tilde_omega_upper_bound_1}
        \\
        &= -H(W^n)_{\overline{\rho}}
        -\Tr{\sum_{x^n} p_{\mathbf{x}^n}(x^n)\, \omega_{x^n} \log{\omega_0^{\otimes n}}}
        \label{Equation:QRE_tilde_omega_upper_bound_2}
        \\
        &\geq -\sum_{j=1}^n H(W_j)_{\overline{\rho}} -\Tr{\sum_{x^n} p_{\mathbf{x}^n}(x^n)\, \omega_{x^n} \log{\omega_0^{\otimes n}}}
        \label{Equation:QRE_tilde_omega_upper_bound_3}
    \end{align}
\end{subequations}
where 
\eqref{Equation:QRE_tilde_omega_upper_bound_3} follows from entropy sub-additivity. % 
By the mixed product property of Kronecker product, we can write the last term in \eqref{Equation:QRE_tilde_omega_upper_bound_3} as
\begin{align}
    &\Tr{\sum_{x^n} p_{\mathbf{x}^n}(x^n)\, \omega_{x^n} \log{\omega_0^{\otimes n}}}
    \nonumber\\
    &= \sum_{j=1}^n\Tr{\sum_{x} p_{\mathbf{x}_j}(x) \omega_{x} \log{\omega_0}}
    \,.
    \label{Equation:trace_decomposition_3}
\end{align}
Then, plugging in \eqref{Equation:trace_decomposition_3} into \eqref{Equation:QRE_tilde_omega_upper_bound_3} yields
\begin{subequations}
    \label{Equation:QRE_tilde_omega_upper_bound_second}
    \begin{align}
        \delta_n^{\text{cov}}
        &\geq \sum_{j=1}^n \left[-H\!\left( \overline{\rho}_{W_j} \right) -\Tr{\overline{\rho}_{W_j} \log{\omega_0}}\right]
        \label{Equation:QRE_tilde_omega_upper_bound_second_1}
        \\
        &= \sum_{j=1}^n \RelEntropy{\overline{\rho}_{W_j}}{\omega_0}
        \label{Equation:QRE_tilde_omega_upper_bound_second_2}
        \\
        &\geq n\RelEntropy{\overline{\rho}_W}{\omega_0}
        \label{Equation:QRE_tilde_omega_upper_bound_second_3}
    \end{align}
\end{subequations}
where 
\eqref{Equation:QRE_tilde_omega_upper_bound_second_3} follows from the convexity of the QRE.
Finally, combining \eqref{equation:key_converse_mut_info_lower_bound_forth_bound} and \eqref{Equation:QRE_tilde_omega_upper_bound_second_3}, we achieve the following lower bound,
\begin{align}
    I(\mathbf{x}^n;W^n)_{\overline{\rho}} \geq n\mu_n D(\omega_1 \| \omega_0) - 2\delta_n^{\text{cov}}
    \,.
\end{align}
It is important to note that for our result to be valid, we must show that still $\delta_n^{\text{cov}} \to 0$ for $n \to \infty$. By the quantum Pinsker’s inequality,
\begin{align}
    \delta_n^{\text{cov}}
    &\geq
    n\RelEntropy{\overline{\rho}_W}{\omega_0}
    \nonumber\\
    &\geq n\frac{\TrNorm{\overline{\rho}_W - \omega_0}^2}{2\log{2}}
    \nonumber\\
    &= n\frac{\TrNorm{[(1 - \mu_n)\omega_0 
    + \mu_n \omega_1] - \omega_0}^2}{2\log{2}}
    \nonumber\\
    &= n\mu_n^2\frac{\TrNorm{\omega_1 - \omega_0}^2}{2\log{2}}
    \,.
\end{align}
Thus, the following condition must be satisfied:
\begin{align}
    \lim_{n \to \infty} \sqrt{n}\mu_n = 0
    \,.
\end{align}
\qed % 

\subsection{Proof of Theorem~\ref{Theorem:Secrecy_Key_Converse} (Minimal Key Rate)}
\label{sec:Secrecy_Key_Converse_proof}
Consider a sequence of  covert secrecy coding schemes with key assistance such that, for $n$ channel uses, 
achieve a decoding error probability
$\overline{P}_e^{(n)} \leq \varepsilon_n$, 
satisfy the covertness constraint
$\RelEntropy{\overline{\rho}_{W^n}}{\omega_0^{\otimes n}} \leq \delta_n^{\text{cov}}$,
and ensure % 
secrecy
$% 
\frac{1}{\abs{\mathcal{M}}}\sum_{m\in\mathcal{M}}
\TrNorm{
    \rho^{(m)}_{W^n} - \breve{\rho}_{W^n}}
    \leq \delta_n^{\text{sec}}$,
with $\varepsilon_n \to 0$,  $\delta_n^{\text{cov}} \to 0$ and $\delta_n^{\text{sec}} \to 0$ as $n \to \infty$.
Let $\mathbf{m}$ denote a uniform message,
 $\mathbf{k}$  the shared secret key,
and $\mathbf{x}^n\in \{0,1\}^n$  the channel
input.
\begin{proof}[Proof of Theorem~\ref{Theorem:Secrecy_Key_Converse}]
First, we notice that
\begin{subequations}
\label{eq:logK_1st_lower_bound}
    \begin{align}
        \log{\abs{\mathcal{K}}} 
        &= H(\mathbf{k} | \mathbf{m})
        \label{eq:logK_1st_lower_bound_2}
        \\
        &\geq H(\mathbf{k} | \mathbf{m}) - H(\mathbf{k} | \mathbf{m} W^n)_{{\rho}}
        \label{eq:logK_1st_lower_bound_3}
        \\
        &= I(\mathbf{k} ; W^n | \mathbf{m})_{{\rho}}  
        \label{eq:logK_1st_lower_bound_4}
    \end{align}
\end{subequations}
where \eqref{eq:logK_1st_lower_bound_2} follows because the key $\mathbf{k}$ is uniformly distributed and statistically independent of the classical message $\mathbf{m}$, and \eqref{eq:logK_1st_lower_bound_3} follows because the conditional entropy of a c-q state is non-negative. % 
Next, we recall that  % 
the secrecy criterion guarantees
\begin{align}
    \frac{1}{\abs{\mathcal{M}}}\sum_{m\in\mathcal{M}}
    \TrNorm{
    \rho^{(m)}_{W^n} - \breve{\rho}_{W^n}}
    \leq 
    \delta_n^{\text{sec}}
    \,.
\end{align}
By the Alicki-Fannes-Winter inequality (entropy continuity)~\cite{Winter:16p}, this implies
\begin{align}
    I(\mathbf{m} ; W^n)_\rho 
    &\leq \delta_n^{\text{sec}}\log{\abs{\mathcal{M}}} + g_2(\delta_n^{\text{sec}})
    \label{eq:mutual_info_upper_bound}
\end{align}
where $g_2(x) \equiv (1+x)h_2\left( \frac{x}{1+x} \right)$, and $h_2(x)$ is the binary entropy (see \cite[Sec. 23.1.1]{W2017}).
Therefore,
    \begin{align}
        &\log{\abs{\mathcal{K}}}
        \nonumber\\
        &\geq 
        I(\mathbf{k};W^n|\mathbf{m})_\rho + I(\mathbf{m};W^n)_\rho - 
        \delta_n^{\text{sec}}\log{\abs{\mathcal{M}}} - g_2(\delta_n^{\text{sec}})
        \nonumber\\
        &= I(\mathbf{k}\mathbf{m};W^n)_\rho - 
        \delta_n^{\text{sec}}\log{\abs{\mathcal{M}}} - g_2(\delta_n^{\text{sec}})
        \label{eq:logK_2nd_lower_bound_3}
        \,.
    \end{align}
Now, since in Theorem~\ref{Theorem:Secrecy_Key_Converse} we assume a deterministic encoder $f_n$, where $\mathbf{x}^n = f_n(\mathbf{m}, \mathbf{k})$, we have
\begin{align}
    I(\mathbf{x}^n;W^n|\mathbf{m}\mathbf{k})_\rho 
    &= 0
    \,,
\end{align}
and thus
\begin{align}
    &\log{\abs{\mathcal{K}}}
    \nonumber\\
    &\geq I(\mathbf{k}\mathbf{m};W^n)_\rho + I(\mathbf{x}^n;W^n|\mathbf{m}\mathbf{k})_\rho  - 
    \delta_n^{\text{sec}}\log{\abs{\mathcal{M}}} - g_2(\delta_n^{\text{sec}})
    \nonumber\\
    &= I(\mathbf{x}^n\mathbf{m}\mathbf{k};W^n)_\rho - 
    \delta_n^{\text{sec}}\log{\abs{\mathcal{M}}} - g_2(\delta_n^{\text{sec}})
    \nonumber\\
    &\geq I(\mathbf{x}^n;W^n)_{\overline{\rho}} - 
    \delta_n^{\text{sec}}\log{\abs{\mathcal{M}}} - g_2(\delta_n^{\text{sec}})
    \nonumber\\
    &\geq n\mu_n D(\omega_1 \| \omega_0) - 2\delta_n^{\text{cov}} - \delta_n^{\text{sec}}\log{\abs{\mathcal{M}}} - g_2(\delta_n^{\text{sec}})
\end{align}
where % 
the last inequality follows from Lemma~\ref{Lemma:key_converse_I(x^n;W^n)_upper_bound},
for $\mu_n$ such
that $\lim_{n \to \infty} \sqrt{n}\mu_n = 0$.
Thus, dividing  by $\sqrt{n D(\overline{\rho}_{W^n}||\omega_0^{\otimes n})}$
gives
\begin{align}
 L_{\text{key}}
 &=   \frac{\log{\abs{\mathcal{K}}}}{\sqrt{n D(\overline{\rho}_{W^n}||\omega_0^{\otimes n})}} 
 \nonumber\\
    &\geq 
    \frac{n\mu_n}{\sqrt{n D(\overline{\rho}_{W^n}||\omega_0^{\otimes n})}}  \left[\RelEntropy{\omega_1}{\omega_0} - \frac{2\delta_n^{\text{cov}}}{n\mu_n} \right.
    \nonumber\\
    &\left.\phantom{===========} -
     \frac{g_2(\delta_n^{\text{sec}})}{n\mu_n} 
     - \delta_n^{\text{sec}} \cdot \frac{\log{\abs{\mathcal{M}}}}{n\mu_n}
     \right] 
    \,.
    \label{Equation:converse_logK_3rd_lower_bound}
\end{align}
By the upper bound  in Lemma~\ref{Lemma:key_converse_logM_upper_bound}, we have
\begin{align}
    \log\abs{\mathcal{M}} 
    &\leq 
    \frac{n\mu_n}{1 - \varepsilon_n}\left(D(\sigma_1||\sigma_0) + \frac{1}{n\mu_n}\right)
    \label{Equation:converse_message_upper_bound}
    \,.
\end{align}
Recall that we assume $\operatorname{supp}(\sigma_1) \subseteq \operatorname{supp}(\sigma_0)$ (in both 
Theorem~\ref{Theorem:Covert_Secrecy_Capacity} and
Theorem~\ref{Theorem:Secrecy_Key_Converse}). Hence,
 $D(\sigma_1||\sigma_0)$ is finite. As $\log{\abs{\mathcal{M}}} $ tends to infinity, the last bound implies
\begin{align}
    \lim_{n \to \infty} n\mu_n = \infty
    \,.
    \label{Equation:converse_key_n_mun}
\end{align}
The theorem considers the required key rate in order to achieve the following information rate:
\begin{align}
    L_\text{S} = \frac{\log{\abs{\mathcal{M}}}}{\sqrt{n D(\overline{\rho}_{W^n}||\omega_0^{\otimes n})}} 
    &\geq
    (1 - \vartheta_n)\frac{D(\sigma_1||\sigma_0)}{\sqrt{\frac{1}{2} \chi^2(\omega_1||\omega_0)}}
\end{align}
(see \eqref{Equation:L_S} and \eqref{Equation:key_converse_info_rate_bound_in_theorem}).
Together with \eqref{Equation:converse_message_upper_bound}, this yields 
\begin{align}
    \frac{n\mu_n}{\sqrt{n D(\overline{\rho}_{W^n}||\omega_0^{\otimes n})}} 
    &\geq \frac{(1 - \varepsilon_n)(1 - \vartheta_n)D(\sigma_1||\sigma_0)}{\sqrt{\frac{1}{2} \chi^2(\omega_1||\omega_0)}\left(D(\sigma_1||\sigma_0) + \frac{1}{n\mu_n}\right)}
    \,.
    \label{Equation:converse_lower_bound_nmu_n}
\end{align}
Then, by \eqref{Equation:converse_logK_3rd_lower_bound} with \eqref{Equation:converse_message_upper_bound}, % 
\begin{align}
    L_{\text{key}}
    &\geq \frac{(1 - \varepsilon_n)(1 - \vartheta_n)D(\sigma_1||\sigma_0)}{\sqrt{\frac{1}{2} \chi^2(\omega_1||\omega_0)}\left(D(\sigma_1||\sigma_0) + \frac{1}{n\mu_n}\right)} 
    \nonumber\\
    &\phantom{=}\times\left[\RelEntropy{\omega_1}{\omega_0} - \frac{2\delta_n^{\text{cov}}}{n\mu_n} \right.
    \nonumber\\
    &\left.\phantom{====} -
    \frac{g_2(\delta_n^{\text{sec}})}{n\mu_n} -  \frac{\delta_n^{\text{sec}}}{1 - \varepsilon_n}\cdot\left(D(\sigma_1||\sigma_0) + \frac{1}{n\mu_n}\right)\right] 
    \,.
\end{align}
 The expression on the right-hand side tends to 
\begin{align}
    \frac{D(\omega_1||\omega_0)}{\sqrt{\frac{1}{2}\chi^2(\omega_1||\omega_0)}} 
\end{align}
in the limit of ${n \to \infty} $,
by \eqref{Equation:converse_key_n_mun}. This completes the converse proof for the key rate.
\end{proof}

\section{Proof of Theorem~\ref{Theorem:Covert_Secrecy_Capacity_No_Key} (Covert Secrecy Capacity Without Assistance)}
\label{Appendix:Covert_Secrecy_Capacity_No_Key_Proof}

We give the proof of Theorem~\ref{Theorem:Covert_Secrecy_Capacity_No_Key}, which characterizes the covert secrecy capacity without key assistance. The proof establishes both the achievability and the converse parts, and builds upon the analytical tools and arguments developed in Sections~\ref{Section:Secrecy_Proof}-\ref{Section:Minimal_Key} above.

\subsection{Achievability Proof}

\subsubsection{Random Codebook Analysis}
First, we show achievability of secret and covert classical-quantum communication, using a random codebook.
\begin{proposition}[Random secrecy code]
\label{Proposition:classical_covert_and_secret_no_key}
    Consider a covert memoryless classical-quantum channel such that
    $\text{supp}(\sigma_1) \subseteq \text{supp}(\sigma_0)$, 
    $\text{supp}(\omega_1) \subseteq \text{supp}(\omega_0)$, and $\omega_1\neq\omega_0$. Let $\alpha_n = \frac{\gamma_n}{\sqrt{n}}$
    with $\gamma_n \in o(1) \cap \omega\left(\frac{(\log{n})^{\frac{7}{3}}}{n^{\frac{1}{6}}}\right)$.
    Then, 
    for any $\zeta_n \in o(1) \cap \omega\left((\log n)^{-\frac{2}{3}}\right)$, 
    there exist
    $\zeta_n^{(1)} \in \omega\left((\log n)^{-\frac{4}{3}}
    n^{-\frac{1}{3}}\right), \zeta_n^{(2)} \in \omega\left((\log n)^{-2}\right), \zeta_n^{(3)} \in
    \omega\left((\log n)^{-1}\right)$
    and a classical-quantum covert secrecy code with a random codebook $\mathscr{C}$, 
    {for the transmission of a secret message
    $m\in\mathcal{M}$ and a 
    public message
    $\ell\in\mathcal{L}$},
    such that, for $n$ sufficiently large:
    \begin{align}
        \log{\abs{\mathcal{M}}\abs{\mathcal{L}}} 
        &= (1 -\zeta_n)\gamma_n\sqrt{n}\RelEntropy{\sigma_1}{\sigma_0} \,, 
        \nonumber\\
        \log{\abs{\mathcal{L}}} 
        &= (1 + \zeta_n)\gamma_n\sqrt{n}\RelEntropy{\omega_1}{\omega_0}
        \label{eq:Theorem1LogMLogL}
    \end{align}
    and
\begin{subequations}
    \begin{align}
         &\mathbb{E}_{\mathscr{C}}\left\{\overline{P}_e^{(n)}\right\} \leq e^{-5\zeta_n^{(1)}\gamma_n\sqrt{n}},
         \label{eq:Theorem3Errors_no_key}
         \\
         &\mathbb{E}_{\mathscr{C}}\left\{\abs{\RelEntropy{\overline{\rho}_{W^n}}{\omega_0^{\otimes n}} - \RelEntropy{\omega_{\alpha_n}^{\otimes n}}{\omega_0^{\otimes n}}}\right\} \leq e^{-4\zeta_n^{(2)}\gamma^{\frac{3}{2}}_n n^\frac{1}{4}}, 
         \label{eq:Theorem3Covertness_no_key}
        \\
        &\mathbb{E}_{\mathscr{C}}\left\{\frac{1}{\abs{\mathcal{M}}}\sum_{m\in\mathcal{M}}\TrNorm{\rho^{(m)}_{W^n} - \omega_{\alpha_n}^{\otimes n}}\right\} 
        \leq  e^{-3\zeta_n^{(3)}\gamma_n^{\frac{3}{2}} n^{\frac{1}{4}}}
        \label{eq:Theorem3Secrecy_no_key}
    \end{align}
    \label{Equation:Random_Code_Bounds_no_key}
    \end{subequations}
where the expectation 
is with respect to the distribution of the random codebook $\mathscr{C}$ (see Subsection~\ref{Subsubsection:Secrecy_code}). % 
\end{proposition}

\begin{proof}[Proof of Proposition~\ref{Proposition:classical_covert_and_secret_no_key}]
Assume
\begin{align}
(1 - \zeta_n)\RelEntropy{\sigma_1}{\sigma_0} > (1 + \zeta_n)\RelEntropy{\omega_1}{\omega_0}
\,.
\label{Equation:Secrecy_Case_2_no_key}
\end{align}
We recall that without secrecy, covertness does not require a key, i.e. $\log{\abs{\mathcal{K}}} = 0$ (see \cite[Th. 1]{10886999}). % 
\subsection*{Code Construction}
Our codebook construction is based on the standard  binning technique.
\paragraph*{Classical codebook generation}
Select $\abs{\mathcal{M}} \abs{ \mathcal{L}}$ independent codewords $c(m,\ell) $,  $(m,\ell)\in\mathcal{M}\times \mathcal{L}$, % 
each  i.i.d. according to
$\text{Bernoulli}(\alpha_n)$.  
Reveal the codebook in public, to Alice, Bob, and Willie.

\paragraph*{Encoder}
{
Given the secret message $m \in \mathcal{M}$ and the public message $\ell \in \mathcal{L}$, transmit $x^n=c(m,\ell)$. % 
}

\paragraph*{Decoder}
Bob performs the square-root measurement
$\{{\Upsilon}_{m,\ell}\}$
for $(m,\ell)\in\mathcal{M}\times\mathcal{L}$, as in \eqref{Equation:Upsilon}. 

\subsection*{Error and covertness analysis}
The error and covertness derivations follow the previous results without secrecy in Bullock et al.~\cite{10886999}. % 
Bob's output state given the overall ``message'' $\breve{m}=(m,\ell)$, is 
$% 
    \sigma_{\mathbf{c}(\breve{m})} = \bigotimes_{i=1}^n \sigma_{\mathbf{c}_i(\breve{m})} 
$. % 
We denote the random codebook by 
$\mathscr{C}=\{\mathbf{c}(\breve{m})\}$. 
The average error probability  with respect to the codebook $\mathscr{C}$ % 
is then
\begin{equation}
    \overline{P}_e^{(n)}(\mathscr{C}) = \frac{1}{\abs{\breve{\mathcal{M}}}}\sum_{\breve{m}\in\breve{\mathcal{M}}}{\left(1 - \Tr{{\Upsilon}_{\breve{m}}\sigma_{\mathbf{c}(\breve{m})}}\right)} 
\end{equation}
as without secrecy, where we take $\breve{m}=(m,\ell)$ and 
$\breve{\mathcal{M}}\equiv \mathcal{M}\times \mathcal{L}$.
Therefore, by the same error analysis as in covert communication without secrecy 
\cite[Th. 1]{10886999}, the expected error probability is bounded as in \eqref{eq:Theorem3Errors_no_key},
 for
 \begin{equation}
     \log{\abs{\breve{\mathcal{M}}}} = \log{\abs{\mathcal{M}}% 
     \abs{\mathcal{L}}} = (1 -\zeta_n)\gamma_n\sqrt{n}\RelEntropy{\sigma_1}{\sigma_0} \,.
     \label{Equation:Secrecy_Message_Size_no_key}
 \end{equation}
We now consider the covertness.
If Alice has used the channel, 
Willie's average state is
$% 
    \overline{\rho}_{W^n} = \frac{1}{\abs{\mathcal{M}}\abs{\mathcal{L}}}\sum_{m,\ell}% 
    {\omega_{\mathbf{c}(m,\ell)}}
$, % 
where
$% 
     \omega_{\mathbf{c}(m,\ell)} = \bigotimes_{i=1}^{n}\omega_{\mathbf{c}_i(m,\ell)} 
$. % 
By quantum channel resolvability \cite% 
{hayashi2006quantum} (see Lemma~\ref{Lemma:Quantum_channel_resolvability}), for any $s_n \leq 0$ and $\beta_n \in \mathbb{R}$:
\begin{equation}
    \mathbb{E}_{\mathscr{C}}\left[\TrNorm{\overline{\rho}_{W^n}- \omega_{\alpha_n}^{\otimes n}}\right] 
    \leq 2\sqrt{e^{\beta_n s_n + n\phi(s_n,\alpha_n)}} + \sqrt{\frac{e^{\beta_n} \nu_n}{\abs{\mathcal{M}}\abs{\mathcal{L}}}}  \label{eq:Theorem3ProofCovertnessBound_no_key}
\end{equation}
where $\nu_n$ is the number of distinct eigenvalues of $\omega_{\alpha_n}^{\otimes n}$ 
as defined in \eqref{Equation:omega_alpha}, setting $\alpha = \alpha_n$, and $\phi(s_n,\alpha_n)$ as defined in % 
\eqref{Equation:phi_p}.
Note that the bound on the right-hand side of \eqref{eq:Theorem3ProofCovertnessBound_no_key} appeared 
in the derivation without secrecy due to Bullock et al.~\cite{10886999} (see Proof of Lemma 14 therein, Eq. (86)).
The difference, however, is that $\mathcal{L}$
in \cite{10886999} is not a 
{public message},
but rather a pre-shared key. 
Consequently, 
$\log\abs{\breve{\mathcal{M}}}=\log\abs{\mathcal{M}}+\log\abs{\mathcal{L}}$ scales as $\sim \sqrt{n} D(\sigma_1||\sigma_0)$, as in \cite{10886999}. % 
Hence, $\log\abs{\mathcal{M}}$ scales as $\sim \sqrt{n} [D(\sigma_1||\sigma_0)-D(\omega_1||\omega_0)]$ % 
(see \eqref{Equation:Secrecy_Message_Size_no_key}). 
Nonetheless, based on our assumption in \eqref{Equation:Secrecy_Case_2_no_key} and \eqref{Equation:Secrecy_Message_Size_no_key}, we have
\begin{align}% 
\log\left({\abs{\mathcal{M}}\abs{ \mathcal{L}}}\right) \geq (1 + \zeta_n)\gamma_n\sqrt{n}\RelEntropy{\omega_1}{\omega_0}
\,.
\end{align}
Thus, the covertness bound \eqref{eq:Theorem3Covertness_no_key} follows from \cite{10886999}.

\subsection*{Secrecy}
The main novelty of our analysis is in the derivation of secrecy, which was not considered in \cite{10886999}.
Set
\begin{align}
    \log\abs{\mathcal{L}}=(1 + \zeta_n)\gamma_n\sqrt{n}\RelEntropy{\omega_1}{\omega_0}
    \,.
    \label{Equation:SecrecyPublicMessageSize}
\end{align}
We divide the random code $\mathscr{C}$ into $\abs{\mathcal{M}}$ bins,
$\mathscr{C}_{m}$ for 
$m\in\mathcal{M}$,
each of size $\abs{\mathcal{L}}$. Here, unlike in the case with key assistance, the index $\ell$ is a 
{public message},
and \emph{not} a shared key. % 

Let $m\in
\mathcal{M}$. 
Then,
\begin{align}
    &% 
    \TrNorm{\rho_{W^n}^{(m)} - \omega_{\alpha_n}^{\otimes n}}% 
    =% 
    \TrNorm{\frac{1}{\abs{\mathcal{L}}}\sum_{\ell\in\mathcal{L}}
\omega_{\mathbf{c}(m, \ell)} - \omega_{\alpha_n}^{\otimes n}}% 
\,.
\end{align}
In order to establish secrecy, we apply quantum channel resolvability. From Lemma~\ref{Lemma:Quantum_channel_resolvability}, 
\begin{align}
&\mathbb{E}_{\mathscr{C}}\left[\TrNorm{\frac{1}{\abs{\mathcal{L}}}\sum_{\ell\in\mathcal{L}}
\omega_{\mathbf{c}(m, \ell)} - \omega_{\alpha_n}^{\otimes n}}\right] 
\nonumber\\
    &\leq 2\sqrt{% 
    \exp\left[
    \Tilde{\beta}_n \Tilde{s}_n + n{\phi}(\Tilde{s}_n,\alpha_n)
    \right]
    } + \sqrt{\frac{e^{\Tilde{\beta}_n} \nu_n}{\abs{\mathcal{L}}}} 
    \,.
\label{Equation:Secrecy_Resolvability_no_key}
\end{align}

Consider the first term on the right-hand side of  \eqref{Equation:Secrecy_Resolvability_no_key}. 
By 
Lemma~\ref{Lemma:phi_s_alpha_bound}, 
we have
\begin{align}
    &\Tilde{\beta}_n \Tilde{s}_n + n{\phi}(\Tilde{s}_n, \alpha_n) 
    \nonumber\\
    &\leq \Tilde{\beta}_n \Tilde{s}_n + n\left(-\alpha_n \Tilde{s}_n D(\omega_1||\omega_0) + \vartheta_1\alpha_n \Tilde{s}^{2}_n - \vartheta_2 \Tilde{s}^{3}_n\right)
    \,.
\label{Equation:SecrecyCondition1stTermInequality_no_key}
\end{align}
Then,  set
\begin{align}
    \Tilde{\beta}_n &= \left(1 + \frac{\zeta_n}{2}\right)\alpha_n n D(\omega_1||\omega_0)
    \,,
\label{Equation:SecrecyCondition1stTermBeta_no_key}
\\
    \Tilde{s}_n &= 
    -\sqrt{\frac{\alpha_n\zeta_n D(\omega_1||\omega_0)}{4\vartheta_2}}
    \,.
    \label{Equation:s_n_tilde_no_key}
\end{align}
Therefore, 
 for large enough $n$, the condition
$\Tilde{s}_n \geq  \Tilde{s}_0$ in Lemma~\ref{Lemma:phi_s_alpha_bound} holds, for any choice of a negative constant $\Tilde{s}_0 < 0$.

Substituting \eqref{Equation:SecrecyCondition1stTermBeta_no_key}-\eqref{Equation:s_n_tilde_no_key} into the right-hand side of \eqref{Equation:SecrecyCondition1stTermInequality_no_key}, 
yields
\begin{align}
    &\Tilde{\beta}_n \Tilde{s}_n + n{\phi}(\Tilde{s}_n, \alpha_n) 
    \nonumber\\
    &\leq \Tilde{s}_n\alpha_n n \left(\frac{\zeta_n}{2} D(\omega_1||\omega_0) + \vartheta_1 \Tilde{s}_n - \vartheta_2 \Tilde{s}^2_n\alpha^{-1}_n \right)
    \nonumber\\
    &= -\sqrt{\frac{\alpha_n\zeta_n D(\omega_1||\omega_0)}{4\vartheta_2}}\alpha_n n \left(\frac{\zeta_n}{4}D(\omega_1||\omega_0) + \vartheta_1\Tilde{s}_n\right) 
\end{align}
and thus, for sufficiently large $n$,
\begin{align}
    &\exp\left[% 
    \Tilde{\beta}_n \Tilde{s}_n + n{\phi}(\Tilde{s}_n,\alpha_n)
    \right]
    \nonumber\\
    &\leq % 
    \exp\left[-\left(\frac{1}{16}\sqrt{\frac{(D(\omega_1||\omega_0))^3}{\vartheta_2}}\right)\alpha_n^{\frac{3}{2}}\zeta_n^{\frac{3}{2}}n
    \right] % 
    \,.
    \label{Equation:SecrecyCondition1stTermExponent_no_key}
\end{align}
Next, we bound the second term on the right-hand side of \eqref{Equation:Secrecy_Resolvability_no_key}. Setting $\Tilde{\beta}_n $ according to \eqref{Equation:SecrecyCondition1stTermBeta_no_key}, % 
\begin{align}
    \frac{e^{\Tilde{\beta}_n} \nu_n}{\abs{\mathcal{L}}} 
    = \frac{e^{\left(1 + \frac{\zeta_n}{2}\right)\alpha_n n D(\omega_1||\omega_0)}\nu_n}{\abs{\mathcal{L}}}
    \,.
\end{align}
According to Lemma ~\ref{Lemma:omega_alpha_n_number_of_distinct_eigenvalues},  the number of distinct eigenvalues of $\omega_{\alpha_n}^{\otimes n}$ is bounded by $\nu_n\leq(n+1)^{\dim{\mathcal{H}_W}}$. Hence,
\begin{align}
    \frac{e^{\Tilde{\beta}_n} \nu_n}{\abs{\mathcal{L}}}
    \leq \frac{e^{\left(1 + \frac{\zeta_n}{2}\right)\alpha_n n D(\omega_1||\omega_0)}(n+1)^{d_W}}{\abs{\mathcal{L}}} \,.
\end{align}
Thus, the length of the {public message}
in \eqref{Equation:SecrecyPublicMessageSize} yields 
\begin{align}
    \frac{e^{\Tilde{\beta}_n} \nu_n}{\abs{\mathcal{L}}} 
    &\leq 
    \exp\left[-\frac{\zeta_n}{2}\alpha_n n D(\omega_1||\omega_0) + d_W\log(n+1)% 
    \right]
    \,.
    \label{Equation:SecrecyCondition2ndTermExponent_no_key}
\end{align}
Substituting \eqref{Equation:SecrecyCondition1stTermExponent_no_key} and \eqref{Equation:SecrecyCondition2ndTermExponent_no_key} into the secrecy bound \eqref{Equation:Secrecy_Resolvability_no_key} yields
\begin{align}
    &\mathbb{E}_{\mathscr{C}}\left[\TrNorm{
    \frac{1}{\abs{\mathcal{L}}}\sum_{\ell\in\mathcal{L}}
\omega_{\mathbf{c}(m, \ell)}
    - \omega_{\alpha_n}^{\otimes n}}\right] 
    \nonumber\\
    &\leq 2% 
    \exp\left[
    -\frac{1}{32}\left(\sqrt{\frac{(D(\omega_1||\omega_0))^3}{\vartheta_2}}\right)\alpha_n^{\frac{3}{2}}\zeta_n^{\frac{3}{2}}n
    \right] % 
    \nonumber\\
    &\phantom{=}
    + \exp\left[% 
    -\frac{1}{2}\left(\frac{\zeta_n}{2}\alpha_n n D(\omega_1||\omega_0) + d_W\log(n+1)\right)\right]
    \,.
\end{align}
Thus, we conclude that there exists a sequence $\zeta_n^{(3)} \in \omega((\log n)^{-1})$, for large enough $n$, such that
the average leakage is bounded by
\begin{align}
    \mathbb{E}_{\mathscr{C}}\left[\frac{1}{\abs{\mathcal{M}}}\sum_{m\in\mathcal{M}}\TrNorm{\frac{1}{\abs{\mathcal{L}}}\sum_{\ell\in\mathcal{L}}
\omega_{\mathbf{c}(m, \ell)}
    - \omega_{\alpha_n}^{\otimes n}}\right]
    &\leq e^{-3\zeta_n^{(3)}\gamma_n^{\frac{3}{2}}n^{\frac{1}{4}}}
    \,.
\end{align}
\end{proof}

\subsubsection{Derandomization}
We show that the same asymptotic rate can be achieved without common randomness. 
\begin{proposition}[Deterministic codebook]
\label{Proposition:Secrecy_Derandomization_no_key}
    Consider a covert memoryless classical-quantum channel such that
    $\text{supp}(\sigma_1) \subseteq \text{supp}(\sigma_0)$,
    $\text{supp}(\omega_1) \subseteq \text{supp}(\omega_0)$, and $\omega_1\neq\omega_0$.
    Let $\alpha_n = \frac{\gamma_n}{\sqrt{n}}$
     with $\gamma_n \in o(1) \cap \omega\left(\frac{(\log{n})^{\frac{7}{3}}}{n^{\frac{1}{6}}}\right)$. % 
    Then, there exists a 
    classical-quantum covert secrecy
    code with a {deterministic} codebook 
    $\mathscr{C}=\{x^n(m,\ell)\}$,
    {for the transmission of  a secret message
    $m\in\mathcal{M}$ and a 
    public message
    $\ell\in\mathcal{L}$},
    such that
    \begin{align}
        \log{\abs{\mathcal{M}}\abs{\mathcal{L}}} 
        &= (1 -\zeta_n)\gamma_n\sqrt{n}\RelEntropy{\sigma_1}{\sigma_0} \,, 
        \nonumber\\
        \log{\abs{\mathcal{L}}} 
        &= (1 + \zeta_n)\gamma_n\sqrt{n}\RelEntropy{\omega_1}{\omega_0}
        \label{eq:Theorem3LogMLogL_Deterministic_no_key}
    \end{align}
    and
\begin{subequations}
\begin{align}
         &\overline{P}_e^{(n)} \leq e^{-4\zeta_n^{(1)}\gamma_n\sqrt{n}}\,,
         \label{eq:DerandomizationErrors_no_key}
         \\
         &\abs{\RelEntropy{\overline{\rho}_{W^n}}{\omega_0^{\otimes n}} - \RelEntropy{\omega_{\alpha_n}^{\otimes n}}{\omega_0^{\otimes n}}} \leq e^{-2\zeta_n^{(2)}\gamma^{\frac{3}{2}}_n n^\frac{1}{4}}\,, 
          \label{eq:Derandom_Covertness_no_key}
        \\
        &
        \frac{1}{\abs{\mathcal{M}}}\sum_{m\in\mathcal{M}}\TrNorm{\rho^{(m)}_{W^n} - \omega_{\alpha_n}^{\otimes n}} 
        \leq  
        e^{-2\zeta_n^{(3)}\gamma_n^{\frac{3}{2}} n^{\frac{1}{4}}}
        \label{eq:Derandomization_Secrecy_no_key}
    \end{align}
\label{Equation:Derandomization_Average_no_key}
\end{subequations}
for sufficiently large $n$, where % 
    $\zeta_n$, % 
    $\zeta_n^{(1)}% 
    $, 
    $\zeta_n^{(2)} % 
    $, and 
    $\zeta_n^{(3)}$
    as in
    Proposition ~\ref{Proposition:classical_covert_and_secret_no_key} above.
\end{proposition}

\begin{proof}[Proof of Proposition~\ref{Proposition:Secrecy_Derandomization_no_key}]
 Consider a random codebook $\mathscr{C}=\{\mathbf{x}^n(m,\ell)\}
 $ as in  Proposition ~\ref{Proposition:classical_covert_and_secret_no_key}, and
define the following probabilistic events,
\begin{align}
    &\mathscr{A}_\text{decoder} = \left\{\overline{P}_e^{(n)} \leq e^{-4\zeta_n^{(1)}\gamma_n\sqrt{n}}\right\} \,,\\
    &\mathscr{A}_\text{covert} =
    \nonumber\\
    & \left\{\abs{\RelEntropy{\overline{\rho}_{W^n}}{\omega_0^{\otimes n}} - \RelEntropy{\omega_{\alpha_n}^{\otimes n}}{\omega_0^{\otimes n}}} \leq e^{-2\zeta_n^{(2)}\gamma^{\frac{3}{2}}_n n^\frac{1}{4}}\right\} \,,\\
    &\mathscr{A}_\text{secrecy} = \left\{
    \frac{1}{\abs{\mathcal{M}}}\sum_{m\in\mathcal{M}}\TrNorm{\rho^{(m)}_{W^n} - \omega_{\alpha_n}^{\otimes n}} \leq 
    e^{-2\zeta_n^{(3)}\gamma_n^{\frac{3}{2}} n^{\frac{1}{4}}}
    \right\} 
    \,.
\end{align}
Using the union of events bound, Markov's inequality, and \eqref{Equation:Random_Code_Bounds_no_key}, we have
\begin{align}
    &\Pr{\mathscr{A}^c_\text{decoder} \cup \mathscr{A}^c_\text{covert} \cup \mathscr{A}^c_\text{secrecy}} 
    \nonumber\\
    &\leq \Pr{\mathscr{A}^c_\text{decoder}} + \Pr{\mathscr{A}^c_\text{covert}} + \Pr{\mathscr{A}^c_\text{secrecy}}
    \nonumber\\
    &\leq e^{-\zeta_n^{(1)}\gamma_n\sqrt{n}} + e^{-2\zeta_n^{(2)}\gamma^{\frac{3}{2}}_n n^\frac{1}{4}} + 
    e^{-\zeta_n^{(3)}\gamma^{\frac{3}{2}}_n n^\frac{1}{4}}
    \label{Equation:deterministic_code_probability_upper_bound_no_key}
\end{align}
which tends to zero as $n\to\infty$.
We deduce that there exists a  codebook that satisfies
the error, covertness and secrecy requirements in \eqref{Equation:Derandomization_Average_no_key}. % 
\end{proof}

\subsubsection{Proof of Lemma~\ref{Lemma:Expergated_classical_code} (Expurgated covert secrecy)}
\label{Subsubsection:Expurgation}
In the derivation of the entanglement-generation capacity theorem in Section~\ref{sec:Proof Covert Entanglement Generation}, we use the classical-quantum secrecy code in order to construct a code for entanglement generation. 
For this purpose, we need a bound on the  maximum rather than message-average error criteria. 
To achieve this,
we use the standard  expurgation argument. % 
    Consider a uniformly distributed {secret} message 
    $\mathbf{m}\sim \mathrm{Unif}(\mathcal{M})$ and % 
    {public message}
    $\mathbf{l}\sim \mathrm{Unif}(\mathcal{L})$, as in Proposition~\ref{Proposition:Secrecy_Derandomization_no_key}.
    On the one hand, % 
    \begin{align}
        &\Pr_{\mathbf{m},\mathbf{l}}\left\{ \left\{{P}_e^{(n)}(\mathbf{m},\mathbf{l})> e^{-\zeta_n^{(1)}\gamma_n\sqrt{n}} \right\} \right.
        \nonumber\\
        &\left.\phantom{====}\bigcup \left\{\TrNorm{\rho^{(\mathbf{m})}_{W^n} - \omega_{\alpha_n}^{\otimes n}} 
        > e^{-\zeta_n^{(3)}\gamma_n^{\frac{3}{2}} n^{\frac{1}{4}}} \right\}\right\} \nonumber \\
        &\stackrel{(a)}{\leq} \frac{\mathbb{E}_{\mathbf{m},\mathbf{l}}\left\{P_e^{(n)}(\mathbf{m},\mathbf{l})\right\}}{e^{-\zeta_n^{(1)}\gamma_n\sqrt{n}}} 
        + 
        \frac{\mathbb{E}_{\mathbf{m}}\left\{\TrNorm{\rho^{(\mathbf{m})}_{W^n} - \omega_{\alpha_n}^{\otimes n}}\right\}}{e^{-\zeta_n^{(3)}\gamma_n^{\frac{3}{2}} n^{\frac{1}{4}}}}
        \nonumber \\
        &= \frac{\overline{P}_e^{(n)}}{e^{-\zeta_n^{(1)}\gamma_n\sqrt{n}}} 
        +
        \frac{\frac{1}{\abs{\mathcal{M}}}\sum_{m\in\mathcal{M}} \TrNorm{\rho^{(m)}_{W^n} - \omega_{\alpha_n}^{\otimes n}}}{e^{-\zeta_n^{(3)}\gamma_n^{\frac{3}{2}} n^{\frac{1}{4}}}} \nonumber \\
        &\stackrel{(b)}{\leq} e^{-3\zeta_n^{(1)}\gamma_n\sqrt{n}} + e^{-\zeta_n^{(3)}\gamma_n^{\frac{3}{2}} n^{\frac{1}{4}}} \nonumber \\
        &\leq e^{-2\zeta_n^{(1)}\gamma_n\sqrt{n}} 
        \label{Equation:uniform_zeta1}
    \end{align}
    for sufficiently large $n$,
    where $(a)$ follows the union bound and % 
    Markov's inequality, % 
    and $(b)$ follows from \eqref{eq:DerandomizationErrors_no_key} and \eqref{eq:Derandomization_Secrecy_no_key}. 

    On the other hand, we can also write
    \begin{align}
        &\Pr_{\mathbf{m},\mathbf{l}}\left\{\left\{{P}_e^{(n)}(\mathbf{m},\mathbf{l}) > e^{-\zeta_n^{(1)}\gamma_n\sqrt{n}} \right\} \right.
        \nonumber\\
        &\left.\phantom{====}\bigcup \left\{\TrNorm{\rho^{(\mathbf{m})}_{W^n} - \omega_{\alpha_n}^{\otimes n}} 
        > e^{-\zeta_n^{(3)}\gamma_n^{\frac{3}{2}} n^{\frac{1}{4}}} \right\}\right\} \nonumber \\
        &=
        \frac{1}{\abs{\mathcal{M}}\abs{\mathcal{L}}}
        \left|\left\{{(m,\ell)\in\mathcal{M}\times\mathcal{L}} \,:\; {P}_e^{(n)}(m,\ell) > e^{-\zeta_n^{(1)}\gamma_n\sqrt{n}}
         \right.\right.
        \nonumber\\
        &\left.\left.\phantom{========}
        \text{ or }
        \TrNorm{\rho^{(m)}_{W^n} - \omega_{\alpha_n}^{\otimes n}} 
        > e^{-\zeta_n^{(3)}\gamma_n^{\frac{3}{2}} n^{\frac{1}{4}}}\right\} \right| % 
        \,.
        \label{Equation:uniform_set_size}
    \end{align}
    Together,  \eqref{Equation:uniform_zeta1} and \eqref{Equation:uniform_set_size} imply
    \begin{align}
    &\left|\left\{{(m,\ell)\in\mathcal{M}\times\mathcal{L}} \,:\; {P}_e^{(n)}(m,\ell)> e^{-\zeta_n^{(1)}\gamma_n\sqrt{n}}
    \right.\right.
        \nonumber\\
        &\left.\left.\phantom{========}
        \text{ or }
        \TrNorm{\rho^{(m)}_{W^n} - \omega_{\alpha_n}^{\otimes n}} 
        > e^{-\zeta_n^{(3)}\gamma_n^{\frac{3}{2}} n^{\frac{1}{4}}}\right\} \right| % 
        \nonumber\\
    &\leq \varepsilon_n \abs{\mathcal{M}}\abs{\mathcal{L}}
    \label{Equation:expurgated_set_size}
    \end{align}
    where we have defined $\varepsilon_n \equiv e^{-2\zeta_n^{(1)}\gamma_n\sqrt{n}}$.

    Intuitively, at most $\varepsilon_n \abs{\mathcal{M}}\abs{\mathcal{L}}$ of the pairs
    $(m,\ell)$
    have a ``bad'' error probability % 
    or leakage distance. % 
        We note that, in order to avoid the removal of the entire message set, 
we first remove the worst fraction $\varepsilon_n$ of the codewords for each $m$, % 
and then update the 
{public message}
indices of the remaining codewords. % 
Next, we ``throw away'' the worst fraction $\varepsilon_n$ of the % 
{secret} messages.
    We are left with a {secret} message set $\mathcal{M}'$ and a 
    {public message set}
    $\mathcal{L}'$ of sizes at least $(1-\varepsilon_n)\abs{\mathcal{M}}$ and $(1-\varepsilon_n)\abs{\mathcal{L}}$, respectively,
    i.e., % 
    \begin{align}
        \abs{\mathcal{M}'} 
        &\geq \left(1 - e^{-2\zeta_n^{(1)}\gamma_n\sqrt{n}}\right)\abs{\mathcal{M}} \,,
        \nonumber\\
        \abs{\mathcal{L}'} 
        &\geq \left(1 - e^{-2\zeta_n^{(1)}\gamma_n\sqrt{n}}\right)\abs{\mathcal{L}}
        \,.
        \label{Equation:Expurgated_Size}
    \end{align}
    Observe 
that the expurgation has a negligible impact on 
    the  {secret} information 
    rate, % 
    as $\log{\abs{\mathcal{M}'}}\geq \gamma_n\sqrt{n}
    \left[(1 -2\zeta_n)\RelEntropy{\sigma_1}{\sigma_0}
    -% 
    (1 + \zeta_n)\RelEntropy{\omega_1}{\omega_0}
    \right]_+$, for sufficiently large $n$.

    It remains to show that covertness criterion still holds after expurgation.
    Denote Willie's average state  with respect to the expurgated code by
    \begin{align}
        \overline{\rho}_{W^n}' &= \frac{1}{\abs{\mathcal{M}'}\abs{\mathcal{L}'}}\sum_{(m,\ell) \in \mathcal{M}'\times \mathcal{L}'}
        \rho_{W^n}^{(m,\ell)}
        \,.
    \end{align}
    Letting $ \mu=\frac{\abs{\mathcal{M}'}}{\abs{\mathcal{M}}}$ and $ \lambda=\frac{\abs{\mathcal{L}'}}{\abs{\mathcal{L}}}$,  % 
    we have
    \begin{align}
        &\TrNorm{\overline{\rho}_{W^n} - \overline{\rho}_{W^n}'} 
        \nonumber\\
        &= % 
        \left\|\frac{1}{\abs{\mathcal{M}}\abs{\mathcal{L}}}\sum_{(m,\ell) \in \mathcal{M}\times \mathcal{L}}{\rho_{W^n}^{(m,\ell)}} \right.
        \nonumber\\
        &\left.\phantom{=========} - \frac{1}{\abs{\mathcal{M}'}\abs{\mathcal{L}'}}\sum_{(m,\ell) \in \mathcal{M}' \times \mathcal{L}'}{\rho_{W^n}^{(m,\ell)}} \right\|_1 % 
        \nonumber \\
        &= \frac{1}{\abs{\mathcal{M}}\abs{\mathcal{L}}}\TrNorm{\sum_{(m,\ell) \in \mathcal{M}\times \mathcal{L}}{\rho_{W^n}^{(m,\ell)}} -\frac{1}{\mu\lambda} \sum_{(m,\ell) \in \mathcal{M}'\times \mathcal{L}'}{\rho_{W^n}^{(m,\ell)}}} \nonumber \\
        &= \frac{1}{\abs{\mathcal{M}}\abs{\mathcal{L}}}
        \left\|\left(1 - \frac{1}{\mu\lambda}\right)\sum_{(m,\ell) \in \mathcal{M}'\times \mathcal{L}'}{\rho_{W^n}^{(m,\ell)}} \right.
        \nonumber\\
        &\left.\phantom{============} + \sum_{(m,\ell) \notin \mathcal{M}'\times \mathcal{L}'}{\rho_{W^n}^{(m,\ell)}} \right\|_1 % 
        \,.
         \end{align}
         By \eqref{Equation:Expurgated_Size},
         \begin{align}
\mu\geq 1-\varepsilon_n \,,\;
\lambda\geq 1-\varepsilon_n
        \,.
         \end{align}
Hence, by the triangle inequality,
 \begin{align}
        &\TrNorm{\overline{\rho}_{W^n} - \overline{\rho}_{W^n}'}
        \nonumber\\
        &\leq 
        \frac{1}{\abs{\mathcal{M}}\abs{\mathcal{L}}}\left(1 - \frac{1}{\mu\lambda}\right)\sum_{(m,\ell) \in \mathcal{M}'\times \mathcal{L}'}\TrNorm{{\rho_{W^n}^{(m,\ell)}}} 
        \nonumber\\
        &\phantom{=} + \frac{1}{\abs{\mathcal{M}}\abs{\mathcal{L}}}\sum_{(m,\ell) \notin \mathcal{M}'\times \mathcal{L}'}\TrNorm{{\rho_{W^n}^{(m,\ell)}}} \nonumber \\
        &=\mu\lambda \left(1 - \frac{1}{\mu\lambda}\right)+1-\mu\lambda
        \nonumber \\
        &\leq  \left(1 - \frac{1}{\mu\lambda}\right)+1-\mu\lambda
        \nonumber \\
        &\leq \frac{2\varepsilon_n - \varepsilon_n^2}{(1 - \varepsilon_n)^2} + 2\varepsilon_n - \varepsilon_n^2 \nonumber \\
        &\leq 2\sqrt{\varepsilon_n}
        \,.
        \label{Equation:expurgated_average_state_bound}
    \end{align}
    Next, we consider covertness with respect to the expurgated code. Observe that
    \begin{align}
        &\abs{\RelEntropy{{\overline{\rho}'}_{W^n}}{\omega_0^{\otimes n}} - \RelEntropy{\overline{\rho}_{W^n}}{\omega_0^{\otimes n}}}
        \nonumber\\
        &= \abs{-H\left({\overline{\rho}'}_{W^n}\right) + H\left(\overline{\rho}_{W^n}\right) + \Tr{(\overline{\rho}_{W^n} - {\overline{\rho}'}_{W^n})\log{\omega_0^{\otimes n}}}} \nonumber \\
        &\leq \abs{-H\left({\overline{\rho}'}_{W^n}\right) + H\left(\overline{\rho}_{W^n}\right)} 
        \nonumber\\
        &\phantom{=} + \TrNorm{\overline{\rho}_{W^n} - {\overline{\rho}'}_{W^n}}\cdot\norm{\log{\omega_0^{\otimes n}}}_{\infty} 
        \label{Equation:expurgated_vs_non_expurgated_covertness_bound}
        \end{align}        
    where the inequality follows from the triangle inequality and Lemma~\ref{Lemma:H\"older inequality}, taking $p=1$ and $q=\infty$. 
Based on entropy continuity  \cite[Lemma 1]{Winter:16p}, % 
the first term is bounded by 
\begin{align}
         \abs{-H\left({\overline{\rho}'}_{W^n}\right) + H\left(\overline{\rho}_{W^n}\right)} 
        &\leq \sqrt{\varepsilon_n} \log{d_W^n} + h_2(\sqrt{\varepsilon_n}) % 
    \end{align}
where $h_2(p) \equiv -p\log(p) - (1-p)\log(1-p)$ is the binary entropy function, which is bounded by 
$h_2(p)\leq 2\sqrt{p}$
(see \cite[Th. 1.2]{topsoe2001bounds}). 
As for the second term on the right-hand side of \eqref{Equation:expurgated_vs_non_expurgated_covertness_bound},
\begin{align}
        &\TrNorm{\overline{\rho}_{W^n} - {\overline{\rho}'}_{W^n}}\cdot\norm{\log{\omega_0^{\otimes n}}}_{\infty} 
        \nonumber\\
        &= \TrNorm{\overline{\rho}_{W^n} - {\overline{\rho}'}_{W^n}}\cdot n\log{\left((\lambda_{\min}(\omega_0))^{-1}\right)} \nonumber \\
        &\leq 2\sqrt{\varepsilon_n} \cdot n\log{\left((\lambda_{\min}(\omega_0))^{-1}\right)} 
\label{Equation:Norm_lambda_min}
\end{align}
     based on the definition of supremum norm and by
    \eqref{Equation:expurgated_average_state_bound}. 
    Furthermore, by the triangle inequality, % 
    \begin{align}
        &\abs{\RelEntropy{{\overline{\rho}'}_{W^n}}{\omega_0^{\otimes n}} - \RelEntropy{\omega_{\alpha_n}^{\otimes n}}{\omega_0^{\otimes n}}} \nonumber \\ 
        &\leq \abs{\RelEntropy{{\overline{\rho}'}_{W^n}}{\omega_0^{\otimes n}} - \RelEntropy{\overline{\rho}_{W^n}}{\omega_0^{\otimes n}}} 
        \nonumber\\
        &\phantom{=} +
        \abs{\RelEntropy{\overline{\rho}_{W^n}}{\omega_0^{\otimes n}} - \RelEntropy{\omega_{\alpha_n}^{\otimes n}}{\omega_0^{\otimes n}}} \nonumber \\
        &\leq 
        \left(\log{d_W} + 2\log{\left((\lambda_{\min}(\omega_0))^{-1}\right)}\right)ne^{-\zeta_n^{(1)}\gamma_n\sqrt{n}} 
        \nonumber\\
        &\phantom{=} + 2e^{-\frac{1}{2}\zeta_n^{(1)}\gamma_n\sqrt{n}}
        + e^{-2\zeta_n^{(2)}\gamma^{\frac{3}{2}}_n n^\frac{1}{4}} 
    \end{align}
    where % 
    the last line follows the bounds in \eqref{Equation:expurgated_vs_non_expurgated_covertness_bound}-\eqref{Equation:Norm_lambda_min} 
        and the bound on the covertness criterion \eqref{eq:Derandom_Covertness_no_key} for the deterministic codebook in Proposition~\ref{Proposition:Secrecy_Derandomization_no_key}.
    Thus, there exists 
    $\tilde{\zeta}_n^{(2)} \in \omega\left((\log n)^{-2}\right)$ such that the covertness requirement \eqref{eq:ExpurgatedCovertness} holds.
\qed

\subsubsection{Rate Analysis}
\label{Subsection:secrecy_rate_analysis_no_key}
Lemma~\ref{Lemma:Expergated_classical_code} proves the existence of a reliable, covert and secret communication scheme such
that \eqref{eq:ExpurgatedCovertness} holds. Therefore,
\begin{align}
    \RelEntropy{\overline{\rho}_{W^n}}{\omega_0^{\otimes n}} 
    \leq \RelEntropy{\omega_{\alpha_n}^{\otimes n}}{\omega_0^{\otimes n}} + e^{-\tilde{\zeta}_n^{(2)}\gamma^{\frac{3}{2}}_n n^\frac{1}{4}}
    \,.
\end{align}
Furthermore, by Lemma~\ref{Lemma:avg_QRE_to_chi_square} (see \cite[Lemma 5]{9344627}), we have
\begin{align}
    D(\omega_{\alpha_n}^{\otimes n}||\omega_0^{\otimes n})
    = \frac{1}{2}\gamma_n^2\chi^2(\omega_1||\omega_0) + 
    {O}\left(\frac{\gamma_n^3}{\sqrt{n}}\right)
\end{align}
where $\alpha_n = \frac{\gamma_n}{\sqrt{n}}$.
Thus, % 
\begin{align}
    &\frac{\log{\abs{\mathcal{M}}}}{\sqrt{n D(\overline{\rho}_{W^n}||\omega_0^{\otimes n})}}
    \nonumber\\
    &\geq % 
    \frac{\gamma_n\sqrt{n}
    \left[(1 -2\zeta_n)\RelEntropy{\sigma_1}{\sigma_0}
    -% 
    (1 + \zeta_n)\RelEntropy{\omega_1}{\omega_0}
    \right]_+}{\sqrt{\frac{1}{2}n\gamma_n^2\chi^2(\omega_1||\omega_0) + ne^{-\tilde{\zeta}_n^{(2)}\gamma^{\frac{3}{2}}_n n^\frac{1}{4}} +  
    {O}\left(\sqrt{n}\gamma_n^3\right)}}
    \,.
\end{align}
In the limit of $n\to\infty$, we achieve the covert secrecy rate
    \begin{align}
    L_\text{S}
    \geq \frac{\left[\RelEntropy{\sigma_1}{\sigma_0}
    -\RelEntropy{\omega_1}{\omega_0}
    \right]_+}{\sqrt{\frac{1}{2}\chi^2(\omega_1||\omega_0)}}
\end{align}
without key assistance. 

\subsection{Converse Proof 
}
\label{Subsection:Converse:Covert_Secrecy_Capacity_No_Key}
Consider covert secrecy without key assistance, i.e., $\log\abs{\mathcal{K}}=0$ , as described in  Section~\ref{Section:Covert-Secret Communication Over Classical-Quantum Channels}.
{Instead, Alice uses her secret message $m \in \mathcal{M}$ and her public message $\ell \in \mathcal{L}$ to generate the input to the channel $\mathbf{x}^n = f_n(\mathbf{m}, \mathbf{l})$. At the channel output, Bob tries to detect the secret message $m$ and the public message $\ell$ with high
probability.
}

We observe that it suffices to prove the converse part with respect to a weaker secrecy requirement:
\begin{align}
I(\mathbf{m};W^n)_\rho \leq \beta_n \log\abs{\mathcal{M}}
\label{Eq:Weak_Secrecy}
\end{align}
where $\beta_n$ tends to zero as 
$n\to \infty$ (see \cite[Sec. 23.1.1]{W2017}).
Thus, by the chain rule,
\begin{subequations}
    \begin{align}
        I(\mathbf{m}\mathbf{l};W^n)_\rho 
        &= I(\mathbf{m};W^n)_\rho + I(\mathbf{l};W^n|\mathbf{m})_\rho
        \label{eq:info_rate_converse_mutual_info_mlW^n_1}
        \\
        &\leq \beta_n \log\abs{\mathcal{M}} + \log\abs{\mathcal{L}} \,.
        \label{eq:info_rate_converse_mutual_info_mlW^n_6}
    \end{align}    \label{eq:info_rate_converse_mutual_info_mlW^n}
\end{subequations}
Therefore,
\begin{align}
    \log{\abs{\mathcal{M}}} 
    &\leq 
    \log\abs{\mathcal{M}} + \beta_n \log\abs{\mathcal{M}} + \log\abs{\mathcal{L}} - I(\mathbf{m}\mathbf{l};W^n)_\rho
    \,.
\end{align}
Now, since we assume a deterministic encoder $f_n$, where $\mathbf{x}^n = f_n(\mathbf{m}, \mathbf{l})$, we have
\begin{align}
I(\mathbf{m}\mathbf{l};W^n)_\rho&=
    I(\mathbf{m}\mathbf{l}\mathbf{x}^n;W^n)_\rho
    \nonumber\\
     &=
I(\mathbf{x}^n;W^n)_\rho+I(\mathbf{m}\mathbf{l};W^n|\mathbf{x}^n)_\rho
\nonumber\\
&    = I(\mathbf{x}^n;W^n)_{\overline{\rho}}
    \,,
\end{align}
and thus, % 
\begin{align}
    \log{\abs{\mathcal{M}}}
    &\leq (1+% 
    {\beta_n})\log\abs{\mathcal{M}}\abs{\mathcal{L}}-I(\mathbf{x}^n;W^n)_{\overline{\rho}}
    \,.
    \label{eq:info_rate_converse_logM_upper_bound}
\end{align}

Next, by Fano’s inequality,
\begin{subequations}
    \begin{align}
        \log\abs{\mathcal{M}}\abs{\mathcal{L}} 
        &\leq I(\mathbf{m}\mathbf{l} ; B^n)_\rho + 1 + \varepsilon_n\log\abs{\mathcal{M}}\abs{\mathcal{L}}
         \\
        &= I(\mathbf{x}^n ; B^n)_\rho + 1 + \varepsilon_n\log\abs{\mathcal{M}}\abs{\mathcal{L}}
        \label{eq:info_rate_converse_logML_bound_5}
    \end{align}
(see \cite[Th. 10.7.3]{W2017}). 
    \label{eq:info_rate_converse_logML_bound}
\end{subequations}
Then, by the same arguments as in 
Subsection~\ref{Appendix:key_converse_logM_upper_bound} (see \eqref{equation:key_converse_logM_upper_bound_third_bound_1}-\eqref{equation:key_converse_logM_upper_bound_final_bound} therein), we obtain
\begin{align}
    \log\abs{\mathcal{M}}\abs{\mathcal{L}} 
    &\leq \frac{1}{1 - \varepsilon_n}\left(n\mu_nD(\sigma_1||\sigma_0) + 1\right)
    \,.
    \label{Equation:converse_info_rate_logML_upper_bound_from_Lemma}
\end{align}
Furthermore, by Lemma~\ref{Lemma:key_converse_I(x^n;W^n)_upper_bound}, we have
    \begin{align}
        I(\mathbf{x}^n;W^n)_{\overline{\rho}} \geq n\mu_n D(\omega_1 \| \omega_0) - 2\delta_n^{\text{cov}}
        \label{Equation:converse_info_rate_mutal_info_lower_bound}
    \end{align}
    as the analysis depends solely on the input-output relation of  the channel.
Therefore, by \eqref{Equation:converse_info_rate_logML_upper_bound_from_Lemma} and 
\eqref{Equation:converse_info_rate_mutal_info_lower_bound}, we can upper bound \eqref{eq:info_rate_converse_logM_upper_bound} further as
\begin{align}
    &\log{\abs{\mathcal{M}}}
    \nonumber\\
    &\leq \frac{1+% 
    {\beta_n}}{1 - \varepsilon_n}\left(n\mu_nD(\sigma_1||\sigma_0) + 1\right)-n\mu_n D(\omega_1 \| \omega_0) + 2\delta_n^{\text{cov}}
\end{align}
for $\mu_n \in o(1)$ such
that $\lim_{n \to \infty} \sqrt{n}\mu_n = 0$.
Thus, dividing  by $\sqrt{n D(\overline{\rho}_{W^n}||\omega_0^{\otimes n})}$
yields the following bound on the covert secret rate,
\begin{align}
 L_\text{S} 
 &=
 \frac{\log{\abs{\mathcal{M}}}}{\sqrt{n D(\overline{\rho}_{W^n}||\omega_0^{\otimes n})}} 
 \nonumber\\
    &\leq 
    \frac{n\mu_n}{\sqrt{n D(\overline{\rho}_{W^n}||\omega_0^{\otimes n})}}  \left[\frac{1+% 
    {\beta_n}}{1 - \varepsilon_n} D(\sigma_1||\sigma_0)-\RelEntropy{\omega_1}{\omega_0} \right.
    \nonumber\\
    &\left.\phantom{===============}
    +\frac{1+% 
    {\beta_n}}{n\mu_n(1 - \varepsilon_n)} + \frac{2\delta_n^{\text{cov}}}{n\mu_n} 
     \right] 
    \,.
    \label{Equation:converse_info_rate_upper_bound}
\end{align}
{
By the same arguments as in 
Subsection~\ref{Appendix:key_converse_I(x^n;W^n)_lower_bound}, we have
\begin{align}
    n D(\overline{\rho}_{W^n}||\omega_0^{\otimes n})
    &\geq
    n^2\RelEntropy{\overline{\rho}_W}{\omega_0} 
    \nonumber\\
    &= n^2\mu_n^2\left(\frac{1}{2} \chi^2(\omega_1||\omega_0) + 
    {O}(\mu_n)\right)
\end{align}
where the equality follows from Lemma~\ref{Lemma:avg_QRE_to_chi_square} (see \eqref{Equation:QRE_tilde_omega_upper_bound}-\eqref{Equation:QRE_tilde_omega_upper_bound_second}). 
Therefore, 
\begin{align}
     \frac{n\mu_n}{\sqrt{n D(\overline{\rho}_{W^n}||\omega_0^{\otimes n})}}
     \leq \frac{1}{\sqrt{\frac{1}{2} \chi^2(\omega_1||\omega_0) + 
     {O}(\mu_n)}}
     \,.
     \label{Equation:converse_info_rate_chi_sqaure_bound}
\end{align}
Recall that we assume $\operatorname{supp}(\sigma_1) \subseteq \operatorname{supp}(\sigma_0)$ in 
Theorem~\ref{Theorem:Covert_Secrecy_Capacity_No_Key}. Hence,
$D(\sigma_1||\sigma_0)$ is finite. As {$\log{\abs{\mathcal{M}}\abs{\mathcal{L}}}$} tends to infinity, the % 
bound in 
\eqref{Equation:converse_info_rate_logML_upper_bound_from_Lemma}
implies
\begin{align}
    \lim_{n \to \infty} n\mu_n = \infty
    \,.
    \label{Equation:converse_info_rate_n_mun}
\end{align}
Thus, the expression on the right-hand side of \eqref{Equation:converse_info_rate_upper_bound} % 
tends to 
\begin{align}
    \frac{D(\sigma_1||\sigma_0)-D(\omega_1||\omega_0)}{\sqrt{\frac{1}{2}\chi^2(\omega_1||\omega_0)}} 
\end{align}
in the limit of ${n \to \infty} $.
This completes the converse proof for the covert secrecy capacity theorem, without key assistance.
 \qed
}

\section{Proof of Lemma~\ref{Lemma:quantum_code_1st_approx}}% 
\label{Appendix:Quantum_Code_First_Approximation}

We prove Lemma~\ref{Lemma:quantum_code_1st_approx} following similar steps as in \cite{Devetak05} \cite[Sec. 24.4]{W2017}.
We show that the actual state  $\ket{\tau}_{RMB^nW^n\widehat{M}\widehat{L}}$ can be approximated by $\ket{\eta}_{RMB^nW^n\widehat{M}\widehat{L}}$ (cf. \eqref{Equation:Bob_encoded_shared_state} and \eqref{Equation:quantum_protocol_state_after_1st_approx_lemma}), using the following consequence of Parseval's relation.
See notations in Table~\ref{tab:eg-notation}.
\begin{lemma}[see 
{\cite[Lemma 4]{Devetak05}}]
    Consider two collections of orthonormal states $\{\ket{\eta_j}\}_{j \in \mathcal{J}}$, and $\{\ket{\tau_j}\}_{j \in \mathcal{J}}$ such that $\braket{\eta_j}{\tau_j} \geq 1 - \varepsilon$ for all $j$. Then, there exist phases $\{h(j)\}$ and $\{t(j)\}$ such that
    \begin{align}
        \braket{\tilde{\eta}}{\tilde{\tau}} \geq 1 - \varepsilon
    \end{align}
    where
    \begin{align}
        \ket{\tilde{\eta}} = \frac{1}{\sqrt{N}}\sum_{j=1}^{N}e^{ih(j)}\ket{\eta_j},
        \quad
        \ket{\tilde{\tau}} = \frac{1}{\sqrt{N}}\sum_{j=1}^{N}e^{it(j)}\ket{\tau_j}
    \end{align}
    \label{Lemma:lemma_A03_Wilde}
\end{lemma}
Now, we proceed to the proof of Lemma~\ref{Lemma:quantum_code_1st_approx} from Section~\ref{sec:Proof Covert Entanglement Generation}.
\begin{proof}[Proof of Lemma~\ref{Lemma:quantum_code_1st_approx}]
    For 
    every $m\in \{0,\ldots, T-1\}$ and $\ell\in\mathcal{L}$,
    we define the following states
    \begin{multline}
        \ket{\tau_{m,\ell}}_{B^nW^n\widehat{M}\widehat{L}} 
        \\
        = \sum_{\substack{m'\in\mathcal{M},\\ \ell'\in\mathcal{L}}}\left(\sqrt{\Lambda^{(m',\ell')}_{B^n}} \otimes \identity_{W^n}\right)\ket{x^n(m,\ell)}_{B^nW^n} 
        \\
        \otimes \ket{m'}_{\widehat{M}} \otimes \ket{\ell'}_{\widehat{L}} 
    \end{multline}
    {and}
    \begin{align}
        &\ket{\eta_{m,\ell}}_{B^nW^n\widehat{M}\widehat{L}} = \ket{x^n(m,\ell)}_{B^nW^n} \otimes \ket{m}_{\widehat{M}} \otimes \ket{\ell}_{\widehat{L}}
    \end{align}
    where $\Lambda^{(m,\ell)}_{B^n}$ is defined in \eqref{Equation:Bob_decoding_success_for_quantum_code}. Then, we have,
    \begin{align}
        &\braket{\eta_{m,\ell}}{\tau_{m,\ell}}
        \nonumber\\
        &= \bra{x^n(m,\ell)}_{B^nW^n} \otimes \bra{m}_{\widehat{M}} \otimes \bra{\ell}_{\widehat{L}}
        \nonumber\\
        &\phantom{=}\times\sum_{\substack{m'\in\mathcal{M},\\ \ell'\in\mathcal{L}}}\left(\sqrt{\Lambda^{(m',\ell')}_{B^n}} \otimes \identity_{W^n}\right)
        \nonumber\\
        &\phantom{=}\times\ket{x^n(m,\ell)}_{B^nW^n} \otimes \ket{m'}_{\widehat{M}} \otimes \ket{\ell'}_{\widehat{L}} \nonumber \\
        &= \bra{x^n(m,\ell)}\left(\sqrt{\Lambda^{(m,\ell)}_{B^n}} \otimes \identity_{W^n}\right)\ket{x^n(m,\ell)}_{B^nW^n} \nonumber \\
        &\geq \bra{x^n(m,\ell)}\left(\Lambda^{(m,\ell)}_{B^n} \otimes \identity_{W^n}\right)\ket{x^n(m,\ell)}_{B^nW^n} \nonumber \\
        &= \Tr{\Lambda^{(m,\ell)}_{B^n}\sigma_{x^n(m,\ell)}} \nonumber \\
        &\geq 1 - e^{-\zeta_n^{(1)}\gamma_n\sqrt{n}}
    \end{align}
    where the first inequality follows from the fact that $0\leq \Lambda^{(m,\ell)}_{B^n}\leq \identity$, % 
    and thus $\sqrt{\Lambda^{(m,\ell)}_{B^n}} \geq \Lambda^{(m,\ell)}_{B^n}$, and the second inequality follows from  \eqref{Equation:Bob_decoding_success_for_quantum_code}. 
    Then, by the auxiliary lemma above, Lemma~\ref{Lemma:lemma_A03_Wilde}, there exist phases $t(m,\ell)$ and $h(m,\ell)$ such that
    \begin{align}
        \bra{\tilde{\eta}_{m}}\ket{\tilde{\tau}_{m}} \geq 1 - e^{-\zeta_n^{(1)}\gamma_n\sqrt{n}}
        \label{Equation:fidelity_of_tilde_states}
    \end{align}
    where we define the states $\tilde{\tau}_{m}$ and $\tilde{\eta}_{m}$ by
    \begin{align}
        &\ket{\tilde{\tau}_{m}}_{B^nW^n\widehat{M}\widehat{L}} 
        = \frac{1}{\sqrt{\abs{\mathcal{L}}}} \sum_{\ell\in\mathcal{L}} e^{it(m,\ell)}\ket{\tau_{m,\ell}}_{B^nW^n\widehat{M}\widehat{L}} 
        \intertext{and}
        &\ket{\tilde{\eta}_{m}}_{B^nW^n\widehat{M}\widehat{L}} = \frac{1}{\sqrt{\abs{\mathcal{L}}}} \sum_{\ell\in\mathcal{L}} e^{ih(m,\ell)}\ket{\eta_{m,\ell}}_{B^nW^n\widehat{M}\widehat{L}}
        \,,
    \end{align}
    for $m\in \{0,\ldots,T-1\}$.
    We now observe that the states that we are interested in, $\ket{\tau}_{RMB^nW^n\widehat{M}\widehat{L}}$ and $\ket{\eta}_{RMB^nW^n\widehat{M}\widehat{L}}$ from  \eqref{Equation:Bob_encoded_shared_state} and \eqref{Equation:quantum_protocol_state_after_1st_approx_lemma}, respectively,  can be written as
    \begin{align}
        \ket{\tau}_{RMB^nW^n\widehat{M}\widehat{L}}
        = \frac{1}{\sqrt{T}}\sum_{m=0}^{T-1}\ket{m}_R\otimes\ket{m}_{M}\otimes\ket{\Tilde{\tau}_{m}}_{B^nW^n\widehat{M}\widehat{L}} 
    \end{align}  
    and
    \begin{align}
        \ket{\eta}_{RMB^nW^n\widehat{M}\widehat{L}}
        = \frac{1}{\sqrt{T}}\sum_{m=0}^{T-1}\ket{m}_R\otimes\ket{m}_{M}\otimes\ket{\Tilde{\eta}_{m}}_{B^nW^n\widehat{M}\widehat{L}}
        .
    \end{align}    
    Thus, their fidelity satisfies
    \begin{align}
        \bra{\tau}\ket{\eta}% 
        &= \left(\frac{1}{\sqrt{T}}\sum_{m=0}^{T-1}\bra{m}_R\otimes\bra{m}_{M}\otimes\bra{\Tilde{\tau}_{m}}_{B^nW^n\widehat{M}\widehat{L}}\right)
        \nonumber\\
        &\phantom{=}\times
        \left(\frac{1}{\sqrt{T}}\sum_{m'=0}^{T-1}\ket{m'}_R\otimes\ket{m'}_{M}\otimes\ket{\Tilde{\eta}_{m'}}_{B^nW^n\widehat{M}\widehat{L}}\right) \nonumber \\
        &= \frac{1}{T}\sum_{m=0}^{T-1}\sum_{m'=0}^{T-1}\bra{m}\ket{m'}\bra{m}\ket{m'}\bra{\Tilde{\tau}_{m}}\ket{\Tilde{\eta}_{m'}} \nonumber \\
        &= \frac{1}{T}\sum_{m=0}^{T-1}\bra{\Tilde{\tau}_{m}}\ket{\Tilde{\eta}_{m}} \nonumber \\
        &\geq 
        1 - e^{-\zeta_n^{(1)}\gamma_n\sqrt{n}}
    \end{align}
    where the inequality follows from  \eqref{Equation:fidelity_of_tilde_states}.
    Therefore, by the Fuchs-van de Graaf Inequalities \cite[Th. 9.3.1]{W2017}, % 
    the trace distance is bounded by
    \begin{align}
        &\TrNorm{\ketbra{\tau}_{RMB^nW^n\widehat{M}\widehat{L}} - \ketbra{\eta}_{RMB^nW^n\widehat{M}\widehat{L}}}
        \nonumber\\
        &\leq 2\sqrt{2}e^{-\frac{1}{2}\zeta_n^{(1)}\gamma_n\sqrt{n}}
        \,.
    \end{align}
This completes the proof of Lemma~\ref{Lemma:quantum_code_1st_approx}.
\end{proof}

\bibliography{References}

\end{document}

%% file: Covert_EG_No_Key_Figure.tex
% \documentclass{paper}[standalone]
% \usepackage[margin=2cm]{geometry}

% \usepackage{tikz}
% \usepackage{amsmath}
% \usetikzlibrary{matrix,positioning}
% \usetikzlibrary{arrows}
% \usetikzlibrary{fit,backgrounds}
% \usetikzlibrary{positioning}

% \usepackage[english]{babel}
% \usepackage{blindtext}
% \usepackage{graphicx}

% \usepackage{algorithm}
% \usepackage{algpseudocode}
% \usepackage{amsmath}
% \usepackage{amssymb}
% \usepackage{setspace}
% \usepackage{xcolor}
% \usepackage{arydshln}
% \usepackage{mathtools}
% \usepackage{cancel}
% \usepackage{physics}
% \usepackage{bbold}

% \begin{document}

\begin{tikzpicture}[scale=1, every node/.style={scale=1},
inner/.style={draw,fill=blue!5,thick,inner sep=3pt,minimum width=8em},
outer/.style={draw=gray,dashed,fill=green!1,thick,inner sep=5pt}]
    % Nodes
    \node (message_state) at (0,0) {$ M$};
    \node (null_state) at (-0.25,-2) {$\ket{0}^{\otimes n}$};
    \node (encoder) at (2,0) [draw, rectangle, minimum width=2cm, minimum height=1cm] {Encoder};
    \node (encoder_output) at (5,0) {};
    \begin{pgfonlayer}{background}
        \node[outer,fit=(message_state) (null_state) (encoder) (encoder_output),label={above left: \textbf{Alice}}] (Alice) {};
    \end{pgfonlayer}
    \node (Channel) at (7,-1) [draw, rectangle, minimum width=2cm, minimum height=3cm, 
    % dashed, 
    fill=gray!20] {$\mathcal{U}_{A\rightarrow BW}^{\otimes n}$};
    \node (Bob_decoder) at (11,0) [draw, rectangle, minimum width=2cm, minimum height=1cm] {Decoder};
    \node (est_message) at (14,0) {${\widehat{M}}$};
    \node (Bob_start) at (9,0) {};
    \node (Bob_end) at (14,0) {};
    \begin{pgfonlayer}{background}
        \node[outer,fit=(Bob_start) (Bob_decoder) (est_message) (Bob_end),label={above right: \textbf{Bob}}] (Bob) {};
    \end{pgfonlayer}
    \node (Willie_detector) at (11,-2) [draw, rectangle, minimum width=2cm, minimum height=1cm] {Detector};
    \node (H0) at (14,-1.8) {$H_0$%: No Tx
    };
    \node (H1) at (14,-2.2) {$H_1$ %: Tx
    };
    \node (Willie_start) at (9,-2) {};
    \begin{pgfonlayer}{background}
        \node[outer,fit=(Willie_start) (Willie_detector) (H0) (H1),label={above right: \textbf{Willie}}] (Willie) {};
    \end{pgfonlayer}
    \begin{pgfonlayer}{background}
        \node (Reference) at (7.1,3) [draw, rectangle, minimum width=13cm, minimum height=1cm,label={\textbf{Alice's Resource}},draw=gray,dashed,fill=purple!5] {};
    \end{pgfonlayer}
    % \node (pre-shared secret) at (7,1.5) [draw, rectangle, minimum width=2cm, minimum height=1cm] {Pre-shared secret};

    % Connections
    \draw[->] (message_state.east) -- (encoder.west);
    \draw[-*] (encoder.east) -- (4,0);
    \draw[-*] (null_state.east) -- (4,-2);
    \draw[<->, bend left, thick] (4.5,-2) to (4.5,0);
    \draw[*->] (3.85,-0.5) -- ++(1,-0.5) -- (Channel.west) node[pos=.67,above] {$A^n$};
    \draw [->] ([yshift=-5mm]Channel.north east) -- (Bob_decoder.west) node[pos=.2,above] {$B^n$};
    \draw[->] (Bob_decoder.east) -- (est_message.west);
    \draw [->] ([yshift=-25mm]Channel.north east) -- (Willie_detector.west) node[pos=.2,above] {$W^n$};
    \draw [->] ([yshift=-3mm]Willie_detector.north east) -- (H0.west);
    \draw [->] ([yshift=-7mm]Willie_detector.north east) -- (H1.west);

    % \draw[-, purple] (message_state.north east) -- ++(0.9,2) -- ++(12,0) -- ++(0.9,-2) -- (est_message.north west);

    \draw[line width=0.4mm, purple] (message_state.west) -- ++(-0.7,1.5) node[left]{$\ket{\Phi}$}  -- ++(0.7,1.5) -- ++(14.6,0) -- ++(0.7,-1.5) node[right]{$\approx\ket{\Phi}$} -- ++(-0.7,-1.5)  ;
    % node[above]{$R$};
    \node (Reference_left) at (0.0,3.2) {$R$};
    \node (Reference_right) at (14.0,3.2) {$R$};

    % \draw[<-, line width=0.2mm, black] (encoder.north) -- ++(0,1.0) -- (pre-shared secret.west);
    % \draw[->, line width=0.2mm, black] (pre-shared secret.east) -- ++(2.65,0) -- (Bob_decoder.north);
    
\end{tikzpicture}

% \end{document}

%% file: Table_I_General.tex
%% ============================================================
%% TABLE I: General Notation — two minipages side-by-side
%% ============================================================
\begin{table*}[tb]
\centering
\caption{General Notation Conventions}
\label{tab:general-notation}
\small

\begin{minipage}[t]{0.48\textwidth}
\vspace{0pt}
\centering
\begin{tabular}{ll}
\hline
\textbf{Symbol} & \textbf{Description} \\
\hline
$\mathcal{X}, \mathcal{Y}, \mathcal{Z}, \ldots$ & Finite sets \\
$\mathbf{x}, \mathbf{y}, \mathbf{z}, \ldots$ & Random variables \\
$x, y, z, \ldots$ & Values of random variables \\
$x^{j}$ & Sequence $(x_1, \ldots, x_j)\in\mathcal{X}^j$ \\
$[i:j]$ & Index set $\{i, i+1, \ldots, j\}$ \\
$[t]_+$ & Positive part of $t$: \\
& $\max(0,t)$ for $t \in \mathbb{R}$ \\
$O(g(n)),\; o(g(n))$ & Big-$O$ and little-$o$ \\
& asymptotic notation \\
$\Omega(g(n)),\; \omega(g(n))$ & Big-$\Omega$ and little-$\omega$ \\
& asymptotic notation \\
$A, B, C, \ldots$ & Quantum systems \\
$\mathcal{H}_A$ & Finite-dimensional Hilbert space \\
$d_A$ or $\mathrm{dim}(\mathcal{H}_A)$ & Dimension of $\mathcal{H}_A$ \\
$\mathscr{L}(\mathcal{H}_A)$ & Set of all operators \\
& $Q:\mathcal{H}_A\to\mathcal{H}_A$ \\
$\mathscr{S}(\mathcal{H}_A)$ & Subset of all density operators \\
$\rho, \sigma$ & Density operators (quantum states) \\
$|\Phi\rangle_{A_1 A_2}$ & Maximally entangled state \\
& on $\mathcal{H}_A^{\otimes 2}$ \\
$\mathsf{F}$ & Quantum Fourier transform unitary \\
% $\mathsf{X}, \mathsf{Z}$ & Heisenberg-Weyl unitaries \\
% $\mathsf{CNOT}$ & Controlled-Not gate \\
\hline
\end{tabular}
\end{minipage}%
\hfill
\begin{minipage}[t]{0.48\textwidth}
\vspace{0pt}
\centering
\begin{tabular}{ll}
\hline
\textbf{Symbol} & \textbf{Description} \\
\hline
$\mathsf{X}, \mathsf{Z}$ & Heisenberg-Weyl unitaries \\
$\mathsf{CNOT}$ & Controlled-Not gate \\
$\Pi$ & Projector \\
$\identity_A$ & Identity operator \\
$\mathrm{id}_A$ & Ideal map \\
$\lambda_{\min}(Q)$ & Minimal eigenvalue of \\
& $Q\in\mathscr{L}(\mathcal{H})$ \\
$\|Q\|_1$ & Trace norm: $\mathrm{Tr}\{|Q|\}$ \\
$\|Q\|_\infty$ & Supremum norm \\
$\tfrac{1}{2}\|\rho - \sigma\|_1$ & Normalized trace distance \\
$F(\rho, \sigma)$ & Fidelity \\
$D(\rho \| \sigma)$ & Quantum relative entropy (QRE) \\
$\chi^2(\rho \| \sigma)$ & Quantum chi-square divergence \\
$H(\rho)$ & Von Neumann entropy \\
$I(A;B)_\rho$ & Quantum mutual information \\
$H(A|B)_\rho$ & Conditional quantum entropy \\
$\mathcal{N}_{A \to B}$ & Quantum channel from $A$ to $B$ \\
$\mathcal{U}_{A \to BW}$ & Stinespring dilation \\
& of the channel \\
$\mathcal{N}^c_{A \to W}$ & Complementary channel \\
& from $A$ to $W$ \\
\hline
\end{tabular}
\end{minipage}

\end{table*}

%% file: Secrecy_Figure_With_Public_Message.tex
% \documentclass{paper}[standalone]
% \usepackage[margin=2cm]{geometry}

% \usepackage{tikz}
% \usepackage{amsmath}
% \usetikzlibrary{matrix,positioning}
% \usetikzlibrary{arrows}
% \usetikzlibrary{fit,backgrounds}
% \usetikzlibrary{positioning}

% \usepackage[english]{babel}
% \usepackage{blindtext}
% \usepackage{graphicx}

% \usepackage{algorithm}
% \usepackage{algpseudocode}
% \usepackage{amsmath}
% \usepackage{amssymb}
% \usepackage{setspace}
% \usepackage{xcolor}
% \usepackage{arydshln}
% \usepackage{mathtools}
% \usepackage{cancel}
% \usepackage{physics}
% \usepackage{bbold}

% \begin{document}

\begin{tikzpicture}[scale=1, every node/.style={scale=1},
inner/.style={draw,fill=blue!5,thick,inner sep=3pt,minimum width=8em},
outer/.style={draw=gray,dashed,fill=green!1,thick,inner sep=5pt}]
    % Nodes
    \node (message_state) at (0,0) {${(m,\ell)}$};
    \node (null_state) at (-0.25,-3) {$0^n$};
    \node (encoder) at (2,0) [draw, rectangle, minimum width=2cm, minimum height=1cm] {Encoder};
    \node (encoder_output) at (5,0) {};
    \begin{pgfonlayer}{background}
        \node[outer,fit=(message_state) (null_state) (encoder) (encoder_output),label={above left: \textbf{Alice}}] (Alice) {};
    \end{pgfonlayer}
    \node (Channel) at (7,-1.5) [draw, rectangle, minimum width=2cm, minimum height=4cm, 
    % dashed, 
    fill=gray!20] {$\mathcal{P}_{X\rightarrow BW}^{\otimes n}$};
    \node (Bob_decoder) at (11,0) [draw, rectangle, minimum width=2cm, minimum height=1cm] {Decoder};
    \node (est_message) at (13.75,0) {${(\widehat{{m}}, \widehat{{\ell}})}$};
    \node (Bob_start) at (9,0) {};
    \node (Bob_end) at (14.25,0) {};
    \begin{pgfonlayer}{background}
        \node[outer,fit=(Bob_start) (Bob_decoder) (est_message)
        (Bob_end)
        ,label={[xshift=42pt] above right: \textbf{Bob}}] (Bob) {};
    \end{pgfonlayer}
    \node (Willie_detector) at (11,-2.4) [draw, rectangle, minimum width=2cm, minimum height=1.8cm] {\begin{tabular}{c}
Detector/  \\
Decoder \\
\end{tabular}%Detector/Decoder
    };
    % \draw[*-] (9,-2) -- (Willie_detector.west);
    \node (H0) at (14,-1.8) {$H_0$%: No Tx
    };
    \node (H1) at (14,-2.2) {$H_1$ %: Tx
    };
    \node (Willie_start) at (9,-2) {};
    % \node (Willie_decoder) at (11,-3.2) [draw, rectangle, minimum width=2cm, minimum height=1cm] {Decoder};
    % \draw[*-] (9,-3.2) -- (Willie_decoder.west);
    \node (Willie_est_message) at (14,-3.1) {$m$};
    % \node (Willie_est_message_red_circle) at (14,-3.1)
    % {$\color{red} \scalebox{2.5}{$\oslash$}$};
    \begin{scope}
        \draw[red, line width=0.8pt] (14,-3.1) circle [radius=0.3];
        \draw[red, line width=0.2pt] (13.8,-3.3) -- (14.2,-2.9);
      \end{scope};
    \begin{pgfonlayer}{background}
        \node[outer,fit=(Willie_start) (Willie_detector) (H0) (H1)
        % (Willie_est_message_red_circle)
        ,
        % (Willie_decoder),
        label={[xshift=22pt] above right: \textbf{Willie}}] (Willie) {};
    \end{pgfonlayer}
    % \begin{pgfonlayer}{background}
    %     \node (Reference) at (7.1,3) [draw, rectangle, minimum width=13cm, minimum height=1cm,label={\textbf{Alice's Resource}},draw=gray,dashed,fill=purple!5] {};
    % \end{pgfonlayer}
    \node (pre-shared secret) at (7,1.5) [draw, rectangle, minimum width=2cm, minimum height=1cm] {Pre-shared secret};

    % Connections
    \draw[->] (message_state.east) -- (encoder.west);
    \draw[-*] (encoder.east) -- (4,0);
    \draw[-*] (null_state.east) -- (4,-3);
    \draw[<->, bend left, thick] (4.5,-3) to (4.5,0);
    \draw[*->] (3.85,-0.5) -- ++(1,-1) -- (Channel.west) node[pos=.67,above] {$x^n$};
    \draw [->] ([yshift=-5mm]Channel.north east) -- (Bob_decoder.west) node[pos=.2,above] {$B^n$};
    \draw[->] (Bob_decoder.east) -- (est_message.west);
    \draw[->] ([yshift=+2mm]Willie_detector.south east) -- (13.6,-3.1); %(Willie_est_message.west);
    % \draw [->] ([yshift=-25mm]Channel.north east) -- (Willie_detector.west) node[pos=.2,above] {$W^n$};
    % \draw[->] ([yshift=+10mm]Channel.south east) -- (9,-2.5) node[pos=0.38,above] {$W^n$} -- ++(0.5,-0.3);
    % \draw[->] ([yshift=+10mm]Channel.south east) -- (9,-2.5) node[pos=0.38,above] {$W^n$};
    \draw [->] ([yshift=11mm]Channel.south east) -- (Willie_detector.west) node[pos=.2,above] {$W^n$};
    \draw [->] ([yshift=-3mm]Willie_detector.north east) -- (H0.west);
    \draw [->] ([yshift=-7mm]Willie_detector.north east) -- (H1.west);

    % \draw[-, purple] (message_state.north east) -- ++(0.9,2) -- ++(12,0) -- ++(0.9,-2) -- (est_message.north west);

    % \draw[line width=0.4mm, purple] (message_state.west) -- ++(-0.7,1.5) node[left]{$\ket{\Phi}$}  -- ++(0.7,1.5) -- ++(14.6,0) -- ++(0.7,-1.5) node[right]{$\approx\ket{\Phi}$} -- ++(-0.7,-1.5)  ;
    % % node[above]{$R$};
    % \node (Reference_left) at (0.0,3.2) {$R$};
    % \node (Reference_right) at (14.0,3.2) {$R$};

    \draw[<-, line width=0.2mm, black] (encoder.north) -- ++(0,1.0) -- (pre-shared secret.west);
    \draw[->, line width=0.2mm, black] (pre-shared secret.east) -- ++(2.65,0) -- (Bob_decoder.north);
    
\end{tikzpicture}

% \end{document}

%% file: Table_II_Secrecy.tex
%% ============================================================
%% TABLE II: Covert Secrecy Model — two minipages side-by-side
%% ============================================================
\begin{table*}[tb]
\centering
\caption{Covert Secrecy Model --- Symbols and Definitions}
\label{tab:secrecy-notation}
\small

\begin{minipage}[t]{0.48\textwidth}
\vspace{0pt}
\centering
\begin{tabular}{ll}
\hline
\textbf{Symbol} & \textbf{Description} \\
\hline
$n$ & Number of channel uses (block length) \\
$\mathcal{P}_{X \to BW}$ & Classical-quantum channel: \\
& $\mathcal{X}\to\mathscr{S}(\mathcal{H}_B\otimes\mathcal{H}_W)$ \\
$A, B, W$ & Quantum systems of Alice, Bob \\
& and Willie \\
$\sigma_0,\, \omega_0$ & Bob and Willie's outputs \\
& for the innocent input $x=0$ \\
$\sigma_1,\, \omega_1$ & Outputs corresponding to \\
& a non-innocent input $x=1$ \\
$\mathcal{M},m$ & Secret message set and \\
& corresponding instance \\
$\mathcal{L}, \ell$ & Public message set and \\
& corresponding instance \\
$\mathcal{K}, k$ & Pre-shared secret key set \\
& and corresponding instance \\
$f$ & Encoding function \\
& $f : \mathcal{M} \times \mathcal{L} \times \mathcal{K} \to \mathcal{X}^n$ \\
$\Lambda_{B^n}$ & Bob's decoding POVM \\
$P_e^{(n)}(m,\ell,k)$ & Conditional probability \\
& of decoding error \\
$\overline{P}_e^{(n)}$ & Average probability of \\
& decoding error \\
$\rho_{W^n}^{(m,\ell,k)}$ & Willie's reduced state conditioned \\
& on $m, \ell, k$ (Alice active) \\
$\rho_{W^n}^{(m)}$ & Willie's reduced state conditioned \\
& on $m$ (Alice active) \\
$\overline{\rho}_{W^n}$ & Willie's average state \\
& (Alice active) \\
$\breve{\rho}_{W^n}$ & Constant state with respect to $m$ \\
$\varepsilon$ & Decoding reliability constraint \\
% $\delta_{\text{cov}}$ & Covertness constraint \\
% $\delta_{\text{sec}}$ & Secrecy constraint \\
\hline
\end{tabular}
\end{minipage}%
\hfill
\begin{minipage}[t]{0.48\textwidth}
\vspace{0pt}
\centering
\begin{tabular}{ll}
\hline
\textbf{Symbol} & \textbf{Description} \\
\hline
$\delta_{\text{cov}}$ & Covertness constraint \\
$\delta_{\text{sec}}$ & Secrecy constraint \\
$L_\text{S}$ & Covert secrecy rate \\
$L_{\text{public}}$ & Public message rate \\
$L_{\text{key}}$ & Key rate \\
$C_\text{S}^{\text{key}}(\mathcal{P})$ & Covert secrecy capacity \\
& (key-assisted) \\
$C_\text{S}(\mathcal{P})$ & Covert secrecy capacity \\
& (unassisted) \\
$\alpha_n$ & Sparse signaling probability \\
$\gamma_n$ & Covert scaling parameter: \\
& $\gamma_n = \sqrt{n}\alpha_n$ \\
$\omega_{\alpha_n}$ & Quantum-secure covert state: \\
& $\omega_{\alpha_n} = (1-\alpha_n)\omega_0+\alpha_n \omega_1$ \\
$\mathscr{C}$ & Classical codebook \\
$\zeta_n$ & Arbitrary bounding parameter \\
& (sets size) \\
$\zeta_n^{(i)}$ & Arbitrary parameters for reliability, \\
& covertness and secrecy bounds \\
$\Tilde{\zeta}_n^{(2)}$ & Covertness arbitrary bounding \\
& parameter for expurgation step \\
$c(m,\ell)$ & Codeword given the \\
& messages $m$ and $\ell$ \\
$\sigma_{\mathbf{c}}, \omega_{\mathbf{c}}$ & Bob and Willie's output state for \\
& a random codeword $\mathbf{c}$ \\
$\mu_n$ & Average probability of transmitting \\
& a non-innocent symbol (converse part) \\
$\beta_n$ & Weak secrecy constraint \\
& (converse part) \\
\hline
\end{tabular}
\end{minipage}

\end{table*}

%% file: Table_III_EG.tex
%% ============================================================
%% TABLE III: Covert Entanglement Generation — two minipages side-by-side
%% ============================================================
\begin{table*}[tb]
\centering
\caption{Covert Entanglement Generation --- Symbols and Definitions}
\label{tab:eg-notation}
\small

\begin{minipage}[t]{0.48\textwidth}
\vspace{0pt}
\centering
\begin{tabular}{ll}
\hline
\textbf{Symbol} & \textbf{Description} \\
\hline
$R$ & Alice's reference system \\
& (resource she keeps) \\
$M$ & Alice's ``quantum message'' \\
& (resource distributed to Bob) \\
$A$ & Channel input \\
$B$ & Bob's output system \\
$\widehat{M}, \widehat{L}$ & Bob's decoded systems \\
$W$ & Willie's output system \\
$\ket{\Phi}_{RM}$ & Maximally entangled state \\
& (prepared locally by Alice) \\
$\mathcal{H}_M$ & ``Quantum message'' space \\
$T=\mathrm{dim}(\mathcal{H}_M)$ & Entanglement dimension \\
$\mathcal{F}_{M \to A^n}$ & Alice's encoding map \\
$\mathcal{D}_{B^n \to \widehat{M}}$ & Bob's decoding map \\
$\tau_{RA^n}$ & State after Alice's encoding \\
$\tau_{RB^nW^n}$ & Joint output state \\
$\tau_{R\widehat{M}}$ & Decoded state \\
$\tau_{W^n}$ & Willie's reduced output state \\
$\sigma_0,\, \omega_0$ & Bob and Willie's outputs \\
& for the innocent input $\ket{0}$ \\
% $\sigma_1,\, \omega_1$ & Outputs corresponding to \\
% & a non-innocent input $\ket{1}$ \\
\hline
\end{tabular}
\end{minipage}%
\hfill
\begin{minipage}[t]{0.48\textwidth}
\vspace{0pt}
\centering
\begin{tabular}{ll}
\hline
\textbf{Symbol} & \textbf{Description} \\
\hline
$\sigma_1,\, \omega_1$ & Outputs corresponding to \\
& a non-innocent input $\ket{1}$ \\
$\ket{0}^{\otimes n}$ & Innocent input (Alice inactive) \\
$\omega_{\alpha_n}$ & Quantum-secure covert state: \\
& $\omega_{\alpha_n} = (1-\alpha_n)\omega_0+\alpha_n \omega_1$ \\
$L_{\mathrm{EG}}$ & Covert entanglement-generation rate \\
$C_{\mathrm{EG}}(\mathcal{N})$ & Covert entanglement-generation capacity \\
$\varepsilon$ & Decoding reliability constraint \\
$\delta$ & Covertness constraint \\
$\zeta_n$ & Arbitrary bounding parameter \\
& (dimension size) \\
$\Tilde{\zeta}_n$ & Arbitrary parameter for reliability \\
& and covertness bounds \\
$\ket{\tau}_{RMB^nW^n\widehat{M}\widehat{L}}$ & Actual decoded shared state \\
& (achievability proof) \\
$\ket{\eta}_{RMB^nW^n\widehat{M}\widehat{L}}$ & Approximated decoded shared state \\
& (achievability proof) \\
$\ket{\mathrm{GHZ}}_{RM\widehat{M}}$ & GHZ state between reference, \\
& Alice, and Bob \\
$\ket{\phi_m}_{A^n}$ & Quantum codeword given index $m$ \vspace{0.1cm} \\ 
\hline
\end{tabular}
\end{minipage}

\end{table*}

%% file: C_EG_example_2_plot_vs_gamma_no_key.tex
\begin{tikzpicture}
    \begin{axis}[
        xlabel={$\gamma$},
        ylabel={$C_\text{EG}(\mathcal{N})$},
        axis lines=middle,        % draw axes through origin
        xmin=0, xmax=1.1,          % adjust to fit your plot
        ymin=0, ymax=24,
        clip=false,  % allow labels outside axis area
        xlabel style={
            at={(axis description cs:0.5,-0.15)},
            anchor=north,
        },
        ylabel style={
        at={(axis description cs:-0.15,0.5)},
        anchor=south,
        rotate=90,
    },
      % label style={font=\LARGE},
        % tick label style={font=\Large},
        xtick={0,0.25,0.5,0.75,1.0}
    ]
    \addplot[domain=0.04:0.999, samples=1000,  thick, blue] {max(0,ln((1-x)/x) * sqrt(2*(1-x)/x))};

    \addplot[only marks, mark=*, mark size=2pt, red] coordinates {(0.5, 0)};

    % \addplot[domain=0.5:0.999, samples=1000, very thick, blue] {0};

    % \node at (axis cs:0,0) [anchor=north east] {$(0,0)$};
    \end{axis}
\end{tikzpicture}